\documentclass[11pt]{amsart}
\usepackage[T1]{fontenc}
\usepackage{amssymb,amsmath, amsthm, amsfonts}
\usepackage{graphicx}
\usepackage{listings}
\usepackage{lstautogobble}
\usepackage{enumerate}
\usepackage[shortlabels]{enumitem}
\usepackage{thmtools}
\usepackage{thm-restate}
\usepackage{amsthm}
\usepackage{verbatim}

\usepackage{mathtools}
\usepackage{physics}
\usepackage{color}
\usepackage{hyperref}
\usepackage[final]{showlabels}
\usepackage[capitalise]{cleveref}
\usepackage[bottom]{footmisc}
\crefformat{equation}{(#2#1#3)}

\usepackage{float}
\restylefloat{table}

\usepackage{mathrsfs}
\setlist{  
  listparindent=\parindent,
  parsep=0pt,
}
\theoremstyle{plain}
\newtheorem{thm}{Theorem}[section]
\newtheorem{prop}[thm]{Proposition}
\newtheorem{lemma}[thm]{Lemma}

\theoremstyle{definition}
\newtheorem{mydef}[thm]{Definition}
\newtheorem{ex}[thm]{Example}

\newtheorem{remark}[thm]{Remark}

\newtheorem{ques}{Question}[section]

\newtheorem*{sum-invol}{Summary of result 1 - $\mathcal{H}_{n}$ is involution}
\newtheorem*{sum-GP_ham}{Summary of result 2 - GP Hamiltonian flows}

\numberwithin{equation}{section} %Equation numbering

\DeclarePairedDelimiter\ipp{\langle}{\rangle}
\DeclarePairedDelimiter{\brak}{\lbrack}{\rbrack}
\DeclarePairedDelimiter{\paren}{\lparen}{\rparen}

\newcommand{\M}{{\mathcal{M}}}

\newcommand{\p}{{\partial}}

\newcommand{\R}{{\mathbb{R}}}
\newcommand{\C}{{\mathbb{C}}}
\newcommand{\N}{{\mathbb{N}}}

\newcommand{\K}{{\mathbb{K}}}
\newcommand{\Ss}{{\mathbb{S}}}
\renewcommand{\H}{{\mathcal{H}}}

\newcommand{\g}{{\mathfrak{g}}}
\newcommand{\G}{{\mathfrak{G}}}
\newcommand{\h}{{\mathfrak{h}}}

\renewcommand{\L}{{\mathcal{L}}}
\newcommand{\Sc}{{\mathcal{S}}}
\newcommand{\A}{{\mathcal{A}}}
\renewcommand{\M}{{\mathcal{M}}}
\newcommand{\D}{\mathcal{D}}

\newcommand{\tl}{\tilde}

\newcommand{\ol}{\overline}
\newcommand{\ul}{\underline}
\newcommand{\uell}{\underline{\ell}}
\newcommand{\ux}{\underline{x}}
\newcommand{\uj}{\underline{j}}
\newcommand{\um}{\underline{m}}
\newcommand{\un}{\underline{n}}
\newcommand{\up}{\underline{p}}
\newcommand{\uq}{\underline{q}}

\newcommand{\ep}{\epsilon}

\newcommand{\W}{{\mathbf{W}}}

\DeclareMathOperator{\Sym}{Sym}

\def\XXint#1#2#3{{\setbox0=\hbox{$#1{#2#3}{\int}$ }
\vcenter{\hbox{$#2#3$ }}\kern-.6\wd0}}

\let\oldtocsection=\tocsection
 
\let\oldtocsubsection=\tocsubsection
 
\let\oldtocsubsubsection=\tocsubsubsection
 
\renewcommand{\tocsection}[2]{\hspace{0em}\oldtocsection{#1}{#2}}
\renewcommand{\tocsubsection}[2]{\hspace{1em}\oldtocsubsection{#1}{#2}}
\renewcommand{\tocsubsubsection}[2]{\hspace{2em}\oldtocsubsubsection{#1}{#2}}

\title[Derivation of Hamiltonian structure for the NLS]{A Rigorous Derivation of the Hamiltonian structure for the nonlinear Schr\"{o}dinger equation}

\author[D. Mendelson]{Dana Mendelson$^1$} 
\address{$^1$  
Department of Mathematics \\ 
University of Chicago\\  
5734 S. University Avenue \\ 
Chicago, IL  60637}
\email{dana@math.uchicago.edu}
\thanks{\hspace{-4.6mm} $^1$ D.M. is funded in part by NSF DMS-1800697.}

\author[A. Nahmod]{Andrea R. Nahmod$^2$}
\address{$^2$ 
Department of Mathematics \\ University of Massachusetts\\  710 N. Pleasant Street, Amherst MA 01003}
\email{nahmod@math.umass.edu}
\thanks{$^2$ A.N. is partially supported by NSF-DMS-1463714 and NSF-DMS-1800852.}

\author[N. Pavlovi\'{c}]{Nata\v{s}a Pavlovi\'c$^3$}
\address{$^3$  
Department of Mathematics\\ 
University of Texas at Austin\\ 
2515 Speedway, Stop C1200\\
Austin, TX 78712}
\email{natasa@math.utexas.edu}
\thanks{$^3$ N.P. is funded in part by NSF DMS-1516228 and DMS-1840314.}

\author[M. Rosenzweig]{Matthew Rosenzweig$^4$}
\address{$^4$  
Department of Mathematics\\ 
University of Texas at Austin\\ 
2515 Speedway, Stop C1200\\
Austin, TX 78712}
\email{rosenzweig.matthew@math.utexas.edu}
\thanks{$^4$ M.R. is funded in part by NSF DMS-1516228 and a Provost Excellence Graduate Fellowship from the University of Texas at Austin.}

\author[G. Staffilani]{Gigliola Staffilani$^5$}
\address{$^5$ Department of Mathematics\\
Massachusetts Institute of Technology\\ 
77 Massachusetts Avenue,  Cambridge, MA 02139}
\email{gigliola@math.mit.edu}
\thanks{$^5$ G.S. is funded in part by NSF DMS-1462401 and DMS-1764403, and the Simons Foundation.}

\begin{document}
\maketitle{}

\begin{abstract}
We consider the cubic nonlinear Schr\"odinger equation (NLS) in any spatial dimension, which is a well-known example of an infinite-dimensional Hamiltonian system. Inspired by the knowledge that the NLS is an effective equation for a system of interacting bosons as the particle number tends to infinity, we provide a derivation of the Hamiltonian structure, which is comprised of both a Hamiltonian functional and a weak symplectic structure, for the nonlinear Schr\"odinger equation from quantum many-body systems. Our geometric constructions are based on a quantized version of the Poisson structure introduced by Marsden, Morrison and Weinstein \cite{MMW1984} for a system describing the evolution of finitely many indistinguishable classical particles.
\end{abstract}

\tableofcontents

\section{Introduction}

Hamiltonian partial differential equations (PDEs) are a ubiquitous class of equations which arise as models of physical systems exhibiting at least one, and often several, conservation laws. While the framework of finite-dimensional Hamiltonian systems was initially introduced to formalize Newtonian mechanics, infinite-dimensional Hamiltonian systems have since become a vast area of study, comprising an important class of models in diverse areas such as fluid mechanics, plasma physics, and quantum many-body systems. Establishing a comprehensive mathematical theory of infinite-dimensional Hamiltonian systems which is rich enough to accommodate all the physical problems of interest seems beyond reach; however, one can make mathematically rigorous sense of infinite-dimensional Hamiltonian systems in many interesting cases, see for instance \cite{CM74} and \cite{AM78}.

The focus of the present work will be a particular example of an infinite-dimensional Hamiltonian PDE, namely, the cubic nonlinear Schr\"{o}dinger equation (NLS):
\begin{equation}\label{nls}
i\p_{t}\phi + \Delta \phi = 2\kappa|\phi|^{2}\phi, \qquad \phi : \mathbb{R}^d \to \mathbb{C}, \quad \kappa \in \{\pm 1\}.
\end{equation}
We will recall the precise Hamiltonian formulation of \eqref{nls} in \eqref{nls_hamil_flow} and \eqref{eq:H_NLS} below.  

Over recent years, many authors have sought how to understand the manner in which the dynamics of the NLS arise as an \emph{effective equation}. By effective equation, we mean that solutions of the NLS equation approximate solutions to an underlying physical equation in some topology in a particular asymptotic regime. For example, the NLS is an effective equation for a system of $N$ bosons interacting pairwise via a delta or approximate delta potential, in the sense that the 1-particle density matrix formed by a solution to the NLS is close to the 1-particle reduced density matrix of the system in trace norm, with error tending to zero as the number of particles tends to infinity. Alternatively, the NLS also arises as an effective equation for water waves, where the multiple scales expansion constructed by solving the NLS approximates slowly modulated wave packet solutions to the water waves problem in Sobolev norm, with error tending to zero as the steepness of the wave packets tends to zero.

In contrast to the vast amounts of activity on the derivation of the dynamics of the NLS, to the best of our knowledge, questions about the origins of the Hamiltonian structure of the NLS have remained unexplored. Indeed, continuing with our two examples from the previous paragraph, the $N$-body Schr\"{o}dinger problem is well-known to admit a description as an infinite-dimensional Hamiltonian system, as are the water waves equations \cite{Zakharov1968}, but we are unaware of work which mathematically demonstrates whether, and if so the manner in which, the Hamiltonian structure of the NLS can be interpreted as a limit of the Hamiltonian structure of the $N$-body Schr\"{o}dinger or water waves problems. 

The Hamiltonian formulation for the NLS has two components: the Hamiltonian functional itself and an underlying phase space geometry provided by a weak Poisson manifold.\footnote{We refer to \cref{def:WP} and \cref{weak_sym} for definitions of a weak Poisson and weak symplectic manifold, respectively.} More precisely, to give the Hamiltonian formulation of the NLS, we endow the $d$-dimensional Schwartz space $\Sc(\R^d)$ with the standard weak symplectic structure
\begin{equation}\label{equ:poisson1}
\omega_{L^2}(f,g) = 2\Im{\int_{\R^d}dx\ol{f(x)}g(x)}, \quad \forall f,g \in \Sc(\R^d).
\end{equation}
Letting $\grad_s$ denote the symplectic $L^2$ gradient, see \cref{schwartz_deriv}, the symplectic form $\omega_{L^2}$ induces the canonical Poisson structure 
\begin{align}\label{pb_intro}
\pb{F}{G}_{L^{2}}(\cdot) \coloneqq \omega_{L^{2}}(\grad_{s}F(\cdot) ,\grad_{s}G(\cdot)),
\end{align}
defined for $F, G$ belonging to a certain sub-algebra $A_{\Sc}\subset C^{\infty}(\Sc(\R^{d});\R)$, the precise description of which we postpone to \cref{schwartz_wpoiss}.  The solution of the NLS \eqref{nls} is then the flow associated to a Hamiltonian equation of motion on the infinite-dimensional weak Poisson manifold $(\Sc(\R^{k}),\A_{\Sc}, \pb{\cdot}{\cdot}_{L^{2}})$. More precisely, \eqref{nls} is equivalent to
\begin{equation}\label{nls_hamil_flow}
\paren*{\frac{d}{dt}\phi}(t) = \grad_{s}\H_{NLS}(\phi(t)),
\end{equation}
where
\begin{equation}
\label{eq:H_NLS}
\H_{NLS}(\phi(t)) \coloneqq \int_{\R^d}dx \paren*{|\nabla \phi(t,x)|^2 + \kappa |\phi(t,x)|^4}.
\end{equation}
The goal of the current work is to derive both the weak Poisson structure and Hamiltonian functional constituting the Hamiltonian formulation of the NLS. Providing a rigorous definition and derivation of the geometry will pose the bulk of the difficulty in this work.

The methods we adopt are guided by the extensive research activity in recent years on the derivation of the NLS from the dynamics of interacting bosons. There are a number of different approaches to this derivation problem, but the one which informs our strategy involves the so-called BBGKY hierarchy,\footnote{Bogoliubov--Born--Green--Kirkwood--Yvon hierarchy.} which is a coupled system of linear equations describing the evolution of a system of finitely many interacting bosons, see \eqref{eq:BBGKY} below. This approach was pioneered by Spohn \cite{Spohn80} in the quantum context of the derivation of the Hartree equation in the mean field scaling regime.\footnote{See also the influential works of Lanford \cite{Lanford75, Lanford76} on the derivation of the Boltzmann equation.} We mention the works of Adami, Bardos, Golse, and Teta and Adami, Golse, and Teta \cite{ABGT2004, AGT2007}, who provided a derivation of the one-dimensional cubic NLS  via the BBGKY approach in an intermediate scaling regime between the mean field and Gross-Pitaevskii regimes. We also mention in particular the works of Erd\"{o}s, Schlein, and Yau \cite{ESY2006, ESY2007, ESY2010}, who provided the first rigorous derivation of the three-dimensional cubic NLS in the Gross-Pitaevskii scaling regime via the BBGKY hierarchy, resolving what was a significant open problem, and the work of Klainerman and Machedon \cite{KM08}, who incorporated techniques from dispersive equations to the study of this problem. There is by now an extensive body of work, spanning many years, on deriving the dynamics of the NLS from many-body quantum systems. A thorough account of this history would take us too far afield from our current goals, and consequently we are not mentioning many important contributions in our very brief account. We instead refer the reader to \cite{schlein_clay} for a general survey and more extensive review on the history of the derivation problem and to the more recent lecture notes \cite{Rougerie2015}.

To appreciate some of the difficulties involved in our pursuit, it is important to note that while the dynamics of a system of $N$-bosons is described by the linear Schr\"odinger evolution of a wave function, such an equation is not amenable to taking the infinite-particle limit directly since the wave functions for different particle numbers do not live in a common topological space. Consequently, in order to take an infinite-particle limit, one performs a non-linear transformation of the $N$-body wave functions and considers sequences of $k$-particle marginal density matrices whose evolution is governed by the BBGKY hierarchy.  In particular, there is no clear link between the evolution of the $N$-particle wave function and the NLS each as Hamiltonian dynamical systems. To complicate matters further,  the BBGKY hierarchy is no longer an evidently Hamiltonian flow. 

At the cost of the added complication of working with the BBGKY hierarchy, the aforementioned works on the derivation of the one-particle dynamics actually yield the following stronger result: the full dynamics of the interacting boson system governed by the BBGKY hierarchy converges to dynamics described by the cubic Gross-Pitaevskii (GP) hierarchy, which is an infinite coupled system of partial differential equations for kernels\footnote{In this work, we follow the widespread convention of using the same notation for both the kernel and the operator.} $(\gamma^{(k)})_{k=1}^{\infty}$ of $k$-particle density matrices, defined in \eqref{eq:GP} below.  The connection to the NLS is then as follows: the GP hierarchy admits a special class of factorized solutions given by
\begin{equation}
\label{eq:fac_map}
\gamma^{(k)} \coloneqq \ket*{\phi^{\otimes k}}\bra*{\phi^{\otimes k}}, \qquad k\in\N,
\end{equation}
where $\phi: I\times \R^d \rightarrow\C$ solves \eqref{nls}.  

One might conjecture that the BBGKY and GP hierarchies provide the required link to understand the derivation of the geometry associated to the Hamiltonian formulation of \eqref{nls}. In particular, it is natural to wonder whether the BBGKY and GP hierarchies are Hamiltonian evolution equations posed on underlying weak Poisson manifolds of density matrices,\footnote{We will in fact work on a Poisson manifold of density matrix \emph{hierarchies}.} and whether the Poisson structure for the infinite-particle setting arises in the infinite-particle limit from the Poisson structure for the $N$-body problem. To summarize, one can pose the following questions:

\begin{ques}\label{q:GP_ham}
Can we connect the Hamiltonian structure of the many-body system with that of the infinite-particle system in the following sense: can the GP hierarchy be realized as a Hamiltonian equation of motion with associated functional $\H_{GP}$ on some weak Poisson manifold? Can the Poisson structure and Hamiltonian functional for the GP hierarchy be derived in a suitable sense from a Poisson structure and Hamiltonian functional at $N$-particle level?
\end{ques}

In the current work, we answer these questions affirmatively and establish, for the first time, a Hamiltonian formulation for the BBGKY and GP hierarchies, see \cref{thm:BBGKY_ham} and \cref{thm:GP_ham} below, and a link between the underlying weak Poisson geometry and Hamiltonian functionals in the finite- and infinite-particle settings, see \cref{prop:LB_lim}. 

Our geometric constructions will rely on a special type of weak Poisson structure, namely a Lie-Poisson structure, on a space of density matrix $\infty$-hierarchies, see \cref{sym_pois} below. These constructions are motivated by the work of Marsden, Morrison, and Weinstein \cite{MMW1984} on the Hamiltonian structure of the classical BBGKY hierarchy, which relates to the earlier works on the Hamiltonian structure for plasma systems discovered in Morrison and Green \cite{MG80}, Morrison \cite{Morrison80}, Marsden and Weinstein \cite{MW82}, Spencer and Kaufman \cite{SK82}, and Spencer \cite{Spencer82}. We refer to \cite{MWRSS83} for more discussion on the Hamiltonian formulation of equations of motion for systems arising in plasma physics. Our geometric perspective for the $N$-body Schr\"odinger equation is inspired by taking a ``quantized'' version of the work of \cite{MMW1984}. By adapting their work to the quantum setting, we obtain the formulae for the Poisson structure for the (quantum) BBGKY hierarchy. Taking the infinite-particle limit, which was not considered in \cite{MMW1984}, we obtain the formula for the Poisson structure we use in the infinite-particle setting. We expect that our proofs can serve as a blueprint for deriving the Hamiltonian structure of more general infinite-particle equations arising from systems of interacting classical and quantum particles.

\medskip
Returning to the setting of the NLS, the fact that the GP hierarchy admits the factorized solutions given by \eqref{eq:fac_map} tells us that the dynamics of the NLS are embedded in those of the GP hierarchy. Given that the NLS is a Hamiltonian system and, with our affirmative answer to \cref{q:GP_ham}, so is the GP hierarchy, one might ask if there exists an embedding of the Hamiltonian structure such that the pullback of this embedding yields the NLS Hamiltonian and phase space geometry from that of the GP. In other words, one can pose the following question:

\begin{ques}\label{q:conn_1bdy}
Given our affirmative answer to the previous question, is there then a natural way to connect the Hamiltonian formulation of the GP hierarchy with the Hamiltonian formulation of the NLS in such a manner so as to respect the geometric structure?
\end{ques}

We provide an affirmative answer to this second question by showing, in \cref{thm:pomo} below, that the natural embedding map taking one-particle functions to factorized density matrices described in \eqref{eq:fac_map} is a Poisson morphism between the weak symplectic manifold constituting the NLS phase space and the weak Poisson manifold\footnote{We refer to \cref{sec:pre} for definitions of Poisson morphism and weak Poisson manifold.} constituting the GP phase space. Moreover, the NLS Hamiltonian, see \eqref{eq:H_NLS} below, is just the pullback of the GP Hamiltonian under this embedding, see \eqref{eq:H_GP} below. In summary, the factorization embedding pulls back the GP Hamiltonian structure to that of the NLS.

\medskip
We claim that our work provides a new perspective on what it means to ``derive'' an equation from an underlying physical problem. Indeed, to justify this assertion, we highlight some parallels between our results and the aforementioned works of Erd\"{o}s et al.\ on the derivation of solutions to the NLS equation from the $N$-body problem. In \cite{ESY2006, ESY2007, ESY2010}, solutions to the BBGKY hierarchy with factorized or asymptotically factorized initial data are shown to converge to solutions of the GP hierarchy as the number of particles tends to infinity. The authors then show that solutions to the GP hierarchy in a certain Sobolev-type space are unique.\footnote{A new proof of this uniqueness result was later given by Chen et al. in \cite{CHPS2015}.} Thus, the solution to the NLS equation provides the unique solution to the GP hierarchy starting from factorized initial data, thereby providing a rigorous derivation of the dynamics of the NLS from \eqref{eq:LSchr}. In the current work, we establish the existence of both the underlying Lie algebra and Poisson structure associated to a Hamiltonian formulation of the BBGKY hierarchy and prove that in the infinite-particle limit, these converge to a (previously unobserved) Hamiltonian structure for the GP hierarchy. Moreover, the BBGKY Hamiltonian, defined in \eqref{hbbgky}, converges to the GP Hamiltonian. Finally, we demonstrate that the Hamiltonian functional and phase space of the NLS can be obtained via the pullback of the canonical embedding \eqref{eq:can_emb}, thereby providing a derivation of the Hamiltonian structure of the NLS. 

\begin{remark}
We note that our work does not address any derivation of the \emph{dynamics} of the nonlinear Schr\"odinger equation from many-body quantum systems in the vein of the aforementioned works by Erd\"{o}s et al.\ \cite{ESY2006, ESY2007, ESY2010}. Our current work is complementary to those in the sense that it addresses geometric aspects of the connection of the NLS with quantum many-body systems, answering questions which are of a different nature than those about the dynamics.
\end{remark}

\begin{remark}
We view this work as part of a broader program of understanding how qualitative properties of PDE arise from underlying physical problems. We also mention the works of Lewin, Nam, and Rougerie \cite{LNR2015} and Fr\"{o}hlich, Knowles, Schlein, and Sohinger \cite{FKSS2017}, which derive invariant Gibbs measures for the NLS from many-body quantum systems, as we believe they are related in spirit to this program.
\end{remark}

We conclude by mentioning an application of our current work. In the one-dimensional cubic case, for which the corresponding one-dimensional cubic  nonlinear Schr\"odinger equation is known to be integrable, we establish in a companion work \cite{MNPRS2_2019} that there exists an infinite sequence of Poisson commuting functionals, which we call energies. The Hamiltonian flow associated to the third energy yields the GP hierarchy, and the corresponding flows for the sequence of energies yield a ``hierarchy of infinite-particle hierarchies'' which generalizes the Schr\"odinger hierarchy of Palais \cite{Palais1997}.

\medskip
In the next section,  \cref{sec:statements}, we will record the precise statements of our main results, which require some additional notation and background. We postpone a subsection on the organization of our paper until the end of this next section.

\section{Statements of main results and blueprint of proofs}\label{sec:statements}

We will now state precisely and outline the proofs of our three main results: \cref{thm:BBGKY_ham}, \cref{thm:GP_ham}, and \cref{thm:pomo}. The first two results provide the affirmative answer to \cref{q:GP_ham}, establishing the BBGKY hierarchy and GP hierarchy, respectively, as Hamiltonian flows. \cref{thm:pomo} provides the link between the Hamiltonian structure for the GP hierarchy and the Hamiltonian structure for the nonlinear Schr\"odinger equation, answering \cref{q:conn_1bdy}.

We recall the $N$-body Schr\"odinger equation, BBGKY hierarchy, and limiting GP hierarchy to set the stage for our discussion of the geometry below. It will be useful going forward to fix the following notation: for $d\geq 1$, we denote the point $(x_1, \ldots, x_N) \in  \R^{dN}$ by $\ul{x}_{N}$. We let $\Sc_{s}(\R^{dN})$ be the subspace of $\Sc(\R^{dN})$ of Schwartz functions which are symmetric in their arguments, that is, for any $\pi\in\Ss_N$\footnote{$\Ss_N$ is the symmetric group of order $N$.} we have
\begin{equation}
\Phi(x_{\pi(1)},\ldots,x_{\pi(N)}) = \Phi(x_{1},\ldots,x_{N}), \qquad \ul{x}_{N} \in \R^{dN}.
\end{equation}
We call $\Sc_{s}(\R^{dN})$ the bosonic Schwartz space, see \cref{sym_schwartz} for more details. 

Consider the $N$-body Schr\"{o}dinger equation
\begin{equation}\label{eq:LSchr}
i\p_{t}\Phi_{N} = H_{N}\Phi_{N}, \qquad \Phi_N \in \Sc_{s}(\R^{dN})
\end{equation}
where $H_N$ is the $N$-body Hamiltonian
\begin{equation}
\label{eq:N_bod_ham}
H_{N} \coloneqq \sum_{j=1}^{N}(-\Delta_{x_{j}}) + \frac{2\kappa}{N-1}\sum_{1\leq i<j\leq N} V_{N}(X_{i}-X_{j}), \qquad \kappa\in\{\pm 1\}.
\end{equation}
The pair interaction potential has the form $V_{N} = N^{d\beta}V(N^{d\beta}\cdot)$, where $\beta\in (0,1)$, $V$ is an even nonnegative function in $C_{c}^{\infty}(\R^d)$ with $\int_{\R}dx V(x)=1$, and $V_{N}(X_{i}-X_{j})$ denotes the operator which is multiplication by $V_{N}(x_{i}-x_{j})$.

The $N$-body density matrix, associated to the wave function $\Phi_N \in  \Sc_{s}(\R^{dN})$ is given by
\[
\Psi_N\coloneqq\ket*{\Phi_{N}}\bra*{\Phi_{N}} \in \L(\Sc_s'(\R^{dN}), \Sc_s(\R^{dN})),\footnote{$ \L(\Sc_s'(\R^{dN}), \Sc_s(\R^{dN}))$ denotes the space of continuous linear maps from symmetric tempered distributions to symmetric Schwartz functions.}
\]
and the reduced density matrix hierarchy 
\[
(\gamma_N^{(k)})_{k=1}^N \coloneqq (\Tr_{k+1,\ldots,N}(\Psi_{N}))_{k=1}^{N}
\]
 solves the \emph{quantum BBGKY hierarchy}
\begin{equation}
\label{eq:BBGKY}
\begin{split}
i\p_t\gamma_N^{(k)} &= \comm{-\Delta_{\ux_k}}{\gamma_N^{(k)}} + \frac{2\kappa}{N-1}\sum_{1\leq i<j\leq k} \comm{V_N(X_i-X_j)}{\gamma_N^{(k)}} \\
&\phantom{=} + \frac{2\kappa(N-k)}{N-1}\sum_{i=1}^k\Tr_{k+1}\paren*{\comm{V_N(X_i-X_{k+1})}{\gamma_N^{(k+1)}}}, \qquad {1\leq k\leq N-1} \\
&=\comm{-\Delta_{\ux_k}}{\gamma_N^{(k)}} + \frac{2\kappa}{N-1}\sum_{1\leq i<j\leq k} \comm{V_N(X_i-X_j)}{\gamma_N^{(k)}}, \qquad k=N,
\end{split}
\end{equation}
where we have introduced the notation $\Delta_{\ux_k}\coloneqq \sum_{j=1}^k \Delta_{x_j}$.

The GP hierarchy is formally obtained from the BBGKY hierarchy \eqref{eq:BBGKY} by letting $N\rightarrow\infty$. More precisely, a time-dependent family of density matrix $\infty$-hierarchies $\Gamma(t)=(\gamma(t)^{(k)})_{k=1}^{\infty}$ solves the GP hierarchy if
\begin{equation}
\label{eq:GP}
i\p_{t}\gamma^{(k)} = -\comm{\Delta_{\ux_{k}}}{\gamma^{(k)}}+ 2\kappa B_{k+1}\gamma^{(k+1)}, \qquad \forall k\in\N
\end{equation}
with $\kappa\in\{\pm 1\}$ and
\begin{equation}
B_{k+1}\gamma^{(k+1)} \coloneqq \sum_{j=1}^{k} \paren*{B_{j;k+1}^{+}-B_{j;k+1}^{-}}\gamma^{(k+1)},
\end{equation}
where
\begin{equation}
\paren*{B_{j;k+1}^{+}\gamma^{(k+1)}}(t,\ux_{k};\ux_{k}') \coloneqq \int_{\R^{2d}}dx_{k+1}dx_{k+1}'\delta(x_{k+1}-x_{k+1}')\delta(x_{j}-x_{k+1})\gamma^{(k+1)}(t,\ux_{k+1};\ux_{k+1}')
\end{equation}
with an analogous definition for $B_{j;k+1}^{-}$ with $\delta(x_{j}-x_{k+1})$ replaced by $\delta(x_{j}'-x_{k+1})$. When $\kappa=1$, we say that the hierarchy is \emph{defocusing} and for $\kappa=-1$, we say that the hierarchy is \emph{focusing} (in analogy with the defocusing and focusing NLS, respectively).

\medskip
As we outlined in the introduction, our first main results establish that the BBGKY hierarchy \eqref{eq:BBGKY} and the GP hierarchy \eqref{eq:GP} are Hamiltonian flows on appropriate weak Lie-Poisson manifolds. To do this, we need to define a suitable phase space for the Hamiltonian evolution in both the finite- and infinite-particle settings. In particular, we need to construct certain Lie-Poisson manifolds of density matrix hierarchies, and we outline this construction in the next subsection. We will also establish that the procedure described above for obtaining the BBGKY hierarchy from the $N$-body Schr\"odinger equation can be given by the composition of several natural Poisson maps, thereby establishing the existence of a natural Poisson morphism which maps the $N$-body Schr\"odinger equation to the BBGKY hierarchy.
 
\subsection{Construction of the Lie algebra $\G_{N}$ and Lie-Poisson manifold $\G_{N}^{*}$}\label{sym_n_pois}
For each $k\in\N$, we let
\[
\g_{k} \coloneqq \{ A^{(k)} \in \L(\Sc_s(\R^k),\Sc_s(\R^k)) : (A^{(k)})^* = -A^{(k)} \},
\]
endowed with the subspace topology of $\L(\Sc_s(\R^k),\Sc_s'(\R^k))$. We define a Lie algebra $(\g_{k},\comm{\cdot}{\cdot}_{\g_{k}})$, with Lie bracket defined by 
\begin{equation}
\comm{A^{(k)}}{B^{(k)}}_{\g_{k}} \coloneqq k\comm{A^{(k)}}{B^{(k)}},
\end{equation}
where the right-hand side denotes the usual commutator bracket.
We refer to elements of $\g_k$  as \emph{$k$-particle bosonic observables}. For $N\in\N$, we then define the locally convex direct sum 
\begin{align} \label{equ:GN}
\G_{N} \coloneqq \bigoplus_{k=1}^{N}\g_{k},
\end{align}
and we refer to elements of $\G_N$ as \emph{observable $N$-hierarchies}.

To define a Lie bracket on the space $\G_N$, we consider the following natural embedding maps. For $N \in \N$ and $k \in \N_{\leq N}$, there exists a smooth map
\begin{equation}
\epsilon_{k,N} : \g_{k}\rightarrow \g_{N},
\end{equation}
which embeds a $k$-particle bosonic observable in the space of $N$-particle bosonic operators so as to have the filtration property
\begin{equation}\label{equ:filt}
\comm{\epsilon_{\ell,N}(\g_{\ell})}{\epsilon_{j,N}(\g_{j})}_{\g_N} \subset \epsilon_{\min\{\ell+j-1,N\},N}\paren*{\g_{\min\{\ell+j-1,N\}}} \subset \g_{N}.
\end{equation} 
Using this filtration property and the injectivity of the maps $\ep_{k,N}$, we can now endow $\G_{N}$ with a Lie algebra structure by defining the bracket
\begin{equation}\label{n_lie}
\comm{A}{B}_{\G_{N}}^{(k)} \coloneqq  \sum_{{1\leq \ell,j\leq N}\atop{\min\{\ell+j-1,N\}=k}} \epsilon_{k,N}^{-1}\paren*{\comm{\epsilon_{\ell,N}\paren*{A^{(\ell)}}}{\epsilon_{j,N}\paren*{B^{(j)}}}_{\g_{N}}}, \qquad k\in\{1,\ldots,N\}.
\end{equation}
Furthermore, the maps $\{\epsilon_{k,N}\}_{k=1}^{N}$ induce a Lie algebra homomorphism 
\begin{equation}\label{iota_eps}
\iota_{\epsilon,N}: \G_{N}\rightarrow \g_{N}, \qquad \iota_{\epsilon,N}(A_N) \coloneqq \sum_{k=1}^{N}\epsilon_{k,N}(A_N^{(k)}), \qquad \forall A_N=(A_N^{(k)})_{k\in\N_{\leq N}}.
\end{equation}
In other words, $\iota_{\epsilon,N}$ maps an observable $N$-hierarchy to an $N$-body bosonic observable. In \cref{sec:conv_n_body}, we will establish several properties of the embedding map, which ultimately enable us to prove the following result.

\begin{restatable}{prop}{NHLA}
\label{prop:NH_LA}
$(\G_N,\comm{\cdot}{\cdot}_{\G_N})$ is a Lie algebra in the sense of \cref{def:la}.
\end{restatable}

Next, we define the real topological vector space
\begin{equation}
\G_{N}^{*} \coloneqq \bigl \{\Gamma_{N}=(\gamma_{N}^{(k)})_{k=1}^{N} \in \prod_{k=1}^{N} \L(\Sc_{s}'(\R^{dk}), \Sc_{s}(\R^{dk})) : (\gamma_{N}^{(k)})^{*} = \gamma_{N}^{(k)}\bigr\},
\end{equation}
and we refer to elements of $\G_{N}^{*} $ as \emph{density matrix $N$-hierarchies}. Let $\A_{H,N}$ be the algebra with respect to point-wise product generated by the functionals in the set
\[
\{F\in C^\infty(\G_N^*;\R) : F(\cdot) = i\Tr(A_N\cdot), \enspace A_N\in \G_N\} \cup \{F\in C^\infty(\G_N^*;\R) : F(\cdot)\equiv C\in\R\}.
\]
We can define a Lie-Poisson structure
on $\G_{N}^{*}$, given by
\begin{equation}
\pb{F}{G}_{\G_{N}^{*}}(\Gamma_{N}) \coloneqq i\Tr\paren*{\comm{dF[\Gamma_{N}]}{dG[\Gamma_{N}]}_{\G_{N}}\cdot\Gamma_{N}}, \qquad \forall \Gamma_{N}\in\G_{N}^{*},
\end{equation}
where $F, G \in \A_{H,N}$.

To construct the weak Lie-Poisson manifold $\G_N^*$, a good heuristic to keep in mind is that density matrices are dual to skew-adjoint operators. The superscript $*$, however, does not denote the literal functional analytic dual, but rather denotes a space in weakly non-degenerate pairing with ${\G}_N$. The fact that we only have weak non-degeneracy means that we will be unable to appeal to classical results on Lie-Poisson structures, see for instance \cref{prop:LP_rev} below, and instead we will proceed by direct proof to establish the following result.

\begin{restatable}{prop}{NHWP}
\label{prop:NH_WP}
$(\G_N^*, \A_{H,N}, \pb{\cdot}{\cdot}_{\G_N^*})$ is a weak Poisson manifold.
\end{restatable}

To establish that the BBGKY hierarchy is a Hamiltonian flow on this weak Poisson manifold, we need to prescribe the \emph{BBGKY Hamiltonian functional}
\begin{equation}\label{hbbgky}
\H_{BBGKY,N}(\Gamma_{N}) \coloneqq \Tr\paren*{\W_{BBGKY,N}\cdot\Gamma_{N}}, 
\end{equation}
where $-i\W_{BBGKY,N}$ is the observable $2$-hierarchy defined by
\begin{equation}
\W_{BBGKY,N} \coloneqq (-\Delta_{x}, \kappa V_{N}(X_{1}-X_{2}),0,\ldots).
\end{equation}

We can now state the following theorem, which establishes that the BBGKY hierarchy admits a Hamiltonian formulation and lays the groundwork for our answering of \cref{q:GP_ham}.

\begin{restatable}{thm}{BBGKYham}
\label{thm:BBGKY_ham}
Let $I\subset \R$ be a compact interval. Then $\Gamma_N = (\gamma_N^{(k)})_{k=1}^N \in C^\infty(I;\G_N^*)$ is a solution to the BBGKY hierarchy \eqref{eq:BBGKY} if and only if 
\begin{equation}
\frac{d}{dt} \Gamma_N = X_{\H_{BBGKY,N}} (\Gamma_N),
\end{equation}
where $X_{\H_{BBGKY,N}}$ is the unique vector field defined by $\H_{BBGKY,N}$ (see \cref{def:WP}) with respect to the weak Poisson structure $(\G_N^*, \A_{H,N}, \pb{\cdot}{\cdot}_{\G_N^*})$.
\end{restatable}

\subsection{Derivation of the Lie algebra $\G_{\infty}$ and Lie-Poisson manifold $\G_{\infty}^{*}$}\label{sym_pois}
Having established the necessary framework at the $N$-body level, we are now prepared to address the infinite-particle limit of our constructions. Via the natural inclusion map, one has $\G_N \subset \G_M$ for $M \geq N$. Hence, one has a natural limiting algebra\footnote{This discussion could be formulated more precisely in terms of co-limits of topological spaces ordered by inclusion.} given by
\begin{align}\label{colim}
\mathfrak{F}_\infty \coloneqq \bigcup_{N=1}^\infty \G_N = \bigoplus_{k=1}^\infty \g_k.
\end{align}
By embedding $\G_{N}$ into this limiting algebra, the rather complicated Lie bracket $\comm{\cdot}{\cdot}_{\G_{N}}$ converges pointwise to a much simpler Lie bracket.

We let $\Sym_k$ denote the $k$-particle bosonic symmetrization operator, see \cref{def:sym_A}, and we let $\comm{\cdot}{\cdot}_1$ be a certain separately continuous, bilinear map, the precise definition of which we defer to \cref{sec:geom_N}. We establish the following result.

\begin{restatable}{prop}{LBlim}
\label{prop:LB_lim}
Let $N_0\in\N$. For $A=(A^{(k)})_{k\in\N}, B=(B^{(k)})_{k\in\N} \in \G_{N_0}$, we have that
\begin{equation}
\lim_{N\rightarrow \infty} \comm{A}{B}_{\G_N} = C = (C^{(k})_{k \in \N},
\end{equation}
where 
\begin{equation}
C^{(k)} \coloneqq \sum_{{ \ell,j\geq 1}\atop {\ell+j-1=k}}\Sym_k\paren*{\comm{A^{(\ell)}}{B^{(j)}}_1},
\end{equation}
in the topology of $\mathfrak{F}_\infty$. 
\end{restatable}

The topological vector space given in \eqref{colim} is too small to capture the generator of the GP Hamiltonian, defined in \eqref{equ:gp_ham_intro} below. Indeed, the $2$-particle component $V_N(X_1-X_2)$ of the $N$-body Hamiltonian $H_N$ given in \cref{eq:N_bod_ham} converges to the distribution-valued operator\footnote{Not to be confused with operator-valued distribution.} $\delta(X_1-X_2)$ as $N\rightarrow\infty$. The operator $-i\delta(X_1-X_2)$ does not belong to $\g_2$ since it does not map $\Sc_s(\R^{2d})$ to itself. 

Since we will need our Lie algebra $\G_\infty$ to contain the generator of the GP Hamiltonian functional, this necessitates an underlying topological vector space which includes distribution-valued operators (DVOs). The inclusion of DVOs introduces technical difficulties in  the definition of the bracket $\comm{\cdot}{\cdot}_1$. As we will see, the definition of the bracket $\comm{\cdot}{\cdot}_1$, involves compositions of distribution-valued operators in one coordinate, which in general is not possible. Consequently, we need to find a setting in which we can give meaning to such a composition, thus motivating our introduction of the \emph{good mapping property}:

\begin{restatable}[Good mapping property]{mydef}{gmp}
\label{def:gmp}
Let $i\in\N$. We say that an operator $A^{(i)}\in \L(\Sc(\R^{di}),\Sc'(\R^{di}))$ has the \emph{good mapping property} if for any $\alpha\in\N_{\leq i}$, the continuous bilinear map
\begin{equation*}
\begin{split}
&\Sc(\R^{di}) \times\Sc(\R^{di}) \rightarrow \Sc'(\R^d) \hat{\otimes} \Sc(\R^d) \\
&(f^{(i)},g^{(i)}) \mapsto \int_{\R^{i-1}} dx_{1} \ldots dx_{\alpha-1}dx_{\alpha+1} \ldots dx_{i} A^{(i)}(f^{(i)})(x_1, \ldots, x_i) g^{(i)}(x_1, \ldots, x_{\alpha-1},x_{\alpha}',x_{\alpha+1}, \ldots,x_{i}),
\end{split}
\end{equation*}
may be identified with a continuous bilinear map $\Sc(\R^{di})\times \Sc(\R^{di}) \rightarrow \Sc(\R^{2d})$.\footnote{We use $\hat{\otimes}$ to denote the completion of the tensor product in either the projective or injective topology (which coincide). See \cref{sec:bos} for furhter discussion.}  
\end{restatable}

Here and throughout this paper, an integral should be interpreted as a distributional pairing, unless specified otherwise. We will denote by $\L_{gmp}(\Sc(\R^{di}),\Sc'(\R^{di}))$ the subset of $\L(\Sc(\R^{di}),\Sc'(\R^{di}))$ of operators with the good mapping property.

\begin{remark}
It is evident that $\L_{gmp}(\Sc(\R^{di}),\Sc'(\R^{di}))$ is closed under linear combinations and therefore a subspace. Note that here and throughout we endow $\L(\Sc(\R^{di}),\Sc'(\R^{di}))$ with the topology of uniform convergence on bounded sets, and we endow $\L_{gmp}$ with the subspace topology.  To see that $\L_{gmp}$ is a proper subspace of $\L$, consider the multiplication operator $\delta(\ul{X}_{2}) \in \L(\Sc(\R^{2d}),\Sc'(\R^{2d}))$.
\end{remark}

 The formula for the limiting Lie bracket given in \cref{prop:LB_lim} has a greatly simplified form compared to the $N$-body bracket $\comm{\cdot}{\cdot}_{\G_N}$ due to the vanishing of the higher ``contraction commutators''. Moreover, as we prove in \cref{ssec:GMP}, the good mapping property gives an appropriate definition to the bracket $\comm{A^{(i)}}{B^{(j)}}_1$ as a well-defined element of $\L_{gmp}(\Sc(\R^{dk}),\Sc'(\R^{dk}))$. Hence, we can take advantage of the good mapping property and extend the limiting formula from \cref{prop:LB_lim} to a map on a much larger real topological vector space $\G_\infty$ given by the locally convex direct sum
\begin{align}\label{ginf_intro}
\G_{\infty} \coloneqq \bigoplus_{k=1}^{\infty}\g_{k,gmp}, \quad \g_{k,gmp} \coloneqq \{A^{(k)} \in \L_{gmp}(\Sc_{s}(\R^{dk}), \Sc_{s}'(\R^{dk})) : A^{(k)} = - (A^{(k)})^{*}\}.
\end{align}

We refer to the elements of $\G_{\infty}$ as \emph{observable $\infty$-hierarchies}, and the elements of $\g_{k, gmp}$ as \emph{$k$-particle bosonic observables}. The verification of the Lie algebra axioms then proceeds by direct computation, and we are able to establish the following result.

\begin{restatable}{prop}{LA}
\label{prop:G_inf_br}
$(\G_{\infty},\comm{\cdot}{\cdot}_{\G_{\infty}})$ is a Lie algebra in the sense of \cref{def:la}.
\end{restatable}

Analogously to the $N$-body setting, our second step is the dual problem of building a weak Lie-Poisson manifold $(\G_{\infty}^{*},\A_{\infty},\pb{\cdot}{\cdot}_{\G_{\infty}^{*}})$. If we were in the finite-dimensional setting or a ``nice'' infinite-dimensional setting, such as $\G_{\infty}^{*}$ being a Fr\'{e}chet space and $\G_{\infty}$ being its predual, then this step would follow from standard results (see \cref{ssec:LA_facts}). While $\G_{\infty}^{*}$ is Fr\'{e}chet, the predual of $\G_{\infty}^{*}$ is
\begin{equation}
\bigl\{A=(A^{(k)})_{k\in\N} \in \bigoplus_{k=1}^{\infty} \L(\Sc_{s}(\R^{dk}),\Sc_{s}'(\R^{dk}) : (A^{(k)})^{*} = -A^{(k)} \bigr\},
\end{equation}
which is too large a space for the Lie bracket $\comm{\cdot}{\cdot}_{\G_{\infty}}$ to be well-defined. Therefore, the standard procedure for obtaining a Lie-Poisson manifold from a Lie algebra can only serve as inspiration.

We define the real topological vector space
\begin{equation}
\G_{\infty}^{*} \coloneqq \bigl\{\Gamma=(\gamma^{(k)})_{k\in\N}\in\prod_{k=1}^{\infty} \L(\Sc_{s}'(\R^{dk}),\Sc_{s}(\R^{dk})) : \gamma^{(k)} = (\gamma^{(k)})^{*} \enspace \forall k\in\N\bigr\},
\end{equation}
where the topology is the product topology. Using the isomorphism
\begin{equation}
\L(\Sc_{s}'(\R^{dk}),\Sc_{s}(\R^{dk})) \cong \Sc_{s,s}(\R^{dk}\times\R^{dk}),
%\L(\Sc_{s}'(\R^{k}),\Sc_{s}(\R^{k})) \cong \Sc_{s}(\R^{k}) \otimes \Sc_{s}(\R^{k}) \cong \Sc_{s,s}(\R^{k}\times\R^{k}),
\end{equation}
the elements of $\G_{\infty}^{*}$, which we call \emph{density matrix $\infty$-hierarchies}, are infinite sequences of $k$-particle integral operators with Schwartz class kernels $K(\ux_{k};\ux_{k}')$, which are separately invariant under permutation in the $\ux_{k}$ and $\ux_{k}'$ coordinates.

Let $\A_{\infty}$ be the algebra with respect to point-wise product generated by functionals in the set
\begin{equation}
\begin{split}
&\{F\in C^{\infty}(\G_{\infty}^{*};\R) : F(\cdot) = i\Tr(A\cdot),\,\, A\in\G_{\infty}\}  \cup \{F\in C^{\infty}(\G_{\infty}^{*};\R) : F(\cdot) \equiv C\in\R\}.
\end{split}
\end{equation}
We will observe later that, importantly, our choice of $\A_{\infty}$ contains the observable $\infty$-hierarchy $-i\W_{GP}$, which generates the GP Hamiltonian.

As in the finite-particle setting, the Lie algebra structure on $\G_{\infty}$ canonically induces a Poisson structure on $\G_{\infty}^{*}$. This canonical Poisson structure, which is called a Lie-Poisson structure, is defined by the Poisson bracket
\begin{equation}\label{eq:PB}
\pb{F}{G}_{\G_{\infty}^{*}}(\Gamma)  \coloneqq i\Tr\paren*{\comm{dF[\Gamma]}{dG[\Gamma]}_{\G_{\infty}}\cdot\Gamma}, \qquad \forall \Gamma\in \G_{\infty}^{*},
\end{equation}
where $F,G\in C^{\infty}(\G_{\infty}^{*};\R)$ are functionals in the unital\footnote{i.e. containing a multiplicative identity} sub-algebra $\A_{\infty}$ and we identify the G\^ateaux derivatives $dF[\Gamma], dG[\Gamma]$ as observable $\infty$-hierarchies via the trace pairing $i\Tr(\cdot)$. We will ultimately establish the following result, which provides the underlying geometric structure required to address \cref{q:GP_ham}.

\begin{restatable}{prop}{LP}
\label{prop:LP}
$(\G_{\infty}^{*},\A_{\infty},\pb{\cdot}{\cdot}_{\G_{\infty}^{*}})$ is a weak Poisson manifold.
\end{restatable}

\medskip
Define the \emph{Gross-Pitaevskii Hamiltonian functional} 
\begin{equation}
\H_{GP}:\G_{\infty}^{*}\rightarrow\R
\end{equation}
by
\begin{equation}\label{equ:gp_ham_intro}
\H_{GP}(\Gamma) \coloneqq -\Tr_{1}\paren*{\Delta_{x_1}\gamma^{(1)}} + \Tr_{1,2}\paren*{\delta(X_{1}-X_{2})\gamma^{(2)}}, \qquad \Gamma=(\gamma^{(k)})_{k\in\N} \in \G_{\infty}^{*},
\end{equation}
where $\Tr_{1,\ldots,j}$ denotes the $j$-particle generalized trace, see \cref{sec:trace} for definition and discussion. Then we can rewrite $\H_{GP}$ as
\begin{equation}
\label{eq:H_GP}
\H_{GP}(\Gamma) = \Tr\paren*{\W_{GP}\cdot\Gamma}, \qquad \W_{GP}\coloneqq (-\Delta_{x_1}, \delta(X_{1}-X_{2}),0,\ldots),
\end{equation}
which one should compare with \eqref{hbbgky}. 

\begin{remark}
Note that $-i\W_{GP}$ is an observable $\infty$-hierarchy, that is, an element of $\G_\infty$. Since we have the convergence $-i\W_{BBGKY,N} \rightarrow -i\W_{GP}$ in $\G_\infty$, as $N\rightarrow\infty$, it follows that $\H_{BBGKY,N}\rightarrow\H_{GP}$ in $C^\infty(\G_\infty^*;\R)$ endowed with the topology of uniform convergence on bounded sets.
\end{remark}

 We now state our next main result, which addresses the final component of \cref{q:GP_ham}:
\begin{restatable}[Hamiltonian structure for GP]{thm}{GPham}
\label{thm:GP_ham}
Let $I\subset \R$ be a compact interval. Then $\Gamma\in C^{\infty}(I;\G_{\infty}^{*})$ is a solution to the GP hierarchy \eqref{eq:GP} if and only if
\begin{equation}
\paren*{\frac{d}{dt}\Gamma}(t) = X_{\H_{GP}}(\Gamma(t)), \qquad \forall t\in I,
\end{equation}
where $X_{\H_{GP}}$ is the unique Hamiltonian vector field defined by $\H_{GP}$ with respect to the weak Poisson structure $(\G_{\infty}^{*},\A_{\infty},\pb{\cdot}{\cdot}_{\G_{\infty}^{*}})$. 
\end{restatable}

\begin{remark}\label{rem:all_dim}
The result of \cref{thm:GP_ham} extends, with an almost identical proof, to the Hartree hierarchy, and it seems likely that this result should also extend to the quintic GP hierarchy \cite{CP2011} and other variants which account for higher-order particle interactions \cite{X2015}. 
\end{remark}

We now give a geometric formulation of the procedure by which one obtains the BBGKY hierarchy from the $N$-body Schr\"odinger equation. The results described below will be proved in \cref{ssec:dm_pomo}. To record the Hamiltonian structure for the $N$-body Schr\"odinger equation, we equip the bosonic Schwartz space $\Sc_s(\R^{dN})$ with the standard symplectic structure and define the Hamiltonian functional
\begin{equation}
\H_{N}(\Phi_{N}) \coloneqq \frac{1}{N}\int_{\R^{dN}}d\ux_{N}\ol{\Phi_{N}(\ux_{N})} \paren*{H_{N}\Phi_{N}}(\ux_{N}), \qquad \forall \Phi_{N} \in\Sc_{s}(\R^{dN}).
\end{equation}
Then the Schr\"{o}dinger equation \cref{eq:LSchr} can be viewed as a Hamiltonian flow on this weak symplectic manifold. We can endow the space $\L(\Sc_{s}'(\R^{dN}),\Sc_{s}(\R^{dN}))$ of bosonic density matrices with a weak Poisson structure by defining
\begin{equation}
\pb{F}{G}_{N} \coloneqq i\Tr_{1,\ldots,N}\paren*{\comm{dF[\Psi_{N}]}{dG[\Psi_{N}]}_{\g_N}\Psi_{N}}, \qquad \forall \Psi_{N}\in\L(\Sc_{s}'(\R^{dN}),\Sc_{s}(\R^{dN})),
\end{equation}
where $dF$ and $dG$ denote the G\^ateaux derivatives, see \cref{gateaux_deriv}, of $F$ and $G$, which are smooth real-valued functionals with suitably regular G\^ateaux derivatives. Then the Poisson bracket $\pb{\cdot}{\cdot}_{N}$ is a Lie-Poisson bracket induced by the Lie algebra of $N$-body bosonic observables with Lie bracket given by $\comm{\cdot}{\cdot}_{\g_N}$.

There is a canonical map from $N$-body wave functions to $N$-body density matrices given by
\begin{equation}
\iota_{DM,N}: \Sc_{s}(\R^{dN}) \rightarrow \L(\Sc_{s}'(\R^{dN}),\Sc_{s}(\R^{dN})), \qquad \iota_{DM,N}(\Phi_{N}) \coloneqq \ket*{\Phi_{N}}\bra*{\Phi_{N}}.
\end{equation}
We will show in \cref{prop:DM_po} that 
\[
\iota_{DM,N}: (\Sc_{s}(\R^{dN}), \pb{\cdot}{\cdot}_{L^2, N} ) \to (\L(\Sc_{s}'(\R^{dN}),\Sc_{s}(\R^{dN})), \pb{\cdot}{\cdot}_{N}),
\]
is a Poisson morphism\footnote{We recall $\pb{\cdot}{\cdot}_{L^2, N} = N \pb{\cdot}{\cdot}_{L^2}$, and see \eqref{pb_intro} for a definition of $\pb{\cdot}{\cdot}_{L^2}$. We also note that the co-domain of this map will be replaced by the appropriate space of $N$-body density matrices. } and consequently maps solutions of the Schr\"{o}dinger equation \cref{eq:LSchr} to solutions of the von Neumann equation
\begin{equation}\label{eq:vN}
i\p_{t}\Psi_{N} = \comm{H_{N}}{\Psi_{N}},
\end{equation}
where the right-hand side denotes the usual commutator. Defining the Hamiltonian functional
\begin{equation}
\H_{N}(\Psi_{N}) \coloneqq \frac{1}{N}\Tr_{1,\ldots,N}\paren*{H_{N}\Psi_{N}}, \qquad \forall \Psi_{N} \in \L(\Sc_{s}'(\R^{dN}),\Sc_{s}(\R^{dN})),
\end{equation}
the von Neumann equation \eqref{eq:vN} can be viewed as a Hamiltonian equation of motion on the weak Poisson manifold $(\L(\Sc_{s}'(\R^{dN}),\Sc_{s}(\R^{dN})), \pb{\cdot}{\cdot}_{N})$. We will prove in \cref{prop:RDM_Po} that the dual of the map $\iota_{\epsilon,N}$ given in \eqref{iota_eps} induces a canonical morphism of Poisson manifolds, which is precisely the \emph{reduced density matrix map}, given by
\begin{equation}
\iota_{RDM,N}=\iota_{\epsilon,N}^*: \g_{N}^{*} \rightarrow \G_{N}^{*}, \qquad \iota_{RDM,N}(\Psi_{N}) \coloneqq (\Tr_{k+1,\ldots,N}(\Psi_{N}))_{k=1}^{N} \eqqcolon (\gamma_N^{(k)})_{k=1}^N,
\end{equation}
which maps solutions of the von Neumann equation to solutions of the quantum BBGKY hierarchy.

\subsection{The connection with the NLS} 

We will now tie together our main results and state the result which provides an affirmative answer to \cref{q:conn_1bdy}. We connect the GP hierarchy to the cubic NLS, each as infinite-dimensional Hamiltonian systems, through the canonical embedding
\begin{equation}
\label{eq:can_emb}
\iota: \Sc(\R^d) \rightarrow \G_\infty^*, \qquad \phi \mapsto (\ket*{\phi^{\otimes k}}\bra*{\phi^{\otimes k}})_{k\in\N}.
\end{equation}
Although $\iota$ is rather trivial in terms of the simplicity of its definition, and for this reason we sometimes refer to $\iota$ as the trivial embedding, it has the important property of being a Poisson morphism (see \cref{def:po_map} below).

\begin{restatable}{thm}{pomo}
\label{thm:pomo}
The map $\iota$ is a Poisson morphism of $(\Sc(\R^d),\A_{\Sc},\pb{\cdot}{\cdot}_{L^2})$ into $(\G_\infty^*,\A_{\infty},\pb{\cdot}{\cdot}_{\G_\infty^*})$, i.e. it is a smooth map such that
\begin{equation}
\pb{F\circ\iota}{G\circ\iota}_{L^2}(\phi) =\pb{F}{G}_{\G_\infty^*}(\iota(\phi)), \qquad \forall \phi\in\Sc(\R^d),
\end{equation}
for all functionals $F,G\in\A_{\infty}$.
\end{restatable}

We conclude by discussing why the results described in this section provide ``a rigorous derivation of the Hamiltonian structure of the NLS''. It is a quick computation to show that the pullback of the GP Hamiltonian \eqref{eq:H_GP} under the map $\iota$, denoted by $\iota^*\H_{GP}$, equals the NLS Hamiltonian \eqref{eq:H_NLS},\footnote{In particular, as a corollary of \cref{thm:GP_ham} and \cref{thm:pomo}, we obtain the well-known fact that if $\phi(t)$ is a solution to the cubic NLS \eqref{nls}, then $\Gamma(t)\coloneqq \iota(\phi(t))$ is a solution to the GP hierarchy \eqref{eq:GP}.} that is
\begin{align}\label{equ:pull}
\iota^*\H_{GP} = \H_{NLS}.
\end{align}
Hence, \cref{thm:pomo}, \cref{thm:GP_ham} and \eqref{equ:pull} ultimately demonstrate that the Hamiltonian functional and phase space of the NLS can be obtained via the pullback of the canonical embedding \eqref{eq:can_emb}. Together with the results of \cref{ssec:dm_pomo}, which provide a geometric correspondence between the $N$-body Schrodinger equation and the BBGKY hierarchy, and \cref{prop:LB_lim}, which enables us to take the infinite-particle limit of our geometric constructions at the $N$-body level, this provides a rigorous derivation of the Hamiltonian structure of the NLS from the Hamiltonian formulation of the $N$-body Schr\"odinger equation.

\subsection{Organization of the paper}

\cref{sec:pre} is devoted to preliminary material on weak Poisson manifolds modeled on locally convex spaces, Lie algebras, and tensor products. The reader familiar with infinite-dimensional Poisson manifolds and Lie algebras may wish to skip the first two subsections upon first reading and instead consult them as necessary during the reading of \cref{sec:geom_N} and \cref{sec:geom}.

In \cref{sec:geom_N}, we build the requisite Lie algebra structure for $\G_{N}$ and weak Lie-Poisson structure for $\G_{N}^{*}$, thereby proving \cref{prop:NH_LA} and \cref{prop:NH_WP}. \cref{ssec:N_LA} contains the Lie algebra construction, and \cref{ssec:LP_N} contains the dual Lie-Poisson construction. Lastly, in \cref{ssec:dm_pomo}, we show that the familiar maps of forming a density matrix from a wave function and taking the sequence of reduced density matrices of a density matrix have geometric content. Namely, we prove \cref{prop:DM_po} and \cref{prop:RDM_Po}, which assert that these maps are Poisson morphisms.

In \cref{sec:geom}, we build the requisite Lie algebra structure for $\G_{\infty}$ and weak Lie-Poisson structure for $\G_{\infty}^{*}$, thereby proving \cref{prop:G_inf_br} and \cref{prop:LP}. The section is broken up into several subsections. \cref{ssec:geo_LA} is devoted the Lie algebra construction, and \cref{ssec:geo_LPM} is devoted to the dual Lie-Poisson construction. Finally, we will prove \cref{thm:pomo} in \cref{ssec:po_emb}.

Lastly, in \cref{sec:gp_flows}, we prove our Hamiltonian flows results \cref{thm:BBGKY_ham} and \cref{thm:GP_ham}, which assert that the BBGKY and GP hierarchies, respectively, are Hamiltonian flows on the weak Lie-Poisson manifolds constructed in the previous sections.

\begin{remark}
In \cref{sec:conv_n_body}, \cref{sec:geom}, and  \cref{sec:gp_flows}, we will fix the dimension to be one for simplicity, but we emphasize that our results hold \textit{independently} of the dimension.
\end{remark}

\medskip
We have also included two appendices to make this work as self-contained as possible. \cref{local_cvx} contains some background material on locally convex spaces, specifying certain choices which we make in the current work, which in infinite dimensions can lead to non-equivalent definitions. \cref{app:DVO} is devoted to technical facts about distribution-valued operators and topological tensor products, which justify the manipulations used extensively in this paper. Furthermore, this appendix includes an elaboration on the good mapping property, in particular, some technical consequences of it which are used in the body of the paper.

\section{Notation}\label{sec:not}
\subsection{Index of notation}

We include \cref{tab:notation} as a notational guide for the various symbols which appear in this work. In this table, we either provide a definition of the notation or a reference for where the symbol is defined. When definitions for these objects may have appeared in the introduction, we will give references to where they first appear in subsequent sections.

\begin{table}[h]
   \caption{Notation} 
   \label{tab:notation}
   \small % text size of table content
   \centering % center the table
   \begin{tabular}{ll} % alignment of each column data
   \textbf{Symbol} & \textbf{Definition} \\ 
      \hline
%   \midrule
   $(\ul{x}_k)$, $\ul{x}_{k}$ & $(x_1, \ldots, x_k)$ \\
   $\ul{x}_{m_1; m_k}$ & $(x_{m_1}, \ldots, x_{m_k})$ \\
   $\ul{x}_{i;i+k}$ & $(x_i, \ldots, x_{i+k})$ \\
      $d\ul{x}_{k}$ & $ dx_1 \cdots dx_k$ \\
         $d\ul{x}_{i;i+k}$ & $dx_i \cdots dx_{i+k}$ \\
      $\mathbb{N}_{\leq i}$ or $\mathbb{N}_{\geq i}$  & $\{ n \in \mathbb{N} \,:\, n \leq i \}$ or $\{ n \in \mathbb{N} \,:\, n \geq i \}$ \\
      $\Ss_{k}$ & symmetric group on $k$ elements \\
      $\Sc(\R^{k}), \Sc'(\R^{k})$ & Schwartz space on $\R^k$ and tempered distributions on $\R^k$ \\
      $\mathcal{D}'(\R^{k})$ & distributions on $\R^k$ \\
      $\Sc_s(\R^{k}), \Sc_s'(\R^{k})$ & symmetric Schwartz space, \cref{sym_schwartz}, and symmetric tempered distributions\\
      $\L(E;F)$ & continuous linear maps between locally convex spaces $E$ and $F$\\
      $\tl{\L}(\Sc(\R^k),\Sc(\R^k))$ & $\L(\Sc(\R^k),\Sc(\R^k))$ equipped with the subspace topology induced by $\L(\Sc(\R^k),\Sc'(\R^k))$\\
      $\tl{\L}(\Sc_s(\R^k),\Sc_s(\R^k))$ & analogous to previous definition \\
      $dF$ & the G\^ateaux derivative of $F$, \cref{gateaux_deriv}\\
      $\grad$ or $\grad_{s}$ & the real or symplectic $L^2$ gradients, \cref{def:re_grad} and \cref{schwartz_deriv}\\
       $A_{(\pi(1),\ldots, \pi(k))}$ & conjugation of an operator by a permutation, see \eqref{eq:op_coord} \\
       $\Sym(f)$ & symmetrization operator for functions, \cref{def:sym_f} \\
       $\Sym(A)$ & symmetrization operator for operators, \cref{def:sym_A}\\
       $L^2_s(\R^k)$ & symmetric wave functions, \cref{sym_wave}\\
       $B_{i;j}^{\pm}, B_{i;j}$ & contraction operators, \cref{contraction}\\
       $\phi^{\otimes k}$ & $k$-fold tensor of $\phi$ with itself, \eqref{tensor_def}\\
       $\omega_{L^2}$ & symplectic form on $L^2(\R^k)$, \eqref{l2_symp}\\
       $\A_{\Sc}$ & see \cref{schwartz_wpoiss} and \eqref{equ:Asc}\\
       $\{ \cdot, \cdot\}_{L^2}$ & Poisson bracket on $L^2(\R^k)$, \eqref{l2_bracket}\\
       $A_{(j_1,\ldots,j_k)}^{(k)}$ & $k$-particle extension, \eqref{k_part} \\
	$\g_k$ & locally convex space of $k$-body bosonic observables, \eqref{gk_def}\\
	$(\G_N, [\cdot, \cdot]_{\G_N})$ & Lie algebra of observable $N$-hierarchies, \eqref{GN_def} \\
	$\circ_r$ & r-fold contraction, \eqref{eq:A_c_B}\\
	$(\G_N^*, \A_\infty, \{\cdot, \cdot\}_{\G_N^*})$ & Lie-Poisson manifold of density matrix $N$-hierarchies, \eqref{GNstar_def}\\
	$\g_{k,gmp}$ & locally convex space of $k$-body observables satisfying the good mapping property, \eqref{gkgmp_def}\\
        $(\G_\infty, [\cdot, \cdot]_{\G_\infty})$ & Lie algebra of observable $\infty$-hierarchies, \eqref{Ginf_def} and \eqref{eq:C_def}\\
        $\circ_{\alpha}^\beta$ & contraction operator, \cref{lem:gmp}\\
        $(\G_\infty^*, \A_\infty, \{\cdot, \cdot\}_{\G_\infty^*})$ & Lie-Poisson manifold of density matrix $\infty$-hierarchies, \eqref{Ginf_star_def}, \cref{a_inf_def} and  \eqref{equ:poisson_def}\\
        $\Tr_{1,\ldots,N}$ & generalized trace, \cref{def:gen_trace}\\
        $\Tr_{k+1,\ldots,N}$ & generalized partial trace, \cref{prop:partial_trace}
   \end{tabular}
\end{table}

\section{Preliminaries}\label{sec:pre}

\subsection{Weak Poisson structures and Hamiltonian systems}
The classical notion of Poisson structure, as can be found in \cite{MR2013}, is ill-suited outside the Hilbert or Banach manifold setting due to the fact that for a given smooth, locally convex manifold $M$, not every functional in $C^\infty(M, \R)$, the space of smooth, real-valued functionals on $M$, need admit a Hamiltonian vector field. Since we will need to work with Fr\'{e}chet manifolds, an alternative theory is needed. We opt for the notion of a \emph{weak Poisson structure} due to Neeb et al. \cite{NST2014}. 

We recall that a unital subalgebra $\A \subseteq C^{\infty}(M;\R)$ contains constant functions and is closed under pointwise multiplication.
\begin{mydef}[Weak Poisson manifold]\label{def:WP}
A \emph{weak Poisson structure} on $M$ is a unital subalgebra $\mathcal{A}\subset C^{\infty}(M;\R)$ and a bilinear map $\pb{\cdot}{\cdot}:\mathcal{A}\times\mathcal{A}\rightarrow\mathcal{A}$ satisfying the following properties:
\begin{enumerate}[(P1)]
\item\label{item:wp_P1}
The bilinear map $\pb{\cdot}{\cdot}$, is a Lie bracket and satisfies the Leibnitz rule
\begin{equation}
\pb{F}{GH} = \pb{F}{G}H+G\pb{F}{H}, \qquad \forall F,G,H\in\mathcal{A}.
\end{equation}
We call $\pb{\cdot}{\cdot}$ a \emph{Poisson bracket}.
\item\label{item:wp_P2}
For all $m\in M$ and $v\in T_{m}M$ satisfying $dF[m](v)=0$ for all $F\in\A$, we have that $v=0$.
\item\label{item:wp_P3}
For every $H\in\A$, there exists a smooth vector field $X_{H}$ on $M$ satisfying
\begin{equation}
\label{eq:ham_vf_deriv}
X_{H}F=\pb{F}{H}, \qquad \forall F\in\A.\footnote{In the left-hand side of identity \eqref{eq:ham_vf_deriv}, we use the notation $X_H$ to denote the vector field identified as a derivation.}
\end{equation}
We call $X_{H}$ the \emph{Hamiltonian vector field} associated to $H$.
\end{enumerate}
If properties \ref{item:wp_P1} - \ref{item:wp_P3} are satisfied, then we call the triple $(M,\A,\pb{\cdot}{\cdot})$ a \emph{weak Poisson manifold}.
\end{mydef}

We now record some observations from \cite{NST2014} about the definition of a weak Poisson structure.

\begin{remark}\label{rem:hvf_u}
\ref{item:wp_P2} implies that the Hamiltonian vector field $X_{H}$ associated to some $H\in\A$ is uniquely determined by the relation
\begin{equation}
\pb{F}{H}(m) = (X_{H}F)(m)=dF[m](X_{H}(m)), \qquad \forall F\in\A.
\end{equation}
Indeed, if $X_{H,1}$ and $X_{H,2}$ are two smooth vector fields satisfying the preceding relation, then the smooth vector $\widetilde{X}_{H}\coloneqq X_{H,1}-X_{H,2}$ satisfies
\begin{equation}
dF[m](\widetilde{X}_{H}(m))=0, \qquad \forall F\in\A,
\end{equation}
for all $m\in M$, which by \ref{item:wp_P2} implies that $\widetilde{X}_{H}\equiv 0$.
\end{remark}

\begin{remark}
For all $F,G,H\in\A$, we have that
\begin{align}
\comm{X_{F}}{X_{G}}H &= \pb{\pb{H}{G}}{F}-\pb{\pb{H}{F}}{G} \nonumber\\
&= \pb{H}{\pb{G}{F}} \nonumber\\
&= X_{\pb{G}{F}}H.
\end{align}
Hence, by \cref{rem:hvf_u}, $\comm{X_{F}}{X_{G}}=X_{\pb{G}{F}}$ for $F,G\in\A$. Additionally, the Leibnitz rule for $\pb{\cdot}{\cdot}$ implies the identity
\begin{equation}\label{eq:vf_L}
X_{FG}=FX_{G}+GX_{F}, \qquad \forall F,G\in\A.
\end{equation}
\end{remark}

\begin{remark}\label{rem:gen_vf}
If $\mathcal{A}\subset C^{\infty}(M;\R)$ is a unital sub-algebra which satisfies properties \ref{item:wp_P1} and \ref{item:wp_P2} of \cref{def:WP}, then \cref{eq:vf_L} implies that the subspace
\begin{equation}
\{H\in \mathcal{A} : X_{H} \enspace \text{exists as in \ref{item:wp_P3}}\}
\end{equation}
is a sub-algebra of $\mathcal{A}$ with respect to pointwise product. Hence, it suffices to verify property \ref{item:wp_P3} for a generating subset $\A_{0}\subset \A$.
\end{remark}

We note that unlike in the finite-dimensional setting, a symplectic form $\omega:V\times V\rightarrow \R$ on an infinite-dimensional locally convex space $V$ need not represent every continuous linear functional via $\omega(\cdot,v)$, for some $v\in V$. If the  form does satisfy such a Riesz-representation-type condition, we call a symplectic form $\omega$ \emph{strong}, otherwise, we call $\omega$ \emph{weak}. Analogously, a 2-form $\omega$ on a smooth locally convex manifold $M$ is strong (resp. weak) if all forms $\omega_{p}:T_{p}M\times T_{p}M\rightarrow\R$, for $p\in M$, are strong (resp. weak).

\begin{mydef}[Weak symplectic manifold]\label{weak_sym}
Let $M$ be a smooth locally convex manifold, and let $\mathcal{X}(M)$ denote smooth vector fields on $M$. A \emph{weak symplectic manifold} is a pair $(M,\omega)$ consisting of a smooth manifold $M$ and a closed non-degenerate 2-form $\omega$ on $M$. 
\end{mydef}

Given a weak symplectic manifold, we denote the Lie algebra of Hamiltonian vector fields on $M$ by
\begin{equation}
\mathrm{ham}(M,\omega) \coloneqq \{X\in\mathcal{X}(M) : \exists H\in C^{\infty}(M;\R) \enspace \text{s.t.} \enspace \omega(X, \cdot) = dH\}.
\end{equation}
Similarly, we denote the larger Lie algebra of symplectic vector fields on $M$ by
\begin{equation}
\mathrm{sp}(M,\omega) \coloneqq \{X\in\mathcal{M} : \L_{X}\omega =0\},
\end{equation}
where $\L_X$ denotes the Lie derivative with respect to the vector field $X$.

With this definition  in hand, we see that one has the desired implication analogous to the finite dimensional setting, namely that weak symplectic manifolds canonically lead to weak Poisson manifolds. 
\begin{remark}[Weak symplectic $\Rightarrow$ weak Poisson] \label{prop:sym_poi}
Let $(M,\omega)$ be a weak symplectic manifold. Let
\begin{equation}
\A \coloneqq \{H\in C^{\infty}(M;\R) : \exists X_{H}\in\mathcal{X}(M) \enspace \text{s.t.} \enspace \omega( X_{H}, \cdot) =dH\},
\end{equation}
then
\begin{equation}
\pb{\cdot}{\cdot}:\A\times\A\rightarrow\A, \qquad \pb{F}{G} \coloneqq \omega(X_{F},X_{G}) = dF[X_{G}]=X_{G}F
\end{equation}
defines a Poisson bracket on $\A$ satisfying properties \ref{item:wp_P1} and \ref{item:wp_P3}. If we additionally have that for each $m\in\M$ and all $v\in T_{m}M$, the condition
\begin{equation}
\omega(X(m),v)=0, \qquad \forall X\in\mathrm{ham}(M,\omega)
\end{equation}
implies that $v=0$, then property \ref{item:wp_P2} is also satisfied. Consequently, the triple $(M,\A,\pb{\cdot}{\cdot})$ is a weak Poisson manifold.
\end{remark}

We now turn to mappings between weak Poisson manifolds which preserve the Poisson structures. This leads to the notion of a Poisson mapping, alternatively Poisson morphism.

\begin{mydef}[Poisson map]\label{def:po_map}
Let $(M_{j},\A_{j},\pb{\cdot}{\cdot}_{j})$, for $j=1,2$, be weak Poisson manifolds. We say that a smooth map $\varphi:M_{1}\rightarrow M_{2}$ is a \emph{Poisson map}, or \emph{morphism of Poisson manifolds}, if $\varphi^{*}\A_{2}\subset \A_{1}$ and
\begin{equation}
\varphi^{*}\pb{F}{G}_{2}=\pb{\varphi^{*}F}{\varphi^{*}G}_{1}, \qquad \forall F,G\in\A_{2}.
\end{equation}
\end{mydef}

\begin{remark}
In \cite{NST2014}, the authors define a Poisson morphism 
\[
\varphi: (M_{1},\A_{1},\pb{\cdot}{\cdot}_{1}) \rightarrow (M_{2},\A_{2},\pb{\cdot}{\cdot}_{2})
\]
with the requirement that $\varphi^{*}\A_{2}=\A_{1}$. We drop this requirement in our \cref{def:po_map}.
\end{remark}

As an example, we demonstrate that the Schwartz space $\Sc(\R^{k})$ is a weak, but not strong, symplectic manifold. The following analysis also holds for the bosonic Schwartz space $\Sc_{s}(\R^{k})$ \emph{mutatis mutandis}, which will be important for our applications in the sequel.

We equip the space $\Sc(\R^{k})$ with a real pre-Hilbert inner product by defining
\begin{equation}
\ip{f}{g}_{\Re} \coloneqq 2\Re{ \int_{\R^{k}}d\ux_{k}\ol{f(\ux_{k})} g(\ux_{k})}.
\end{equation}
The operator $J: \Sc(\R^{k})\rightarrow\Sc(\R^{k})$ defined by $J(f) \coloneqq if$ defines an almost complex structure on $(\Sc(\R^{k}),\ip{\cdot}{\cdot}_{\Re})$, leading to the \emph{standard $L^{2}$ symplectic form}
\begin{equation}\label{l2_symp}
\omega_{L^{2}}(f,g) \coloneqq \ip{Jf}{g}_{\Re} = 2 \Im{\int_{\R^{k}}d\ux_{k}\ol{f(\ux_{k})}g(\ux_{k})}, \qquad \forall f,g\in\Sc(\R^{k}).
\end{equation}

\begin{prop}\label{prop:Sch_wsym}
$(\Sc(\R^{k}),\omega_{L^{2}})$ is a weak symplectic manifold.
\end{prop}
\begin{proof}
$\Sc(\R^{k})$ is trivially a smooth manifold modeled on itself. Moreover, it is evident from its definition that $\omega_{L^{2}}$ is bilinear, alternating, and closed. To see that $\omega_{L^{2}}$ is non-degenerate, let $f\in\Sc(\R^{k})$ and suppose that
\begin{equation}
\omega_{L^{2}}(f,g) = 0 \qquad \forall g\in\Sc(\R^{k}).
\end{equation}
It then follows tautologically that $\Im{\ip{f}{g}}=0$. Replacing $g$ by $ig$, we obtain that $\Re{\ip{f}{g}}=0$, which implies that $\ip{f}{f}=0$, hence $f=0$.
\end{proof}

Now given a functional $F\in C^{\infty}(\Sc(\R^{k});\R)$, the G\^ateaux derivative $dF[f]$ at the point $f\in\Sc(\R^{k})$ defines a tempered distribution. We consider the case when $dF[f]$ can be identified with a Schwartz function via the inner product $\ip{\cdot}{\cdot}_{\Re}$. The next lemma follows by the Lebesgue lemma\footnote{We use the name Lebesgue lemma to refer to the result that if $u_1,u_2$ are two locally integrable functions such that $u_1=u_2$ in distribution, then $u_1=u_2$ point-wise almost everywhere.} and the same argument used to prove non-degeneracy in \cref{prop:Sch_wsym}.

\begin{lemma}[Uniqueness of gradient]\label{lem:grad_unq}
Let $F\in C^{\infty}(\Sc(\R^{k});\R)$ and $f\in\Sc(\R^{k})$. Suppose that there exist $g_{1},g_{2}\in\Sc(\R^{k})$ such that
\begin{equation}
\ip{g_{1}}{\delta f}_{\Re} = dF[f](\delta f) = \ip{g_{2}}{\delta f}_{\Re}, \qquad \forall \delta f\in\Sc(\R^{k}).
\end{equation}
Then $g_{1}=g_{2}$.
\end{lemma}

\begin{mydef}[Real $L^{2}$ gradient]\label{def:re_grad}
We define the \emph{real $L^{2}$ gradient} of $F\in C^{\infty}(\Sc(\R^{k});\R)$ at the point $f\in\Sc(\R^{k})$, denoted by $\grad F(f)$, to be the unique element of $\Sc(\R^{k})$ (if it exists) such that
\begin{equation}
dF[f](\delta f) = \ip{\grad F(f)}{\delta f}_{\Re}, \qquad \forall \delta f\in\Sc(\R^{k}).
\end{equation}
We say that $F$ has a real $L^{2}$ gradient if $\grad F:\Sc(\R^{k}) \rightarrow\Sc(\R^{k})$ is a smooth map.
\end{mydef}

\begin{remark}\label{schwartz_deriv}
Since the Hamiltonian vector field of $X_{F}$, if it exists, is defined by the relation
\begin{equation}
dF[f](\delta f)= \omega_{L^{2}}(X_{F}(f),\delta f),
\end{equation}
and since $X_F$ is unique by the fact that $\Sc(\R^{k})$ is dense in $\Sc'(\R^{k}$),  we see that $X_{F}(f) = -i\grad F(f)$. In the sequel, we will use the notation $\grad_{s} F \coloneqq X_{F}$, which we refer to as the \emph{symplectic $L^{2}$ gradient}.
\end{remark}

We now use \cref{prop:sym_poi} to show that the symplectic form $\omega_{L^{2}}$, which we recall is defined in \eqref{equ:poisson1}, canonically induces an $L^{2}$ Poisson structure on $\Sc(\R^{k})$.

\begin{prop}\label{schwartz_wpoiss}
Define a subset $A_{\Sc}\subset C^{\infty}(\Sc(\R^{k});\R)$ by
\begin{equation}\label{equ:Asc}
\A_{\Sc} \coloneqq \bigl \{ H \,:\, \grad_{s}H \in C^{\infty}(\Sc(\R^k);\Sc(\R^k)) \bigr\} ,
\end{equation}
and define a bracket $\pb{\cdot}{\cdot}_{L^{2}}$ on $\A_{\Sc} \times\A_{\Sc}$ by
\begin{equation}\label{l2_bracket}
\pb{F}{G}_{L^{2}} \coloneqq \omega_{L^{2}}(\grad_{s}F,\grad_{s}G).
\end{equation}
Then $(\Sc(\R^{k}),\A_{\Sc}, \pb{\cdot}{\cdot}_{L^{2}})$ is a weak Poisson manifold.
\end{prop}
\begin{proof}
By \cref{prop:sym_poi}, we only need to check that for every fixed $g\in \Sc(\R^{k})$, the condition
\begin{equation}\label{eq:syp_con}
\omega_{L^{2}}(X(f),g)=0, \qquad \forall X\in \mathrm{ham}(\Sc(\R^{k}),\omega_{L^{2}})
\end{equation}
implies that $g=0\in\Sc(\R^{k})$. Since $\mathrm{ham}(\Sc(\R^{k}),\omega_{L^{2}})$ contains the constant vector fields $X(\cdot)\equiv f_{0}$, for any fixed $f_{0}\in\Sc(\R^{k})$, we see that by taking $X(f)\coloneqq ig$ for all $f\in\Sc(\R^{k})$, that the condition \cref{eq:syp_con} implies that
\begin{equation}
0=\omega(ig,g) = -2\Im{\int_{\R^{k}}d\ux_{k}\ol{(ig)}(\ux_{k})g(\ux_{k})} = 2\|g\|_{L^{2}(\R^{k})}^{2}.
\end{equation}
Hence, $g=0$, completing the proof.
\end{proof}

\subsection{Some Lie algebra facts}\label{ssec:LA_facts}
In this subsection, we collect some facts about Lie algebras for easy referencing. We outline a canonical construction of a Poisson structure on the dual of a Lie algebra, which is known as a \emph{Lie-Poisson structure}. Furthermore, we will outline a construction of hierarchies of Lie algebras which will serve as an inspiration for our construction of the Lie algebra $\G_{\infty}$. We refer the reader to \cite{MR2013, MMW1984} for more background and details. 

We begin by recording the definition of a Lie algebra for subsequent reference in our proofs.

\begin{mydef}[Lie algebra]\label{def:la}
A \emph{Lie algebra} is a locally convex space $\g$ over the field $\mathbb{F} \in \{\R,\C\}$ together with a separately continuous binary operation $\brak{\cdot,\cdot}:\g\times \g \rightarrow \g$ called the \emph{Lie bracket}, which satisfies the following properties:
\begin{enumerate}[(L1)]
\item\label{item:LA_1}
$[\cdot,\cdot]$ is bilinear.
\item\label{item:LA_2}
$[x,x]=0$ for all $x\in \g$.
\item\label{item:LA_3}
$\brak{\cdot,\cdot}$ satisfies the \emph{Jacobi identity}
\begin{equation}
\brak*{x,\brak*{y,z}} + \brak*{z,\brak*{x,y}} + \brak*{y,\brak*{z,x}}=0
\end{equation}
for all $x,y,z\in \g$.
\end{enumerate}
\end{mydef}

\begin{remark}
Usually (see, for instance, \cite{Omori1979}), a Lie bracket is required to be continuous, as opposed to separately continuous. We drop this requirement in this work, due to functional analytic difficulties.
\end{remark}

\begin{mydef}[Nondegenerate pairings]
Let $V$ and $W$ be topological vector spaces over the field $\mathbb{F}$, and let 
\[
\ip{\cdot}: V\times W\rightarrow \mathbb{F}
\]
be a bilinear pairing between $V$ and $W$. We say that the pairing is \emph{$V$-nondegenerate} (respectively, \emph{$W$-nondegenerate}) if the map $V\rightarrow W^{*}, x\mapsto \ip{x}{\cdot}$ (respectively, $W\rightarrow V^{*}, y\mapsto \ip{\cdot}{y}$) is an isomorphism. If the pairing is both $V$- and $W$-nondegenerate, then we say that the pairing is \emph{nondegenerate}.
\end{mydef}

\begin{mydef}[dual space $\g^{*}$]
Let $(\g,\brak{\cdot,\cdot})$ be a Lie algebra. We say that a topological vector $\g^{*}$ is a \emph{dual space} to $\g$ if there exists a pairing $\ip{\cdot}: \g \times \g^{*}\rightarrow \mathbb{F}$ which is nondegenerate.
\end{mydef}

\begin{ex}
If $\g$ is a reflexive Fr\'{e}chet space, for instance the Schwartz space $\mathcal{S}(\mathbb{R}^{d})$, then taking $\g^{*}$ to be the topological dual of $\g$ equipped with the strong dual topology, the standard duality pairing
\begin{equation*}
\g\times \g^{*} \rightarrow \mathbb{F}: \ip{x}{\varphi} = \varphi(x)
\end{equation*}
is nondegenerate.
\end{ex}

A consequence of the existence of a dual space $\g^{*}$ for a Lie algebra $\g$ is the existence of functional derivatives, which is crucial to proving that the Lie-Poisson bracket in \cref{prop:LP_rev} below is well-defined.

\begin{lemma}[Existence of functional derivatives]
Let $\g$ be a Lie algebra, and let $\g^{*}$ be dual to $\g$ with respect to the nondegenerate pairing $\ip{\cdot}{\cdot}_{\g-\g^*}$. For any functional $F\in C^{1}(\g^{*};\mathbb{F})$, there exists a unique element $\frac{\delta F}{\delta \mu}\in \g$ such that
\begin{equation}
\ip{\frac{\delta F}{\delta \mu}}{\delta\mu}_{\g-\g^*} = dF[\mu](\delta \mu), \qquad \mu, \delta\mu \in \g^{*}.
\end{equation}
\end{lemma}
\begin{proof}
Let $\mu\in \g^{*}$. The G\^ateaux derivative of $F$ at $\mu$ denoted $dF[\mu]$ and defined in \cref{gateaux_deriv} is a continuous linear functional on $\g^{*}$. Hence by the nondegeneracy of the pairing, there exists a unique element $\frac{\delta F}{\delta \mu}\in \g$ such that
\begin{equation*}
\ip{\frac{\delta F}{\delta \mu}}{\delta\mu}_{\g-\g^*} = dF[\mu][\delta\mu], \qquad \delta\mu\in\g^{*}. \qedhere
\end{equation*}
\end{proof}

We now have the necessary ingredients to define the canonical Poisson structure on the dual space $\g^{*}$, which we call the \emph{Lie-Poisson} structure, following Marsden and Weinstein \cite{MW1983}.

\begin{prop}[Lie-Poisson structure]\label{prop:LP_rev}
Let $(\g,\comm{\cdot}{\cdot}_{\g})$ be a Lie algebra, such that the Lie bracket is continuous, and let $\g^{*}$ be dual to $\g$ with respect to the non-degenerate pairing $\ip{\cdot}_{\g-\g^*}$. Define the \emph{Lie-Poisson bracket} 
\begin{equation}
\pb{\cdot}{\cdot} : C^{\infty}(\g^{*};\mathbb{F})\times C^{\infty}(\g^{*};\mathbb{F}) \rightarrow C^{\infty}(\g^{*};\mathbb{F})
\end{equation}
by
\begin{equation}
\pb{F}{G}(\mu) \coloneqq \ip{\comm{\frac{\delta F}{\delta\mu}}{\frac{\delta G}{\delta\mu}}_{\g}}{\mu}_{\g-\g^*}, \qquad \mu \in \g^{*}.
\end{equation}
Then $(C^{\infty}(\g^{*};\mathbb{F}), \pb{\cdot}{\cdot})$ is a Lie algebra.
\end{prop}

\begin{remark}
Note that in the statement of \cref{prop:LP_rev}, we require that the Lie bracket $\comm{\cdot}{\cdot}_{\g}$ be continuous, not merely separately continuous as in \cref{def:la}. Since the Lie brackets we consider in \cref{sec:geom_N} and \cref{sec:geom} are only separately continuous, we do not use \cref{prop:LP_rev} directly, and therefore we have omitted the proof of it. We emphasize, though, that the construction of the proposition inspires our constructions in the sequel.
\end{remark}

\subsection{Bosonic functions, operators and tensor products}\label{sec:bos}

We denote the symmetric group on $k$ letters by $\Ss_{k}$.  For a permutation $\pi\in\Ss_{k}$, we define the map $\pi: \R^{k}\rightarrow \R^{k}$ by
\begin{equation}
\pi(\ul{x}_{k}) \coloneqq (x_{\pi(1)},\ldots,x_{\pi(k)}).
\end{equation}
For a complex-valued, measurable function $f : \R^k \to \C$, we define the map
\begin{equation}\label{eq:pi_func_def}
(\pi f)(\ul{x}_{k}) \coloneqq (f\circ\pi)(\ul{x}_{k}) = f(x_{\pi(1)},\ldots,x_{\pi(k)}).
\end{equation}

We denote the pairing of a tempered distribution $u\in\Sc'(\R^{k})$ with a Schwartz function $f\in\Sc(\R^{k})$ by
\begin{equation}
\ipp{u,f}_{\Sc'(\R^{k})-\Sc(\R^{k})}.
\end{equation}
Throughout, we will use an integral to represent the pairing of a distribution and a test function.  For $1\leq p\leq \infty$, we use the notation $L^{p}(\R^{k})$ to denote Banach space of $p$-integrable functions with norm $\|\cdot\|_{L^{p}(\R^{k})}$. In particular, when $p=2$, we denote the $L^{2}$ inner product by
\begin{equation}
\ip{f}{g} \coloneqq \int_{\R^{k}}d\ux_{k}\ol{f(\ux_{k})} g(\ux_{k}).
\end{equation}
Note that we use the physicist's convention that the inner product is complex linear in the second entry. Similarly, for $u\in\Sc'(\R^{k})$ and $f\in\Sc(\R^{k})$, we use the notation $\ip{u}{f}$ to denote
\begin{equation}
\ip{u}{f} \coloneqq \ol{\ipp{u,\bar{f}}_{\Sc'(\R^{k})-\Sc(\R^{k})}}.
\end{equation}
Alternatively, the right-hand side may be taken as the definition of the tempered distribution $\bar{u}$.

\begin{mydef}
We say that a measurable function $f:\R^{k}\rightarrow \C$ is \emph{symmetric} or \emph{bosonic} if
\begin{equation}
\pi(f) = f
\end{equation}
for all permutations $\pi\in\Ss_{k}$.
\end{mydef}

\begin{mydef}\label{def:sym_f}
We define the \emph{symmetrization operator} $\Sym_k$ on the space of measurable complex-valued functions by
\begin{equation}
\Sym_k(f)(\ul{x}_{k}) \coloneqq \frac{1}{k!}\sum_{\pi\in\Ss_{k}} (\pi f)(\ul{x}_{k}).
\end{equation}
By duality, we can extend the symmetrization operator to $\mathcal{S}'(\R^{k})$.
\end{mydef}

\begin{mydef}[Symmetric Schwartz space]\label{sym_schwartz}
For $k\in\N$, let $\mathcal{S}_{s}(\R^{k})$ denote the subspace of $\mathcal{S}(\R^{k})$ consisting of Schwartz functions $f$ with the property that
\begin{equation}
f(x_{\pi(1)},\ldots,x_{\pi(k)}) = f(\ul{x}_k), \qquad (\ul{x}_k)\in\R^{k}
\end{equation}
for all permutations $\pi\in\Ss_{k}$.
\end{mydef}

\begin{mydef}[Symmetric tempered distribution]
We say that a tempered distribution $u\in\mathcal{S}'(\R^{k})$ is \emph{symmetric} or \emph{bosonic} if for all permutations $\pi\in\Ss_{k}$,
\begin{equation}
\ipp{u,\pi g}_{\Sc'(\R^k)-\Sc(\R^k)} = \ipp{u,g}_{\Sc'(\R^k),\Sc(\R^k)},
\end{equation}
for all $g\in\mathcal{S}(\R^{k})$. We denote the subspace of symmetric tempered distributions by $\mathcal{S}_{s}'(\R^{k}).$
\end{mydef}

\begin{remark}\label{rem:sym}
It is straightforward to check that $\Sym_k$ is a continuous operator $\mathcal{S}(\R^{k}) \rightarrow \mathcal{S}_{s}(\R^{k})$ and $\mathcal{S}'(\R^{k}) \rightarrow \mathcal{S}_{s}'(\R^{k})$. 
%\end{remark}
%
%\begin{remark}
Furthermore, a measurable function $f$ is bosonic if and only if $f=\Sym_k(f)$.
\end{remark}

\begin{lemma}
\label{lem:bos_td}
We have the identification
\begin{equation}
\mathcal{S}_{s}'(\R^{k}) \cong (\mathcal{S}_{s}(\R^{k}))'.
\end{equation}
\end{lemma}
\begin{proof}
Let $\ell \in (\mathcal{S}_{s}(\R^{k}))'$. For all $f\in\mathcal{S}_{s}(\R^{k})$, we have that
\begin{equation}
\ell(f) = \ell(\pi(f)), \qquad \pi\in\Ss_{k}.
\end{equation}
Hence,
\begin{equation}
\ell(f) = \frac{1}{k!}\sum_{\pi\in\Ss_{k}}\ell(\pi(f)) = \ell\paren*{\Sym_k(f)}.
\end{equation}
Since $\Sym_k$ is a continuous linear operator on $\mathcal{S}(\R^{k})$, it follows that $\ell \circ \Sym_k\in \mathcal{S}'(\R^{k})$. Since $\Sym_k(\pi(f)) = \Sym_k(f)$ for any permutation $\pi\in\Ss_{k}$, it follows that $\ell \circ\Sym_k$ is permutation invariant, hence an element of $\mathcal{S}_{s}'(\R^{k})$.
\end{proof}

Given two locally convex spaces $E$ and $F$, we denote the space of continuous linear maps $E\rightarrow F$ by $\L(E;F)$. We topologize $\L(E;F)$ with the topology of bounded convergence.  For our purposes, we will typically have $E,F\in\{\mathcal{S}(\R^{k}), \mathcal{S}_{s}(\R^{k}), \mathcal{S}'(\R^{k}), \mathcal{S}_{s}'(\R^{k})\}$.

\begin{remark}\label{subspace_top}
In the special case where $E=F=\Sc(\R^k)$, we will write $\tl{\L}(\Sc(\R^k),\Sc(\R^k))$ to denote the vector space $\L(\Sc(\R^k),\Sc(\R^k))$ equipped with the subspace topology induced by $\L(\Sc(\R^k),\Sc'(\R^k))$. The same statement holds with the Schwartz space replaced by the bosonic Schwartz space.
\end{remark}

In the case that $E=\Sc(\R^{d})$ and $F=\Sc'(\R^{d})$, the bounded topology is generated by the seminorms
\begin{equation}
\|A\|_{\mathfrak{R}} \coloneqq \sup_{f,g\in\mathfrak{R}} |\ipp{Af,g}_{\Sc'(\R^{d})-\Sc(\R^{d})}|, \qquad \forall A\in \L(\Sc(\R^{d}), \Sc'(\R^{d})),
\end{equation}
where $\mathfrak{R}$ ranges over the bounded subsets of $\Sc(\R^{d})$. An identical statement holds with all spaces replaced by their symmetric counterparts.  We topologize $\mathcal{S}'(\R^{N})$ with the \emph{strong dual topology}, which is the locally convex topology generated by the seminorms of the form
\begin{equation}
\|f\|_{B} \coloneqq \sup_{\varphi\in B} \left|\int_{\R^{N}}d\ul{x}_{N} f(\ul{x}_{N})\varphi(\ul{x}_{N})\right|,
\end{equation}
where $B$ ranges over the family of all bounded subsets of $\mathcal{S}(\R^{N})$. Note that since $\mathcal{S}(\R^{N})$ is a Montel space, bounded subsets are precompact. An identical statement holds with all spaces replaced by their symmetric counterparts.

\begin{mydef}[Symmetric wave functions]\label{sym_wave}
For $k\in\N$, let $L_{s}^{2}(\R^{k})$ denote the subspace of $L^{2}(\R^{k})$ consisting of functions $f$ which are bosonic a.e. 
\end{mydef}

For $A \in \L(\Sc(\R^{k}),\Sc'(\R^{k}))$ and $\tau \in \Ss_k$, we define
\begin{equation}\label{eq:op_coord}
A_{(\tau(1),\ldots, \tau(k))} \coloneqq \tau\circ A \circ \tau^{-1}.
\end{equation}

\begin{mydef}\label{def:sym_A}
 Given $A \in \L(\Sc(\R^{k}),\Sc'(\R^{k}))$, we define its \emph{bosonic symmetrization} $\Sym_k(A)$ by
\begin{equation}\label{eq:sym_repeat}
\Sym_k(A)  \coloneqq \frac{1}{k!} \sum_{\pi\in\Ss_{k}} A_{(\pi(1), \ldots, \pi(k))}.
\end{equation}
\end{mydef}

\begin{mydef}[Bosonic operators]
Let $k\in\N$. We say that an operator $A: \mathcal{S}(\R^{k}) \rightarrow \mathcal{S}'(\R^{k})$ is \emph{bosonic} or \emph{permutation invariant} if $A$ maps $\mathcal{S}_{s}(\R^{k})$ into $\mathcal{S}_{s}'(\R^{k})$.
\end{mydef}

The analogue of \cref{rem:sym} holds for the symmetrization of operators in that symmetrized operators are indeed operators on the bosonic Schwartz space.

\begin{lemma}\label{lem:sym_op_space}
Let $k \in \N$. If $A^{(k)} \in \mathcal{L}(\mathcal{S}(\R^{k}), \mathcal{S}'(\R^{k}))$, then
\begin{equation}
\Sym_k(A^{(k)}) \in \mathcal{L}(\mathcal{S}_{s}(\R^{k}), \mathcal{S}_{s}'(\R^{k})).
\end{equation}
\end{lemma}
\begin{proof}
It suffices to show that for any $k$-particle operator $A^{(k)}\in \mathcal{L}(\mathcal{S}(\R^{k}),\mathcal{S}'(\R^{k}))$ and any permutation $\sigma\in\Ss_{k}$, it holds that
\begin{equation}
\int_{\R^{k}}d\ul{x}_{k}\paren*{\Sym_k(A^{(k)})f}(\ul{x}_{k}) g(\sigma^{-1}(\ul{x}_{k})) = \int_{\R^{k}}d\ul{x}_{k}\paren*{\Sym_k(A^{(k)})f}(\ul{x}_{k})g(\ul{x}_{k})
\end{equation}
for all $f\in\mathcal{S}_{s}(\R^{k})$ and for all $g\in \mathcal{S}(\R^{k})$. To this end, observe that
\begin{align}
&\int_{\R^{k}}d\ul{x}_{k}\paren*{\Sym_k(A^{(k)})f}(\ul{x}_k) g(x_{\sigma^{-1}(1)},\ldots,x_{\sigma^{-1}(k)}) \nonumber\\
&=\int_{\R^{k}}d\ul{x}_{k}\biggl(\frac{1}{k!}\sum_{\pi\in\Ss_{k}} \paren*{A^{(k)}_{(\pi(1),\ldots,\pi(k))} f}(\ul{x}_k)\biggr)g(x_{\sigma^{-1}(1)},\ldots,x_{\sigma^{-1}(k)}).
\end{align}
By definition \eqref{eq:op_coord}, we have
\begin{align}
A^{(k)}_{(\pi(1),\ldots,\pi(k))}f = \pi A^{(k)}(\pi^{-1} f).
\end{align}
Therefore,
\begin{align}
&\frac{1}{k!}\sum_{\pi\in\Ss_{k}} \int_{\R^{k}}d\ul{x}_{k} \paren*{A^{(k)}_{(1,\ldots,k)} (\pi^{-1} f)}(x_{\pi(1)},\ldots,x_{\pi(k)}) g(x_{\sigma^{-1}(1)},\ldots,x_{\sigma^{-1}(k)}) \nonumber\\
&=\frac{1}{k!}\sum_{\pi\in\Ss_{k}} \int_{\R^{k}}d\ul{x}_{k}\paren*{A^{(k)}(\pi^{-1} f)}(\ul{x}_k) g(x_{\pi^{-1}\sigma^{-1}(1)},\ldots,x_{\pi^{-1}\sigma^{-1}(k)}) \nonumber\\
&=\frac{1}{k!}\sum_{\pi\in\Ss_{k}}\int_{\R^{k}}d\ul{x}_{k}\paren*{A^{(k)} f}(\ul{x}_k)g(x_{\pi^{-1}\sigma^{-1}(1)},\ldots,x_{\pi^{-1}\sigma^{-1}(k)}),
\end{align}
where, recalling \eqref{eq:pi_func_def}, the second line follows from a change of variable and the third line follows from the assumption that $f$ is symmetric with respect to permutation of the coordinates. Since for any fixed $\sigma \in  \Ss_{k}$, $\pi \mapsto \pi^{-1}\sigma^{-1}$ defines a bijection of the group $\Ss_{k}$, it follows from a change of summation index that
\begin{align}
&\frac{1}{k!}\sum_{\pi\in\Ss_{k}}\int_{\R^{k}}d\ul{x}_{k}\paren*{A^{(k)} f}(\ul{x}_k)g(x_{\pi^{-1}\sigma^{-1}(1)},\ldots,x_{\pi^{-1}\sigma^{-1}(k)}) \nonumber\\
&= \frac{1}{k!}\sum_{\tilde{\pi}\in\Ss_{k}} \int_{\R^{k}}d\ul{x}_{k} \paren*{A^{(k)}f}(\ul{x}_k) g(x_{\tilde{\pi}(1)},\ldots,x_{\tilde{\pi}(k)}) \nonumber\\\
&=\frac{1}{k!}\sum_{\tilde{\pi}\in \Ss_{k}}\int_{\R^{k}}d\ul{x}_{k} \paren*{A^{(k)}(\tilde{\pi}f)}(x_{\tilde{\pi}^{-1}(1)},\ldots,x_{\tilde{\pi}^{-1}(k)})g(\ul{x}_k) \nonumber\\
&=\int_{\R^{k}}d\ul{x}_{k}\paren*{\Sym_k(A^{(k)})f}(\ul{x}_{k})g(\ul{x}_{k}),
\end{align}
where the penultimate line follows from the assumption that $f$ is symmetric and a change of variable. This concludes the proof.
\end{proof}

The following technical lemma will be useful in the sequel. For definitions and discussion of the generalized trace, see \cref{def:gen_trace}.

\begin{lemma}\label{lem:tr_bos}
Let $k\in\N$, and let $\gamma^{(k)}\in\L(\Sc_{s}'(\R^{k}),\Sc_{s}(\R^{k}))$ and $A^{(k)}\in\L(\Sc(\R^{k}),\Sc'(\R^{k}))$. Then for any permutation $\tau\in\Ss_{k}$, we have that
\begin{equation}
\Tr_{1,\ldots,k}\paren*{A_{(\tau(1),\ldots,\tau(k))}^{(k)}\gamma^{(k)}} = \Tr_{1,\ldots,k}\paren*{A^{(k)}\gamma^{(k)}}.
\end{equation}
\end{lemma}
\begin{proof}
Let $\tau\in\Ss_{k}$. Now let
\begin{equation}
\gamma^{(k)} = \sum_{j=1}^{\infty}\lambda_{j}  \ket*{f_{j}}\bra*{g_j}
\end{equation}
be a decomposition for $\gamma^{(k)}$, where $\sum_{j=1}^{\infty}|\lambda_{j}| \leq 1$, and $\{f_{j}\}_{j=1}^{\infty},\{g_{j}\}_{j=1}^{\infty}$ are sequences tending to zero in $\Sc_{s}(\R^{k})$. In particular, the partial sums
\begin{equation}
\sum_{j=1}^{N}\lambda_{j}  \ket*{f_{j}}\bra*{g_j} \xrightarrow[N\rightarrow\infty]{} \gamma^{(k)} \enspace \text{in $\L(\Sc_{s}'(\R^{k}),\Sc_{s}(\R^{k}))$}.
\end{equation}
Since the map
\begin{equation}
\Tr_{1,\ldots,k}\paren*{A_{(\tau(1),\ldots,\tau(k)}^{(k)}\cdot } :  \L(\Sc'(\R^{k}),\Sc(\R^k)) \rightarrow \C,
\end{equation}
is continuous and the inclusion $\Sc_{s}(\R^{k}) \subset \Sc(\R^{k})$ is trivially continuous, it follows that
\begin{align}
\Tr_{1,\ldots,k}\paren*{A_{(\tau(1),\ldots,\tau(k))}^{(k)}\gamma^{(k)}} &=  \lim_{N\rightarrow\infty} \Tr_{1,\ldots,k}\biggl(A_{(\tau(1),\ldots,\tau(k))}^{(k)}\biggl(\sum_{j=1}^{N}\lambda_{j} \ket*{f_j}\bra*{g_j}\biggr)\biggr) \nonumber\\
&= \lim_{N\rightarrow\infty} \sum_{j=1}^{N} \lambda_{j}\Tr_{1,\ldots,k}\paren*{A_{(\tau(1),\ldots,\tau(k))}^{(k)}\paren*{\ket*{f_j}\bra*{g_j}}}  \nonumber\\
&=\lim_{N\rightarrow\infty} \sum_{j=1}^{N}\lambda_{j} \ip{g_{j}}{A_{(\tau(1),\ldots,\tau(k))}^{(k)}f_{j}}.
\end{align}
Since $f_{j}$ and $g_{j}$ are both bosonic, we have by definition of the notation $A_{(\tau(1),\ldots,\tau(k))}^{(k)}$ in \eqref{eq:op_coord} that
\begin{align}
\ip{g_{j}}{A_{(\tau(1),\ldots,\tau(k))}^{(k)}f_{j}} = \ip{\tau^{-1}(g_{j})}{A^{(k)}(\tau^{-1}(f_{j}))} = \ip{g_{j}}{A^{(k)}f_{j}}, \qquad \forall j\in\N.
\end{align}
Therefore,
\begin{align}
\lim_{N\rightarrow\infty}\sum_{j=1}^{N} \lambda_{j} \ip{g_{j}}{A_{(\tau(1),\ldots,\tau(k))}^{(k)}f_{j}} &= \lim_{N\rightarrow\infty}\sum_{j=1}^{N} \lambda_{j}\ip{g_{j}}{A^{(k)}f_{j}} \nonumber\\
&= \lim_{N\rightarrow\infty} \Tr_{1,\ldots,k}\biggl(A^{(k)}\biggl(\sum_{j=1}^{N}\lambda_{j}\ket*{f_j}\bra*{g_j}\biggr)\biggr)  \nonumber\\
&=\Tr_{1,\ldots,k}\paren*{A^{(k)}\gamma^{(k)}},
\end{align}
where in order to obtain the ultimate equality, we again use the continuity of the functional $\Tr_{1,\ldots,k}\paren*{A^{(k)}\cdot}$ and the convergence of the partial sums.
\end{proof}

We define the usual contraction operator $B_{i;j}$ appearing in the literature on derivation of quantum many-body systems.

\begin{mydef}[The contractions operator $B_{i;j}$]\label{contraction}
Let $k\in\N$. For integers $1\leq i,j\leq k$ with $i\neq j$, we define the continuous linear operators operators
\begin{equation}
B_{i;j}^{\pm} : \mathcal{L}(\mathcal{S}'(\R^{k+1}), \mathcal{S}(\R^{k+1})) \rightarrow  \mathcal{L}(\mathcal{S}'(\R^{k}), \mathcal{S}(\R^{k}))
\end{equation}
by defining the Schwartz kernel of $B_{i;j}^{+}(\gamma^{(k+1)})$ by the formula
\[
B_{i;j}^{+}(\gamma^{(k+1)})(\ul{x}_{k}; \ul{x}_{k}') \coloneqq \int_{\R}dy  \delta(x_i - y)\gamma^{(k+1)}(\ul{x}_{1;j-1},y,\ul{x}_{j;k}; 
\ul{x}_{1;j-1}',y,\ul{x}_{j;k}'), 
\]
for all $(\ul{x}_{k},\ul{x}_{k}')\in\R^{2k}$. Similarly, we define the Schwartz kernel of $B_{i;j}^{-}(\gamma^{(k+1)})$ by the formula
\[
B_{i;j}^{-}(\gamma^{(k+1)})(\ul{x}_{k};\ul{x}_{k}') \coloneqq \int_{\R}dy\delta(x_i' - y) \gamma^{(k+1)}(\ul{x}_{1;j-1},y,\ul{x}_{j;k}; 
\ul{x}_{1;j-1}',y,\ul{x}_{j;k}'), 
\]
for all $(\ul{x}_{k},\ul{x}_{k}')\in\R^{2k}$ We define the continuous linear operator 
\[
B_{i;j}:\mathcal{L}(\mathcal{S}_{s}'(\R^{k+1}), \mathcal{S}_{s}(\R^{k+1})) \rightarrow \mathcal{L}(\mathcal{S}_{s}'(\R^{k}), \mathcal{S}_{s}(\R^{k}))
\]
by
\begin{equation}
B_{i;j} \coloneqq B_{i;j}^{+} - B_{i;j}^{-}.
\end{equation}
\end{mydef}

Given two locally convex spaces $E$ and $F$, we denote an\footnote{ The reader will recall that the algebraic tensor product is only defined up to unique isomorphism.} algebraic tensor product of $E$ and $F$ consisting of finite linear combinations
\begin{equation}
\sum_{j=1}^{n} \lambda _{j} e_{j} \otimes f_{j}, \qquad e_{j}\in E, \enspace f_{j}\in F
\end{equation}
by $E\otimes F$. 
We note that since the spaces we deal with in this paper are nuclear, the topologies of the injective and projective tensor products coincide. Hence, we can unambiguously write $E\hat{\otimes} F$ to denote the completion of $E\otimes F$ under either of the aforementioned topologies.

Given locally convex spaces $E_{j}$ and $F_{j}$ for $j=1,2$ and linear maps $T:E_{1}\rightarrow E_{2}$ and $S:F_{1}\rightarrow F_{2}$, and a tensor product
\begin{equation}
B:E_{1}\times E_{2} \rightarrow E_{1}\otimes E_{2},
\end{equation}
the notation $T\otimes S$ denotes the unique linear map $T\otimes S:E_{1}\otimes F_{1} \rightarrow E_{2}\times F_{2}$ such that
\begin{equation}
(T\otimes S) \circ B = T\times S.
\end{equation}
Note that the existence of such a unique map is guaranteed by the universal property of the tensor product.

When $E$ and $F$ are subspaces of measurable functions on $\R^{m}$ and $\R^{n}$  respectively, and $e \in E$ and $f \in F$, we let $e\otimes f$ denote the function
\begin{equation}
e\otimes f : \R^{m}\times\R^{n} \rightarrow \C, \qquad (e\otimes f)(\ux_{m};\ux_{n}') \coloneqq e(\ux_{m})f(\ux_{n}'),
\end{equation}
which induces a bilinear map $E\times F\rightarrow E\otimes F$.  Similarly, if $E'$ and $F'$ are the duals of spaces of test functions $E$ and $F$, for instance $E'=\D'(\R^{m})$ and $F'=\D'(\R^{n})$, we let $u\otimes v$ denote the unique distribution satisfying 
\begin{equation}
(u\otimes v)(e\otimes f) = u(e) \cdot v(f).
\end{equation}
Finally, if $\phi:\R^{m}\rightarrow\C$ is a measurable function, we use the notation $\phi^{\otimes k}$, for $k\in\N$, to denote the measurable function $\phi^{\otimes k}: \R^{mk}\rightarrow\C$ defined by
\begin{equation}\label{tensor_def}
\phi^{\otimes k}(\ux_{m,1},\ldots,\ux_{m,k}) \coloneqq \prod_{\ell=1}^{k} \phi(\ux_{m,\ell}).
\end{equation}

\section{Geometric structure for the $N$-body problem}\label{sec:conv_n_body}
\label{sec:geom_N}
In this section we establish proofs of the results stated in \cref{sym_n_pois}.

\subsection{Lie algebra $\G_{N}$ of finite hierarchies quantum observables}
\label{ssec:N_LA}
We begin by defining a Lie algebra $\g_k$ of $k$-body observables. We have some freedom to choose our definition of this Lie algebra, provided that our choice is large enough to include the Hamiltonian of the $N$-body problem yet small enough so that operations such as composition and taking adjoints are well-defined. We find that continuous linear maps from the bosonic Schwartz space to itself forms a convenient choice.

For $k\in\N$, define
\begin{equation}\label{gk_def}
\g_{k} \coloneqq \{ A^{(k)} \in \tl{\L}(\Sc_s(\R^k),\Sc_s(\R^k)) : (A^{(k)})^* = -A^{(k)} \},
\end{equation}
where we recall that $\tl{\L}(\Sc_s(\R^k), \Sc_s'(\R^k))$ is defined in \cref{subspace_top}. Let 
\[
\comm{\cdot}{\cdot}_{\g_{k}}: \g_{k} \times \g_{k} \rightarrow \g_{k}
\]
be the usual commutator bracket scaled by a factor of $k$:
\begin{equation}
\comm{A}{B}_{\g_{k}} \coloneqq k\comm{A}{B}= k(AB - BA).
\end{equation}
Note that the commutator is well-defined since the space $\mathcal{L}(\Sc_s(\R^k),\Sc_s(\R^k))$ is closed under composition. We refer to the elements of $\g_k$ as \emph{$k$-body observables}.

The first goal of this subsection is to verify that $(\g_k, \comm{\cdot}{\cdot}_{\g_k})$ is a Lie algebra in the sense of \cref{def:la}. Namely, we prove the following proposition.

\begin{prop}
\label{prop:DM_LA}
$(\g_{k}, \comm{\cdot}{\cdot}_{\g_{k}})$ is a Lie algebra in the sense of \cref{def:la}
\end{prop}
\begin{proof}
That $\comm{\cdot}{\cdot}_{\g_k}$ is algebraically a Lie bracket is immediate from the fact that the commutator satisfies properties \ref{item:LA_1}, \ref{item:LA_2}, and \ref{item:LA_3}. Therefore, it remains to verify that the commutator is separately continuous with respect to the topology on $\g_k$. By symmetry, it suffices to show that for fixed $A^{(k)}\in \g_k$, the map $B^{(k)}\mapsto A^{(k)}B^{(k)}$ is continuous on $\tl{\L}(\Sc_s(\R^k),\Sc_s(\R^k))$, which amounts to showing that for any bounded subset $\mathfrak{R}\subset \Sc_s(\R^k)$, there exists a bounded subset $\tilde{\mathfrak{R}}\subset \Sc_s(\R^k)$, such that
\begin{equation}
\sup_{f,g\in \mathfrak{R}} \left|\ip{g}{A^{(k)}B^{(k)}f}\right| \lesssim \sup_{f,g\in\tilde{\mathfrak{R}}} \left|\ip{g}{B^{(k)}f}\right|.
\end{equation}
Now note that $\ip{g}{A^{(k)}B^{(k)}f} = \ip{(A^{(k)})^*g}{B^{(k)}f}$. Since $(A^{(k)})^*=-A^{(k)}$, it follows from the continuity of $A^{(k)}$ that $(A^{(k)})^*(\mathfrak{R})$ it a bounded subset of $\Sc_s(\R^k)$. Choosing $\tilde{\mathfrak{R}} = \mathfrak{R} \cup (A^{(k)})^{*}(\mathfrak{R})$ completes the proof.
\end{proof}

We next introduce some combinatorial notation used frequently in the sequel. For $N \in \N$ and $k\in \N_{\leq N}$, let $P_{k}^{N}$ denote the collection of $k$-tuples $(j_{1},\ldots,j_{k})$ with $k$ distinct elements drawn from the set $\N_{\leq N}$. Given an element $(j_{1},\ldots,j_{k})\in P_{k}^{N}$, let $(m_1,\ldots,m_{N-k})$ denote the increasing arrangement of $\N_{\leq N}\setminus\{j_1,\ldots,j_k\}$. We denote by $\pi_{j_{1}\cdots j_{k}}\in \Ss_{N}$ the permutation
\begin{equation}
\label{f_perm}
\pi(a) \coloneqq
\begin{cases}
i, & {a=j_i \text{ for $i\in\N_{\leq k}$}}\\
k+i, & {a=m_i \text{ for $i\in\N_{\leq N-k}$}}
\end{cases}.
\end{equation}

Our first lemma defines a continuous linear map $\epsilon_{k,N}$ which allows us to regard a $k$-particle observable as an $N$-particle observable. This map $\epsilon_{k,N}$ is crucial to the definition of the Lie bracket between two observable $N$-hierarchies and by duality, to the Poisson bracket of two density matrix $N$-hierarchies.

For $A^{(k)} \in \L(\Sc_s(\R^k),\Sc_s(\R^k))$, $N \in \N$ with $1 \leq k \leq N$, and $(j_{1},\ldots,j_{k})\in P_{k}^{N}$ we can define the operator
\begin{align}\label{k_part}
A_{(j_1,\ldots,j_k)}^{(k)} \in  \L(\Sc_s(\R^N),\Sc(\R^N))
\end{align}
which acts only on the variables $\{j_1, \ldots, j_k\}$ by defining
\[
A_{(1,\ldots,k)}^{(k)} = A^{(k)} \otimes Id_{N-k}
\]
and setting
\begin{equation}
\label{eq:sym_obs}
A_{(j_1,\ldots,j_k)}^{(k)} = \pi_{j_1\cdots j_k}^{-1} \circ A_{(1,\ldots,k)}^{(k)} \circ \pi_{j_1\cdots j_k}.
\end{equation}
We establish some properties of such operators, which we call $k$-particle extensions, in \cref{prop:ext_k}. These $k$-particle extensions are used to define a map $\epsilon_{k,N}$. We will show first, in the following lemma, that $\varepsilon_{k,N}$ have the desired mapping properties, and then subsequently that the $\epsilon_{k,N}$ are injective, and hence they are proper embeddings of the space $\g_{k}$ into $\g_{N}$.

\begin{remark}
Although $A^{(k)}$ is a priori only defined on the proper subspace $\Sc_s(\R^k)\subset \Sc(\R^k)$, this operator admits an extension to the space $\mathcal{S}(\mathbb{R}^k)$ since we may always consider $A^{(k)}\circ \Sym_k$. We agree going forward to abuse notation by identifying $A^{(k)}$ with this extension. Consequently, we may regard $A_{(j_1,\ldots,j_k)}^{(k)}\in\L(\Sc(\R^N),\Sc(\R^N))$. As the reader will see, though, all our constructions are independent of the choice of extension.
\end{remark}

\begin{lemma}
\label{lem:ep_con}
For integers $1\leq k\leq N$, there is a continuous linear map
\begin{equation}
\epsilon_{k,N}:\mathcal{L}(\Sc_{s}(\mathbb{R}^{k}), \Sc_{s}'(\mathbb{R}^{k})) \rightarrow \mathcal{L}(\Sc_{s}(\mathbb{R}^{N}),\Sc_s'(\R^N))
\end{equation}
defined by
\begin{equation}\label{eps_def}
\epsilon_{k,N}(A^{(k)}) \coloneqq C_{k,N}\sum_{(j_1,\ldots,j_k)\in P_k^N} A_{(j_1,\ldots,j_k)}^{(k)},
\end{equation}
where
\begin{equation}
\label{eq:Ckn_def}
C_{k,N} \coloneqq \paren*{k!{N\choose k}}^{-1}=\frac{1}{N\cdots (N-k+1)}. \footnote{ Note that $C_{k,N}=1/|P_k^N|$.}
\end{equation}
Moreover, if $A^{(k)}\in \L(\Sc_s(\R^k),\Sc_s(\R^k))$, then $\epsilon_{k,N}(A^{(k)}) \in \L(\Sc_s(\R^N),\Sc_s(\R^N))$, and if $A^{(k)}$ is skew-adjoint, then $\epsilon_{k,N}(A^{(k)})$ is skew-adjoint. In particular, $\epsilon_{k,N}({\g}_k) \subset {\g}_N$.
\end{lemma}

\begin{proof}
Fix $1\leq k\leq N$. From \cref{prop:ext_k}, it follows that if $A^{(k)}\in \mathcal{L}(\Sc_s(\R^k),\Sc_s'(\R^k))$, then $\epsilon_{k,N}(A^{(k)})$ as given in \eqref{eps_def} is a well-defined element of $\mathcal{L}(\Sc_s(\R^N),\Sc_s'(\R^N))$ and the map $\epsilon_{k,N}$ is linear. Furthermore, it follows from \cref{lem:ext_sa} that skew-adjointness is preserved. So it remains for us to show that
\begin{equation}\label{contain}
\ep_{k,N}(\mathcal{L}(\Sc_s(\R^k),\Sc_s(\R^k))) \subset \L(\Sc_s(\R^N),\Sc_s(\R^N))
\end{equation}
and that $\epsilon_{k,N}$ is continuous.
\begin{itemize}[leftmargin=*]
\item
Consider the assertion \eqref{contain}. By properties of tensor product and the continuity of $A^{(k)}$, it follows that  $A_{(1,\ldots,k)}^{(k)} = A^{(k)} \hat{\otimes} Id_{N-k}$ is a continuous map of $\Sc_s(\R^k)\otimes \Sc(\R^{N-k})$ to itself, and hence that
\[
A_{(j_1,\ldots,j_k)}^{(k)}: \Sc_s(\R^N) \rightarrow \Sc(\R^N)
\]
is a continuous map follows directly from \eqref{eq:sym_obs}. We thus need to show that $\ep_{k,N}(A^{(k)})(f)$ is bosonic. 

Let $\pi\in\mathbb{S}_N$. It is straightforward from the definition of $A_{(j_1,\ldots,j_k)}^{(k)}$ and \cref{eq:pi_func_def} that, for any test function $f\in \Sc_s(\R^N)$, we have
\begin{equation}\label{pi_Ak}
\pi A_{(j_1,\ldots,j_k)}^{(k)}(f) = A_{(\pi(j_1),\ldots,\pi(j_k))}^{(k)}(\pi f) = A_{(\pi(j_1),\ldots,\pi(j_k))}^{(k)}(f),
\end{equation}
where the ultimate equality follows from $f$ being bosonic. Since $\mathbb{S}_N$ induces a left group action on $P_k^N$, it follows that
\begin{equation}
\sum_{(j_1,\ldots,j_k)\in P_k^N} A_{(j_1,\ldots,j_k)}^{(k)} = \sum_{(j_1,\ldots,j_k)\in P_k^N} A_{(\pi(j_1),\ldots,\pi(j_k))}^{(k)}
\end{equation}
on $\Sc_s(\R^k)$, which implies together with \eqref{pi_Ak} that
\begin{equation}
\pi \ep_{k,N}(A^{(k)})(f) = C_{k,N}\sum_{(j_1,\ldots,j_k)\in P_k^N} \pi A_{(j_1,\ldots,j_k)}^{(k)}(f) = \ep_{k,N}(A^{(k)})(f),
\end{equation}
as desired.
\item
Now we will prove the assertion that $\ep_{k,N}$ is continuous. Let $\mathfrak{R}_N$ be a bounded subset of $\Sc_s(\R^N)$. We need to show that there exists a bounded subset $\mathfrak{R}_k\subset \Sc_s(\R^k)$ such that
\begin{equation}
\sup_{f^{(N)},g^{(N)} \in \mathfrak{R}_N} \left|\ip{g^{(N)}}{\ep_{k,N}(A^{(k)})f^{(N)}}\right| \lesssim \sup_{f^{(k)},g^{(k)}\in\mathfrak{R}_k} \left|\ip{g^{(k)}}{A^{(k)}f^{(k)}}\right|.
\end{equation}
Using the fact that there are finitely many terms in the definition of $\ep_{k,N}$ and that the finite union of bounded subsets is again a bounded subset, it suffices to show that, for $\mathfrak{R}_N$ as above and any tuple $(j_1,\ldots,j_k)\in P_k^N$, there exists a bounded subset $\mathfrak{R}_{(j_1,\ldots,j_k)} \subset \Sc(\R^k)$, such that
\begin{equation}
\label{eq:em_cont_goal}
\sup_{f^{(N)},g^{(N)} \in\mathfrak{R}_N} \left|\ip{g^{(N)}}{A_{(j_1,\ldots,j_k)}^{(k)}f^{(N)}}\right| \lesssim \sup_{f^{(k)}, g^{(k)}\in \mathfrak{R}_{(j_1,\ldots,j_k)}} \left|\ip{g^{(k)}}{A^{(k)}f^{(k)}}\right|,
\end{equation}
since then the desired bounded subset $\mathfrak{R}_k \subset \Sc_s(\R^k)$ is obtained by taking
\[
\mathfrak{R}_k \coloneqq \Sym_k\paren*{\bigcup_{\ul{j}_k\in P_k^N} \mathfrak{R}_{(j_1,\ldots,j_k)}}.
\]
Now \eqref{eq:em_cont_goal} is a consequence of the fact that
\begin{equation}
\L(\Sc_s(\R^k),\Sc_s'(\R^k)) \mapsto \L(\Sc_s(\R^k)\hat{\otimes} \Sc(\R^{N-k}),\Sc'(\R^N)), \qquad A^{(k)} \mapsto A^{(k)} \otimes Id_{N-k}
\end{equation}
is continuous, \eqref{eq:sym_obs}, and the fact that for any $\ul{j}_k\in P_k^N$, the map $\pi_{j_1\ldots j_k}$ defined by \eqref{f_perm} and duality is a continuous endomorphism of $\Sc'(\R^N)$.
\end{itemize}
\end{proof}

We next show that the maps $\epsilon_{k,N}$ are injective. This property is crucial as we will ultimately construct our Lie bracket on the hierarchy algebra by embedding elements of the sequence into the ambient algebra $\g_{N}$, taking the bracket in $\g_{N}$, and then identifying the output as an embedded element of $\g_k$, for some $k\in\N_{\leq N}$.

\begin{lemma}[Injectivity of $\epsilon_{k,N}$]
For integers $1\leq k\leq N$, the map $\epsilon_{k,N}:\g_{k} \rightarrow \g_{N}$ is injective. Consequently, $\ep_{k,N}$ has a well-defined inverse on its image, which we denote by $\ep_{k,N}^{-1}$.
\end{lemma}
\begin{proof}
Fix $1\leq k\leq N$. We will show the contrapositive statement: if $A^{(k)}\neq 0$, then $\epsilon_{k,N}(A^{(k)})\neq 0$. 

We introduce a parameter $n \in \N_0$, with $n < k$. We say that $A^{(k)}$ has \emph{property $\mathbf{P}_n$} if the following holds: there exists $f , g_1, \ldots, g_{k-n} \in \Sc(\R)$ such that 
\begin{equation}
A^{(k)}\biggl(\Sym_k\biggl(f^{\otimes k-n} \otimes \bigotimes_{a=1}^n g_a\biggr)\biggr) \neq 0,
\end{equation}
where the tensor product is understood as vacuous when $n=0$. We define the integer $n_{\min}$ by
\begin{equation}
n_{\min} \coloneqq \max\{\min\{n\in \N_{< k} : \text{$A^{(k)}$ has property $P_n$}\},k\}.\footnote{We adopt the convention that the minimum of the empty set is $\infty$, and therefore we take the maximum with $k$ to ensure that $n_{\min}$ is finite.}
\end{equation}
Note that we must have $n_{\min}<k$, else, by definition of property $\mathbf{P}_n$, we would then have that for all $g_1,\ldots,g_k\in \Sc(\R)$,
\begin{equation}
\label{eq:ld_imp}
A^{(k)}\paren*{\Sym_k(g_1\otimes \cdots \otimes g_k)} = 0.
\end{equation}
By linearity and continuity of $A^{(k)}$ together with density of finite linear combinations of symmetric pure tensors in $\Sc_s(\R^k)$, \eqref{eq:ld_imp} implies that $A^{(k)} \equiv 0$, which is a contradiction.

To avoid notation confusion, we first dispense with the trivial case $n_{\min}=0$. The definition of property $\mathbf{P}_{0}$ implies that there exists an element $f\in\Sc(\R)$ such that $A^{(k)}(f^{\otimes k})\neq 0$. It then follows trivially from the definition of each summand $A_{(j_1,\ldots,j_k)}^{(k)}$ in the definition of $\ep_{k,N}(A^{(k)})$ that
\begin{equation}
\ep_{k,N}(A^{(k)})(f^{\otimes N}) \neq 0 \in \Sc_s'(\R^N).
\end{equation}

We now consider the case $1\leq n_{\min}<k$. The definition of property $\mathbf{P}_{n_{\min}}$ implies that there exist elements $f,g_1,\ldots,g_{n_{\min}}\in\Sc(\R)$ such that
\begin{equation}
A^{(k)}\biggl(\Sym_k\biggl(f^{\otimes k-n_{\min}}\otimes \bigotimes_{a=1}^{n_{\min}} g_a\biggr)\biggr) \neq 0 \in\Sc_s'(\R^k).
\end{equation}
Define an element $h^{(N)}\in\Sc_s(\R^N)$ by
\begin{equation}
h^{(N)} \coloneqq \Sym_N\biggl(f^{\otimes k-n_{\min}} \otimes (\bigotimes_{a=1}^{n_{\min}} g_a) \otimes f^{\otimes N-k}\biggr).
\end{equation}
We claim that $\ep_{k,N}(A^{(k)})(h^{(N)}) \neq 0\in\Sc_s'(\R^N)$. Indeed, unpacking the definition of $\ep_{k,N}(A^{(k)})$ and $\Sym_N$, we have
\begin{align}
\ep_{k,N}(A^{(k)})(h^{(N)}) = C_{k,N}\sum_{\ul{j}_k\in P_k^N} A_{(j_1,\ldots,j_k)}^{(k)}\biggl(\sum_{\pi\in \mathbb{S}_N} \pi(f^{\otimes k-n_{\min}}\otimes (\bigotimes_{a=1}^{n_{\min}} g_a) \otimes f^{\otimes N-k})\biggr).
\end{align}
We first examine the interior sum. For each $\ul{j}_k\in P_k^N$, we can partition $\mathbb{S}_N$ into the sets
\begin{equation}
\mathbb{S}_{\ul{j}_k, r} \coloneqq \{\pi \in \mathbb{S}_N : |\{\pi(k-n_{\min}+1),\ldots,\pi(k)\} \cap \{j_1,\ldots,j_k\}| = r\}
\end{equation}
for $r=0,\ldots,n_{\min}$. We write
\begin{equation}
\sum_{\pi\in \mathbb{S}_N} \pi(f^{\otimes k-n_{\min}}\otimes (\bigotimes_{a=1}^{n_{\min}} g_a) \otimes f^{\otimes N-k}) = \sum_{r=0}^{n_{\min}} \sum_{\pi\in \mathbb{S}_{\ul{j}_k, r}} \pi(f^{\otimes k-n_{\min}}\otimes (\bigotimes_{a=1}^{n_{\min}}g_a) \otimes f^{\otimes N-k}).
\end{equation}
By symmetry considerations, we may suppose that $(j_1,\ldots,j_k) = (1,\ldots,k)$. It is a short counting argument that for each $r\in\{0,\ldots,n_{\min}\}$, we have that
\begin{equation}
\begin{split}
&\sum_{\pi \in\mathbb{S}_{(1,\ldots,k),r}} \pi(f^{\otimes k-n_{\min}} \otimes (\bigotimes_{a=1}^{n_{\min}} g_a) \otimes f^{\otimes N-k})\\
&=C(k,n_{\min},r,N) \sum_{\ul{\ell}_{n_{\min}} \in P_{n_{\min}}^{n_{\min}}} \Sym_k\paren*{f^{\otimes k-r} \otimes \bigotimes_{a=1}^r g_{\ell_a} } \otimes \Sym_{N-k}\paren*{(\bigotimes_{a=r+1}^{n_{\min}} g_{\ell_a}) \otimes f^{\otimes N-n_{\min}-k+r}},
\end{split}
\end{equation}
where $C(k,n_{\min},r,N)$ is another combinatorial factor depending on the data $(k,n_{\min},r,N)$. Each term 
\begin{align}\label{rhs}
\Sym_k\paren*{f^{\otimes k-r} \otimes \bigotimes_{a=1}^r g_{\ell_a} } \otimes \Sym_{N-k}\paren*{(\bigotimes_{a=r+1}^{n_{\min}} g_{\ell_a}) \otimes f^{\otimes N-n_{\min}-k+r}}
\end{align}
is an element of $\Sc_s(\R^k)\hat{\otimes}\Sc_s(\R^{N-k})$, and therefore \eqref{rhs} belongs to the domain of $A_{(1,\ldots,k)}^{(k)}$. Now by definition of $n_{\min}$, we have that for each $r\in \{0,\ldots,n_{\min}-1\}$ that
\begin{align*}
&A_{(1,\ldots,k)}^{(k)}\paren*{\Sym_k\paren*{f^{\otimes k-r} \otimes \bigotimes_{a=1}^r g_{\ell_a} } \otimes \Sym_{N-k}\paren*{(\bigotimes_{a=r+1}^{n_{\min}} g_{\ell_a}) \otimes f^{\otimes N-n_{\min}-k+r}}} \\
&= A^{(k)}\biggl(\Sym_k(f^{\otimes k-r} \otimes \bigotimes_{a=1}^r g_{\ell_a})\biggr)\otimes \Sym_{N-k}\paren*{(\bigotimes_{a=r+1}^{n_{\min}} g_{\ell_a}) \otimes f^{\otimes N-n_{\min}-k+r}} \\
&=0\in \Sc_s'(\R^k)\hat{\otimes}\Sc_s(\R^{N-k}).
\end{align*}
When $r = n_{\min}$, we have that
\begin{align*}
&A_{(1,\ldots,k)}^{(k)}\biggl(\Sym_k(f^{\otimes k-n_{\min}} \otimes \bigotimes_{a=1}^{n_{\min}}g_{\ell_a}) \otimes f^{\otimes N-k})\biggr) \\
&= A^{(k)}\biggl(\Sym_k(f^{\otimes k-n_{\min}} \otimes \bigotimes_{a=1}^{n_{\min}}g_a)\biggr) \otimes f^{\otimes N-k}
\end{align*}
is a non-zero element of $\Sc_s'(\R^k)\hat{\otimes} \Sc_s(\R^{N-k})$ by choice of the elements $f,g_1,\ldots,g_{n_{\min}}\in\Sc(\R)$. Consequently, for a possibly different combinatorial factor $C'(k,N)$, we conclude that
\begin{equation}
\begin{split}
\ep_{k,N}(A^{(k)})(h^{(N)}) &=C(k,N)' \Sym_N\paren*{A^{(k)}\biggl(\Sym_k(f^{\otimes k-n_{\min}} \otimes \bigotimes_{a=1}^{n_{\min}}g_a)\biggr) \otimes f^{\otimes N-k}}
\end{split}
\end{equation}
is a nonzero element of $\Sc_s'(\R^N)$, completing the proof of the lemma.
\end{proof}

We next show that the bracket $\comm{\cdot}{\cdot}_{\g_N}$ respects the hierarchy in the sense that
\begin{equation}
\comm{\epsilon_{\ell,N}(\g_{\ell})}{ \epsilon_{j,N}(\g_{j})}_{\g_N} \subset \epsilon_{\min\{\ell+j-1, N\},N}(\g_{\min\{\ell+j-1,N\}}) \subset \g_{N}.
\end{equation}
This filtration or gradation property is crucial to our definition of the hierarchy Lie bracket in the sequel.

Before proving \cref{lem:hi_fil} below, we introduce some contraction and commutator-type notation used in the proof and in the sequel. Consider integers $N\in\N$, $ \ell,j \in N_{\leq N}$, $k\coloneqq \min\{\ell+j-1,N\}$ and $r\geq 1$ satisfying appropriate conditions. Let $A^{(\ell)}\in \L(\Sc_s(\R^\ell),\Sc_s(\R^\ell))$ and $B^{(j)}\in \L(\Sc_s(\R^j),\Sc_s(\R^j))$. We define the \emph{$r$-fold contractions}
\begin{align}
A^{(\ell)} \circ_r B^{(j)} &\coloneqq A_{(1,\ldots,\ell)}^{(\ell)}\biggl(\sum_{\ul{\alpha}_r\in P_r^\ell} B_{(\ul{\alpha}_r,\ell+1,\ldots,\ell+j-r)}^{(j)}\biggr) \in \L(\Sc_s(\R^k),\Sc'(\R^k)) \label{eq:A_c_B}\\
B^{(j)} \circ_r A^{(\ell)} &\coloneqq B_{(1,\ldots,j)}^{(j)}\biggl(\sum_{\ul{\alpha}_r\in P_r^j} A_{(\ul{\alpha}_r,j+1,\ldots,j+\ell-r)}^{(\ell)}\biggr) \in \L(\Sc_s(\R^k),\Sc'(\R^k)) \label{eq:B_c_A}.
\end{align}
Note that the compositions are well-defined since
\begin{equation}
\sum_{\ul{\alpha}_r\in P_r^\ell} B_{(\ul{\alpha}_r,\ell+1,\ldots,\ell+j-r)}^{(j)} \text{ and } \sum_{\ul{\alpha}_r\in P_r^j} A_{(\ul{\alpha}_r,j+1,\ldots,j+\ell-r)}^{(\ell)}
\end{equation}
have targets which are symmetric under permutation in the first $\ell$ and $j$ coordinates, respectively. We then set
\begin{equation}
\label{eq:comm_r}
\comm{A^{(\ell)}}{B^{(j)}}_r \coloneqq {j\choose r} A^{(\ell)} \circ_r B^{(j)} - {\ell\choose r} B^{(j)}\circ_r A^{(\ell)}.
\end{equation}
The motivation for the combinatorial factors in \eqref{eq:comm_r} will become clear from the proof of \cref{lem:hi_fil} below.

\begin{remark}
\label{rem:tbyt}
We may also proceed term-by-term to define \eqref{eq:A_c_B} and \eqref{eq:B_c_A} by considering an extensions of $A^{(\ell)}$ and $B^{(j)}$ to $\L(\Sc(\R^\ell),\Sc(\R^\ell))$ and $\L(\Sc(\R^j),\Sc(\R^j))$, so that $A_{(1,\ldots,\ell)}^{(\ell)}$ and $B_{(1,\ldots,j)}^{(j)}$ are then elements of $\L(\Sc(\R^k),\Sc(\R^k))$. The choice of extensions is immaterial by the target symmetry of operators with which the extensions are right-composed.
\end{remark}

In the sequel, we will need a technical lemma concerning the separate continuity of the binary operation $\circ_r$. The proof of this result is quite similar to that of (the more general) \cref{lem:gmp} below, so we omit the proof.

\begin{lemma}
\label{lem:ua_ub_cont}
Let $\ell,j,k,N\geq 1$ be integers such that $\ell,j\leq N$ and $\min\{\ell+j-1,N\}=k$. Let $r$ be an integer such that $r_0\leq r\leq \min\{\ell,j\}$, where
\begin{equation}
r_0 \coloneqq \max\{1,\min\{\ell,j\}-(N-\max\{\ell,j\})\}.
\end{equation}
Then the bilinear map
\begin{equation}
(\cdot )\circ_r(\cdot): \tl{\L}(\Sc(\R^\ell), \Sc(\R^\ell)) \times \tl{\L}(\Sc(\R^j),\Sc(\R^j)) \rightarrow \tl{\L}(\Sc(\R^k),\Sc(\R^k))
\end{equation}
is separately continuous.\footnote{We recall that $\tl{\L}(\Sc(\R^k), \Sc(\R^k)$ denotes the space $\L(\Sc(\R^k),\Sc(\R^k))$ of continuous linear maps from Schwartz space to itself equipped with the subspace topology induced by $\L(\Sc(\R^k),\Sc'(\R^k))$.}
\end{lemma}

\begin{lemma}[Filtration of hierarchy]\label{lem:hi_fil}
Let $N\in\mathbb{N}$ and let $1\leq \ell,j\leq N$. Then for any $A^{(\ell)}\in \g_{\ell}$ and $B^{(j)}\in \g_{j}$, there exists a unique $C^{(k)}\in \g_{k}$, for $k\coloneqq \min\{\ell+j-1,N\}$, such that
\begin{equation}
\comm{\epsilon_{\ell,N}(A^{(\ell)})}{\epsilon_{j,N}(B^{(j)})}_{\g_N} = \epsilon_{k,N}(C^{(k)}).
\end{equation}
\end{lemma}

\begin{proof}
By definition,
\begin{align}
&\comm{\epsilon_{\ell,N}(A^{(\ell)})}{\epsilon_{j,N}(B^{(j)})}_{\g_N} \nonumber\\
&= N C_{\ell,N}C_{j,N}\paren*{ \sum_{\um_{\ell}\in P_{\ell}^{N}} A_{(m_1,\ldots, m_\ell)}^{(\ell)}\biggl(\sum_{\un_j \in P_{j}^{N}}  B_{(n_{1},\ldots, n_{j})}^{(j)} \biggr) - \sum_{\un_j \in P_{j}^{N}}  B_{(n_1,\ldots, n_j)}^{(j)}\biggl(\sum_{\um_{\ell}\in P_{\ell}^{N}} A_{(m_1, \ldots, m_\ell)}^{(\ell)}\biggr)} \nonumber\\
&=N C_{\ell,N}C_{j,N}\sum_{r=1}^{\min\{\ell,j\}} \Biggl(\sum_{\um_{\ell}\in P_{\ell}^{N}} A_{(m_1, \ldots, m_\ell)}^{(\ell)}\biggl(\sum_{{\un_j \in P_{j}^{N}}\atop |\{m_{1},\ldots,m_{\ell}\}\cap\{n_{1},\ldots,n_{j}\}|=r} B_{(n_{1},\ldots,n_{j})}^{(j)}\biggr) \nonumber\\
&\hspace{55mm} - \sum_{\un_j \in P_j^N} B_{(n_{1},\ldots,n_{j})}^{(j)}\biggl(\sum_{{\um_\ell \in P_\ell^N}\atop {|\{m_1,\ldots,m_\ell\}\cap\{n_1,\ldots,n_j\}|=r}} A_{(m_1, \ldots, m_\ell)}^{(\ell)}\biggr) \Biggr).
\label{eq:fil_pre}
\end{align}
Without loss of generality, suppose that $\ell \geq j$. We consider the case $\ell+j-1\leq N$. For each integer $1\leq r\leq j$, we have by the $\mathbb{S}_j$-invariance of the operator $B^{(j)}$ that
\begin{equation}
\sum_{{\un_j \in P_{j}^{N}}\atop |\{m_{1},\ldots,m_{\ell}\}\cap\{n_{1},\ldots,n_{j}\}|=r} B_{(n_{1},\ldots,n_{j})}^{(j)} = {j\choose r}\sum_{{\un_j\in P_j^N}\atop {{\{n_1,\ldots,n_r\}\subset \{m_1,\ldots,m_\ell\}} \atop {\{n_{r+1},\ldots,n_j\} \cap \{m_1,\ldots,m_\ell\}=\emptyset}}} B_{(n_{1},\ldots,n_{j})}^{(j)}.
\end{equation}
Similarly, by the $\Ss_\ell$-invariance of the operator $A^{(\ell)}$, we have that
\begin{equation}
\sum_{{\um_\ell \in P_{\ell}^{N}}\atop |\{n_{1},\ldots,n_{j}\}\cap\{m_{1},\ldots,m_{\ell}\}|=r} A_{(m_{1},\ldots,m_{\ell})}^{(\ell)} = {\ell\choose r}\sum_{{\um_\ell\in P_\ell^N}\atop {{\{m_1,\ldots,m_r\}\subset \{n_1,\ldots,n_j\}} \atop {\{m_{r+1},\ldots,m_\ell\} \cap \{n_1,\ldots,n_j\}=\emptyset}}} A_{(m_1,\ldots,m_\ell)}^{(\ell)}.
\end{equation}

Upon relabeling the summation, we see that
\begin{equation}
\label{eq:fil_main}
\begin{split}
\eqref{eq:fil_pre} &=N C_{\ell,N}C_{j,N}\sum_{r=1}^{\min\{\ell,j\}}  \sum_{\up_{\ell+j-r}\in P_{\ell+j-r}^N}\Biggl( {j\choose r} A_{(p_1,\ldots,p_l)}^{(l)}\biggl(\sum_{{1\leq \ell_1,\ldots,\ell_r\leq \ell}\atop {|\{\ell_1,\ldots,\ell_r\}|=r}} B_{(p_{\ell_1},\ldots,p_{\ell_r},p_{\ell+1},\ldots,p_{\ell+j-r})}^{(j)}\biggr) \\
&\hspace{65mm} - {\ell\choose r} B_{(p_1,\ldots,p_j)}^{(j)}\biggl(\sum_{{1\leq j_1,\ldots,j_r\leq j} \atop{|\{j_1,\ldots,j_r\}|=r}} A_{(p_{j_1},\ldots,p_{j_r},p_{j+1},\ldots,p_{j+\ell-r})}^{(\ell)}\biggr) \Biggr).
\end{split} 
\end{equation}

If $r=1$, then the summand of \eqref{eq:fil_main} equals
\begin{align}
&NC_{\ell,N}C_{j,N}\sum_{\up_k \in P_{k}^{N}} j A_{(p_{1},\ldots,p_{\ell})}^{(\ell)}\biggl(\sum_{\alpha=1}^{\ell} B_{(p_{\alpha},p_{\ell+1},\ldots,p_{k})}^{(j)}\biggr) - \ell B_{(p_1,\ldots,p_\ell)}^{(\ell)}\biggl(\sum_{\alpha=1}^{j} A_{(p_{\alpha},p_{j+1},\ldots,p_{k})}^{(\ell)}\biggr) \nonumber\\
&=NC_{\ell,N}C_{j,N}\sum_{\up_k \in P_{k}^{N}} j (A^{(\ell)}\circ_1 B^{(j)})_{(p_1,\ldots,p_k)} - \ell (B^{(j)}\circ_1 A^{(\ell)})_{(p_1,\ldots,p_k)} \nonumber\\
&=\epsilon_{k,N}\paren*{\frac{N C_{\ell,N}C_{j,N}}{C_{k,N}}\Sym_k\paren*{j (A^{(\ell)}\circ_1 B^{(j)}) - \ell (B^{(j)}\circ_1 A^{(\ell)})}}.
\end{align}

Now suppose that $r>1$. Observe that
\begin{equation}
\begin{split}
&\sum_{\up_{\ell+j-r}\in P_{\ell+j-r}^N} \Biggl( {j\choose r} A_{(p_1,\ldots,p_\ell)}^{(\ell)}\biggl(\sum_{{1\leq \ell_1,\ldots,\ell_r\leq \ell}\atop {|\{\ell_1,\ldots,\ell_r\}|=r}} B_{(p_{\ell_1},\ldots,p_{\ell_r},p_{\ell+1},\ldots,p_{\ell+j-r})}^{(j)}\biggr) \\
&\hspace{25mm} - {\ell\choose r} B_{(p_1,\ldots,p_j)}^{(j)}\biggl(\sum_{{1\leq j_1,\ldots,j_r\leq j} \atop{|\{j_1,\ldots,j_r\}|=r}} A_{(p_{j_1},\ldots,p_{j_r},p_{j+1},\ldots,p_{j+\ell-r})}^{(\ell)}\biggr)\Biggr)
\end{split}
\end{equation}
cannot be immediately identified as an embedded element of $\g_{k}$ because the summation is not over tuples $\up_k\in P_{k}^{N}$. Indeed, we are missing $k-(\ell+j-r) = r-1$ coordinates. To address this issue, we observe that we can write $\up_k \in P_{k}^{N}$ as $\up_k = (\up_{\ell+j-r}, \uq_{r-1})$, where $ \up_{\ell+j-r} \in P_{\ell+j-r}^{N}$ and
\begin{equation}
\uq_{r-1}  \in (\N_{\leq N} \setminus \{p_{1},\ldots,p_{\ell+j-r}\})^{r-1}, \, \textup{with }|\{q_{1},\ldots,q_{r-1}\}|=r-1.
\end{equation}
For each $\up_{\ell+j-r} \in P_{\ell+j-r}^{N}$, the number of $(r-1)$-cardinality subsets of $\N_{\leq N} \setminus \{p_{1},\ldots,p_{\ell+j-r}\}$ is 
\[
{{N-\ell-j+r} \choose {r-1}}.
\]
Since there are $(r-1)!$ ways of permuting $r-1$ distinct elements, we conclude that for $\up_{\ell+j-r}\in P_{\ell+j-r}^{N}$,
\begin{align}
|\{\ul{q}_{r-1} \in (\N_{\leq N} \setminus \{p_{1},\ldots,p_{\ell+j-r}\})^{r-1} : |\{q_{1},\ldots,q_{r-1}\}| =r-1\}| &= {{N-\ell-j+r} \choose {r-1}} (r-1)! \nonumber\\
&=\prod_{m=1}^{r-1}(N-k+m),
\end{align}
where we use that $\ell+j-1=k$. Hence, the summand of \eqref{eq:fil_main} equals
\begin{equation}
\begin{split}
&\frac{NC_{\ell,N}C_{j,N}}{\prod_{m=1}^{r-1}(N-k+m)} \sum_{{\up_k \in P_{k}^{N}}} \Biggl({j\choose r} A_{(p_{1},\ldots,p_{\ell})}^{(\ell)}\biggl( \sum_{\uell_r \in P_{r}^{\ell}} B_{(p_{\ell_{1}},\ldots,p_{\ell_{r}}, p_{r+1},\ldots,p_{\ell+j-r})}^{(j)} \biggr)\\
&\hspace{45mm} -  {\ell\choose r} B_{(p_1,\ldots,p_j)}^{(j)}\biggl(\sum_{\uj_r\in P_r^j} A_{(p_{j_1},\ldots,p_{j_r},p_{j+1},\ldots,p_{j+\ell-r})}^{(\ell)}\biggr) \Biggr),
\end{split}
\end{equation}
and by definition, we obtain that this expression equals
\begin{equation}
\epsilon_{k,N}\paren*{\frac{NC_{\ell,N}C_{j,N}}{C_{k,N}\prod_{m=1}^{r-1}(N-k+m)}\Sym_k\paren*{{j\choose r}A^{(\ell)}\circ_r B^{(j)} - {\ell\choose r}B^{(j)}\circ_r A^{(\ell)}}}.
\end{equation}

Now suppose that $\ell+j-1>N$. Then proceeding as above, we see that $r\geq 1$ must in fact satisfy the lower bound
\begin{equation}
r\geq \min\{\ell,j\} - (N-\max\{\ell,j\})\eqqcolon r_0.
\end{equation}
Combining these results, we conclude that
\begin{equation}
\begin{split}
&\comm{\epsilon_{\ell,N}(A^{(\ell)})}{\epsilon_{j,N}(B^{(j)})}_{\g_N} \\
&= \epsilon_{k,N}\biggl( \Sym_k\paren*{\sum_{r=r_0}^{N}\frac{NC_{\ell,N}C_{j,N}}{C_{k,N}\prod_{m=1}^{r-1}(N-k+m)} \paren*{{j\choose r} A^{(\ell)}\circ_r B^{(j)} - {\ell\choose r} B^{(j)}\circ_r A^{(\ell)}}},
\end{split}
\end{equation}
which concludes the proof of the lemma.
\end{proof}

We now have all the technical lemmas needed to define the Lie algebra ${\G}_{N}$ of observable $N$-hierarchies. For $N\in\mathbb{N}$, let $\G_{N}$ denote the locally convex direct sum
\begin{equation}\label{GN_def}
{\G}_{N} \coloneqq \bigoplus_{k=1}^{N} {\g}_{k},
\end{equation}
where we recall that 
\begin{equation}
\g_{k} = \{ A^{(k)} \in \tl{\L}(\Sc_s(\R^k),\Sc_s(\R^k)) : (A^{(k)})^* = -A^{(k)} \}.
\end{equation}
We define a bracket on $A_N = (A_N^{(k)})_{k\in\N_{\leq N}}, B_N = (B_N^{(k)})_{k\in\N_{\leq N}} \in \G_N$ by
\begin{equation}
\comm{A_N}{ B_N}_{{\G}_{N}} \coloneqq C_N = (C_N^{(k)} )_{k\in\N_{\leq N}} ,
\end{equation}
where
\begin{equation}\label{lie_def}
C_N^{(k)} \coloneqq \sum_{{1\leq \ell,j\leq N}\atop {\min\{\ell+j-1,N\}=k}} \epsilon_{k,N}^{-1}\paren*{\comm{\epsilon_{\ell,N}(A_N^{(\ell)})}{\epsilon_{j,N}(B_N^{(j)})}_{\g_{N}}}.
\end{equation}

It remains for us to check that ${\G}_N$ together with its bracket is actually a Lie algebra in the sense of \cref{def:la}, as we have so claimed above. Before doing so, we collect a result which will be useful in the sequel. Namely, that as a byproduct of the proof \cref{lem:hi_fil}, we have the following explicit formula for the Lie bracket $\comm{A_N}{B_N}_{\G_{N}}$ for two observable $N$-hierarchies, which is quite useful for computations.

\begin{prop}[Formula for $\comm{A_N}{B_N}_{\G_{N}}^{(k)}$]
\label{prop:Lie_hi_form}
Let $N\in\N$, and let $A_N=(A_N^{(k)})_{k\in\N_{\leq N}}, B_N=(B_N^{(k)})_{k\in\N_{\leq N}}$ be observable $N$-hierarchies. Then for integers $1\leq k\leq N$, we have that
\begin{equation}
\label{eq:N_LB_form}
\comm{A_N}{B_N}_{\G_{N}}^{(k)} =  \sum_{{1\leq \ell,j\leq N}\atop {\min\{\ell+j-1,N\}=k}} \Sym_k\biggl(\sum_{r=r_0}^{\min\{\ell,j\}} C_{\ell j k r N} \comm{A_N^{(\ell)}}{B_N^{(j)}}_r \biggr),
\end{equation}
where
\begin{equation}
 C_{\ell j k r N} \coloneqq \frac{N C_{\ell,N}C_{j,N}}{C_{k,N} \prod_{m=1}^{r-1}(N-k+m)},\footnote{Recall that $C_{\ell, N} = 1/|P_{\ell}^N|$.} \qquad r_0\coloneqq \max\{1,\min\{\ell,j\}-(N-\max\{\ell,j\})\},
\end{equation}
and where $\comm{\cdot}{\cdot}_r$ is defined in \eqref{eq:comm_r}.
\end{prop}

We now establish \cref{prop:NH_LA}, which is our first main result of this section.

\NHLA*

\begin{proof}[Proof of \cref{prop:NH_LA}]
There are two parts to the verification: an algebraic part and an analytic part.

\begin{itemize}[leftmargin=*]
\item
We first consider the algebraic part, which amounts to checking bilinearity, anti-symmetry, and the Jacobi identity. The first two properties are obvious from the definition of $\G_N$. For the third property, let $A_N,B_N,C_N\in\G_N$. We need to show that
\begin{equation}
\comm{A_N}{\comm{B_N}{C_N}_{\G_N}}_{\G_N} + \comm{C_N}{\comm{A_N}{B_N}_{\G_N}}_{\G_N} + \comm{B_N}{\comm{C_N}{A_N}_{\G_N}}_{\G_N}=0.
\end{equation}
Since $\ep_{k,N}$ is injective, it suffices to show that $\ep_{k,N}$ applied to the left-hand side of the preceding identity equals the zero element of $\g_N$. We only present the details when the component index satisfies $1\leq k<N$ and leave verification of the remaining $k=N$ case as an exercise to the reader. Using the definition of the Lie bracket and bilinearity, we have the identities
\begin{align}
\ep_{k,N}\paren*{\comm{A_N}{\comm{B_N}{C_N}_{\G_N}}_{\G_N}^{(k)}} &= \sum_{j_1+j_2-1=k} \comm{\ep_{j_1,N}(A_N^{(j_1)})}{\ep_{j_2,N}(\comm{B_N}{C_N}_{\G_N}^{(j_2)})}_{\g_N} \nonumber\\
&=\sum_{j_1+j_2-1=k} \sum_{j_3+j_4-1=j_2} \comm{\ep_{j_1,N}(A_N^{(j_1)})}{\comm{\ep_{j_3,N}(B_N^{(j_3)})}{\ep_{j_4,N}(C_N^{(j_4)})}_{\g_N}}_{\g_N} \nonumber \\
&= \sum_{\ell_1+\ell_2+\ell_3=k+2} \comm{\ep_{\ell_1,N}(A_N^{(\ell_1)})}{\comm{\ep_{\ell_2,N}(B_N^{(\ell_2)})}{\ep_{\ell_3,N}(C_N^{(\ell_3)})}_{\g_N}}_{\g_N},\nonumber
\end{align}
\begin{align}
\ep_{k,N}\paren*{\comm{C_N}{\comm{A_N}{B_N}_{\G_N}}_{\G_N}^{(k)}} &= \sum_{j_1+j_2-1=k} \comm{\ep_{j_1,N}(C_N^{(j_1)})}{\ep_{j_2,N}(\comm{A_N}{B_N}_{\G_N}^{(j_2)})}_{\g_N} \nonumber\\
&=\sum_{j_1+j_2-1=k} \sum_{j_3+j_4-1=j_2} \comm{\ep_{j_1,N}(C_N^{(j_1)})}{\comm{\ep_{j_3,N}(A_N^{(j_3)})}{\ep_{j_4,N}(B_N^{(j_4)})}_{\g_N}}_{\g_N} \nonumber \\
&=\sum_{\ell_1+\ell_2+\ell_3=k+2} \comm{\ep_{\ell_3,N}(C_N^{(\ell_3)})}{\comm{\ep_{\ell_1,N}(A_N^{(\ell_1)})}{\ep_{\ell_2,N}(B_N^{(\ell_2)})}_{\g_N}}_{\g_N},\nonumber
\end{align}
\begin{align}
\ep_{k,N}\paren*{\comm{B_N}{\comm{C_N}{A_N}_{\G_N}}_{\G_N}^{(k)}} &= \sum_{j_1+j_2-1=k} \comm{\ep_{j_1,N}(B_N^{(j_1)})}{\ep_{j_2,N}(\comm{C_N}{A_N}_{\G_N}^{(j_2)})}_{\g_N} \nonumber\\
&=\sum_{j_1+j_2-1=k} \sum_{j_3+j_4-1=j_2} \comm{\ep_{j_1,N}(B_N^{(j_1)})}{\comm{\ep_{j_3,N}(C_N^{(j_3)})}{\ep_{j_4,N}(A_N^{(j_4)})}_{\g_N}}_{\g_N} \nonumber \\
&=\sum_{\ell_1+\ell_2+\ell_3=k+2} \comm{\ep_{\ell_2,N}(B_N^{(\ell_2)})}{\comm{\ep_{\ell_3,N}(C_N^{(\ell_3)})}{\ep_{\ell_1,N}(A_N^{(\ell_1)})}_{\g_N}}_{\g_N}.\nonumber
\end{align}
Since $\comm{\cdot}{\cdot}_{\g_N}$ is a Lie bracket and therefore satisfies the Jacobi identity, it follows that for fixed integers $1\leq \ell_1,\ell_2,\ell_3\leq N$,
\begin{equation}
\begin{split}
0 &=\comm{\ep_{\ell_1,N}(A_N^{(\ell_1)})}{\comm{\ep_{\ell_2,N}(B_N^{(\ell_2)})}{\ep_{\ell_3,N}(C_N^{(\ell_3)})}_{\g_N}}_{\g_N} \\
&\phantom{=}+ \comm{\ep_{\ell_3,N}(C_N^{(\ell_3)})}{\comm{\ep_{\ell_1,N}(A_N^{(\ell_1)})}{\ep_{\ell_2,N}(B_N^{(\ell_2)})}_{\g_N}}_{\g_N} \\
&\phantom{=}\phantom{=} + \comm{\ep_{\ell_2,N}(B_N^{(\ell_2)})}{\comm{\ep_{\ell_3,N}(C_N^{(\ell_3)})}{\ep_{\ell_1,N}(A_N^{(\ell_1)})}_{\g_N}}_{\g_N}.
\end{split}
\end{equation}
Hence,
\begin{equation}
\ep_{k,N}\paren*{\comm{A_N}{\comm{B_N}{C_N}_{\G_N}}_{\G_N}^{(k)} + \comm{C_N}{\comm{A_N}{B_N}_{\G_N}}_{\G_N}^{(k)} + \comm{B_N}{\comm{C_N}{A_N}_{\G_N}}_{\G_N}^{(k)}} = 0\in\g_N.
\end{equation}

\item
We now consider the analytic part, which amounts to checking the separate continuity of $\comm{\cdot}{\cdot}_{\G_N}$. Using the anti-symmetry of the bracket, it suffices to show that for $A_N\in \G_N$ fixed, the map
\begin{equation}
\G_N\rightarrow \G_N, \qquad B_N\mapsto \comm{A_N}{B_N}_{\G_N}
\end{equation}
is continuous. Moreover, it suffices to show that for each $k\in \N_{\leq N}$, the map
\[
\G_N\rightarrow \g_k, \qquad B_N\mapsto \comm{A_N}{B_N}_{\G_N}^{(k)}
\]
is continuous. 

Let $(B_{N,a})_{a\in\mathsf{A}}$, where $B_{N,a} =(B_{N,a}^{(k)})_{k\in\N_{\leq N}}$, be a net in $\G_N$ converging to $B_N=(B_N^{(k)})_{k\in\N_{\leq N}} \in \G_N$. By the continuity of the projection maps $\G_N\rightarrow \g_k$ for each $k\in \N_{\leq N}$, we have that $(B_{N,a}^{(k)})_{a\in\mathsf{A}}$ is a net in $\g_k$ converging to $B_N^{(k)}\in\g_k$.  

Unpacking the definition of $[A_N,B_{N,a}]^{(k)}_{\G_N}$ and using the continuity of the $\Sym_k$ operator and the operations of addition and scalar multiplication, together with the fact there are only finitely many terms, it suffices to show that for any integers $1\leq \ell,j\leq N$ satisfying $\min\{\ell+j-1,N\}=k$, any integer  $r_0\leq r\leq \min\{\ell,j\}$, we have the net convergence
\begin{equation}
\comm{A_N^{(\ell)}}{B_{N,a}^{(j)}}_{r} \rightarrow \comm{A_N^{(\ell)}}{B_N^{(j)}}_r 
\end{equation}
in $\tl{\L}(\Sc_s(\R^k),\Sc(\R^k))$. But this convergence is a consequence of \cref{lem:ua_ub_cont}, thus completing the proof.
\end{itemize}
\end{proof}

\subsection{Lie-Poisson manifold $\G_{N}^{*}$ of finite hierachies of density matrices}
\label{ssec:LP_N}
In this subsection, we define the Lie-Poisson manifold $\g_{N}^{*}$ of $N$-body density matrices and the Lie-Poisson manifold $\G_N^{*}$ of density matrix $N$-hierarchies. A good heuristic to keep in mind is that density matrices are dual to skew-adjoint operators. We remind the reader that the superscript $*$ does not denote the literal functional analytic dual of $\g_{N}$ (respectively, $\G_N$) as a topological vector space, but rather a space in weakly non-degenerate pairing with $\g_{N}$ (respectively, $\G_N$).

To begin with, we define the real topological vector space
\begin{equation}
{\g}_{N}^{*} \coloneqq \{\Psi_{N} \in \mathcal{L}(\Sc_s'(\R^N),\Sc_s(\R^N)) : \Psi_N^* = \Psi_N\}
\end{equation}
endowed with the subspace topology. 

\begin{remark}
Our definition of $\g_N^*$ is quite natural as it is isomorphic to the strong dual of $\g_N$. The proof of this fact is quite similar to that of \cref{lem:G_inf_dual} shown below.
\end{remark}

We now define a suitable unital sub-algebra $\mathcal{A}_{DM,N} \subset C^\infty({\g}_N^*; \R)$ of admissible functionals to build a weak Poisson structure for $\g_N^*$.

\begin{mydef}
Let $\mathcal{A}_{DM,N}$ be the algebra with respect to point-wise product generated by the functionals in
\begin{equation}
\{F \in C^\infty(\g_N^*;\R) : F(\cdot) = i\Tr_{1,\ldots,N}(A^{(N)}\cdot ), \enspace A^{(N)}\in \g_N\} \cup \{F\in C^\infty(\g_N^*;\R) : F(\cdot) = C\in\R\}.
\end{equation}
\end{mydef}
In words, $\A_{DM,N}$ is the algebra (under point-wise product) generated by the constants and the image of $\g_N$ under the canonical embedding into $(\g_N^*)^*$.

\medskip
We record the following result, whose proof we omit since it is similar to and simpler than that of \cref{prop:LP}, which will be used in \cref{ssec:dm_pomo} below.

\begin{prop}
\label{prop:DM_WP}
$(\g_N^*, \A_{DM,N}, \pb{\cdot}{\cdot}_{\g_N^*})$ is a weak Poisson manifold.
\end{prop}

Before proceeding, it will be useful to record the following lemma regarding the dual of $\g_N^*$. In particular, we note that the dual of $\g_N^*$ is \emph{not} isomorphic to $\g_N$.

\begin{lemma}[Dual of $\g_N^*$]
The topological dual of $\g_N^*$, denoted by $(\g_N^*)^*$ and endowed with the strong dual topology, is isomorphic to
\begin{equation}
\{A^{(N)} \in \L(\Sc_s(\R^N),\Sc_s'(\R^N)) : (A^{(N)})^* = -A^{(N)}\},
\end{equation}
equipped with the subspace topology induced by $\mathcal{L}(\Sc_s(\R^N),\Sc_s'(\R^N))$, via the canonical bilinear form
\begin{equation}
i\Tr_{1,\ldots,N}(A^{(N)}\Psi_N), \qquad \Psi_N \in \g_N^*.
\end{equation}
\end{lemma}
\begin{proof}
The proof follows from the duality $\L(\Sc_s(\R^N),\Sc_s'(\R^N))) \cong \L(\Sc_s'(\R^N),\Sc_s(\R^N))^*$ together with a polarization-type argument. We leave the details to the reader.
\end{proof}

\begin{remark}
\label{rem:GD_Op_DM}
The previous lemma implies that, given a functional $F\in C^\infty(\g_N^*;\R)$ and a point $\Psi_N\in \g_N^*$, we may identify the continuous linear functional $dF[\Psi_N]$, given by the G\^ateaux derivative of $F$ at the point $\Psi_N$, as a skew-adjoint element of $\L(\Sc_s(\R^N),\Sc_s'(\R^N))$. We will abuse notation and denote this element by $dF[\Psi_N]$. Moreover, as we will see below, it is a small computation using the generating structure of $\A_{DM,N}$ that $dF[\Psi_N]\in \g_N$.
\end{remark}

We next define the Lie-Poisson manifold of density matrix $N$-hierarchies. To begin, define the real topological vector space
\begin{equation}\label{GNstar_def}
\G_{N}^{*} \coloneqq \bigl\{\Gamma_N=(\Gamma_N^{(k)})_{k\in\N_{\leq N}} \in \prod_{k=1}^N \L(\Sc_s'(\R^k),\Sc_s(\R^k)) : \gamma_N^{(k)} = (\gamma_N^{(k)})^* \enspace \forall k\in\N \bigr\}
\end{equation}
endowed with the subspace product topology. We first note that our definition of $\G_N^*$ is quite natural, as it is isomorphic to the topological dual of $\G_N$, a fact we prove in the next lemma.

\begin{lemma}[Dual of $\G_N$]
\label{lem:dual_GN}
The topological dual of $\G_N$, denoted by $(\G_N)^*$ and endowed with the strong dual topology, is isomorphic to $\G_N^*$.
\end{lemma}
\begin{proof}
Using the isomorphism
\begin{equation}
\paren*{ \tl{\L}(\Sc_{s}(\R^{k}),\Sc_{s}(\R^{k}))}^{*} \cong \paren*{\L(\Sc_{s}(\R^{k}),\Sc_{s}'(\R^{k}))}^* = \L(\Sc_s'(\R^k),\Sc_s(\R^k)), \qquad \forall k\in\N,
\end{equation}
which follows from the proof of \cref{lem:gmp_dual} together with the duality of direct sums and direct products, see for instance \cite[Proposition 2 in \S{14}, Chapter 3]{Horvath1966}, we have that
\begin{equation}
\label{eq:Phi'_iso}
\paren*{\bigoplus_{k=1}^{N} \tl{\L}(\Sc_{s}(\R^{k}),\Sc_{s}(\R^{k}))}^{*} \underbrace{\cong}_{\eqqcolon \Phi'} \prod_{k=1}^{N} \L(\Sc_{s}'(\R^{k}),\Sc_{s}(\R^{k})),
\end{equation}
via the canonical trace pairing
\[
(A_N, \Gamma_N) \mapsto i \Tr (A_N \cdot \Gamma_N).
\]
Thus elements of $(\G_{N})^{*}$ may be identified with functionals $i \Tr(\cdot \Gamma_N )$, and so to prove the lemma, we will show that the map
\begin{equation}
\Phi: \G_{N}^{*} \rightarrow (\G_{N})^{*}, \qquad \Gamma_N \mapsto i\Tr\paren*{\cdot\Gamma_N},
\end{equation}
is bijective and that both $\Phi$ and $\Phi^{-1}$ are continuous.

First, we show surjectivity of $\Phi$. Given any functional $F\in (\G_{N})^{*}$, we need to find some density matrix $N$-hierarchy $\Gamma_N \in \G_{N}^{*}$ such that
\begin{equation}
F(A_N) = i \Tr (A_N \cdot \Gamma_N).
\end{equation}
To accomplish this task, we define a functional
\begin{equation}
\widetilde{F} \in \paren*{\bigoplus_{k=1}^{N}\tl{\L}(\Sc_{s}(\R^{k}),\Sc_{s}(\R^{k}))}^{*}
\end{equation}
by the formula
\begin{equation}\label{f_def}
\widetilde{F}(A_N) \coloneqq \frac{1}{2}F\paren*{A_N-A_N^{*}} - \frac{i}{2}F\paren*{(A_N-A_N^{*})} + \frac{1}{2}F\paren*{i(A_N+A_N^{*})} - \frac{i}{2}F\paren*{i(A_N+A_N^{*})}.
\end{equation}
By the canonical dual trace pairing, there exists a unique 
\[
\Gamma_N \in \prod_{k=1}^{N}\L(\Sc_{s}'(\R^{k}),\Sc_{s}(\R^{k}))
\]
 such that
\begin{equation}
\widetilde{F}(A_N) = i\Tr\paren*{A_N\cdot\Gamma_N}, \qquad \forall A_N\in \bigoplus_{k=1}^{N}\tl{\L}(\Sc_{s}(\R^{k}),\Sc_{s}(\R^{k})).
\end{equation}
Evaluating $\widetilde{F}$ on $A_N\in\G_{N}$, that is assuming $A_N = - A_N^*$, we obtain from \eqref{f_def} that
\begin{align}
\paren*{1-i}F(A_N) = i\Tr\paren*{A_N\cdot\Gamma_N},
\end{align}
and adding this expression to its conjugate implies that 
\begin{align*}
2F(A_N) = i\paren*{\Tr\paren*{A_N\cdot\Gamma_N} - \ol{\Tr\paren*{A_N\cdot\Gamma_N}}}.
\end{align*}
Since
\begin{equation*}
(A_N\cdot\Gamma_N)^{(k)}=A_N^{(k)}\gamma_N^{(k)} \in \L(\Sc_s'(\R^k), \Sc_s(\R^k)), \qquad \forall k\in\N_{\leq N},
\end{equation*}
its trace exists in the usual sense of an operator on a separable Hilbert space. Furthermore, the adjoint of $A_N^{(k)}\gamma_N^{(k)}$ as a bounded linear operator on $L_s^2(\R^k)$, denoted by $(A_N^{(k)}\gamma_N^{(k)})^*$, belongs to $\L(\Sc_s'(\R^k),\Sc(\R^k))$. A short computation using the skew- and self-adjointness of $A_N^{(k)}$ and $\gamma_N^{(k)}$, respectively, shows that
\begin{equation*}
(A_N^{(k)}\gamma_N^{(k)})^* = -\gamma_N^{(k)}A_N^{(k)},
\end{equation*}
where we abuse notation by letting $A_N^{(k)}$ also denote the extension to an element of $\L(\Sc_s'(\R^k),\Sc_s'(\R^k))$. Consequently, we are justified in writing
\begin{equation*}
\ol{\Tr_{1,\ldots,k}\paren*{A_N^{(k)}\gamma_N^{(k)}}} = \Tr_{1,\ldots,k}\paren*{(A_N^{(k)}\gamma_N^{(k)})^*} = -\Tr_{1,\ldots,k}\paren*{\gamma_N^{(k)}A_N^{(k)}} = -\Tr_{1,\ldots,k}\paren*{A_N^{(k)}\gamma_N^{(k)}},
\end{equation*}
where the ultimate equality follows from an approximation of $A_N^{(k)}$ and the cyclicity of trace. Therefore,
\begin{equation}
\widetilde{\Gamma}_N = \frac{1}{2}(\Gamma_N+\Gamma_N^{*})
\end{equation}
is the desired density matrix $N$-hierarchy. Injectivity of $\Phi$ follows from the polarization identity by considering elements of $\G_N$ of the form
\begin{equation}
A_{N,k_0}^{(k)} =
\begin{cases}
i\ket*{f^{(k_0)}}\bra*{f^{(k_0)}}, & {k=k_0} \\
0, & {\text{otherwise}}
\end{cases},
\end{equation}
where $k_0\in\N_{\leq N}$ and $f^{(k_0)}\in\Sc_{s}(\R^{k_0})$. Hence $\Phi$ is bijective.

Next, we claim that both $\Phi$ and $\Phi^{-1}$ are continuous. Since $\G_{N}^{*}$ is a Fr\'{e}chet space, it suffices by the open mapping theorem to show that $\Phi$ is continuous. Let $\iota_{\G_{N}}$ denote the canonical inclusion map
\begin{equation}\label{can_inc}
\G_{N}\subset \bigoplus_{k=1}^{N}\tl{\L}(\Sc_{s}(\R^{k}),\Sc_{s}(\R^{k})),
\end{equation}
which is continuous by definition of the subspace topology, with adjoint
\begin{equation}
\iota_{\G_{N}}^{*}:\paren*{\bigoplus_{k=1}^{N}\tl{\L}(\Sc_{s}(\R^{k}),\Sc_{s}(\R^{k}))}^{*} \rightarrow (\G_{N})^{*},
\end{equation}
and let $\iota_{\G_{N}^{*}}$ denote the canonical inclusion map
\begin{equation}
\G_{N}^{*}\subset \prod_{k=1}^{N}\L(\Sc_{s}'(\R^{k}),\Sc_{s}(\R^{k})),
\end{equation}
which is also continuous by definition of the subspace topology. Then we can write
\begin{equation}
\Phi = \iota_{\G_{N}}^{*} \circ (\Phi')^{-1} \circ \iota_{\G_{N}^{*}},
\end{equation}
where $\Phi'$ is the canonical isomorphism described in \eqref{eq:Phi'_iso}. Since $\iota_{\G_{N}}^{*}$ is continuous, as can be checked directly or by appealing to the corollary of Proposition 19.5 in \cite{Treves1967}, it follows that $\Phi$ is the composition of continuous maps, completing the proof of the claim.
\end{proof}

We now need to establish the existence of a Poisson structure for $\G_N^*$. As before, we choose a unital sub-algebra $\A_{H,N}\subset C^\infty(\G_N^*;\R)$, generated by trace functionals and constant functionals, to be the algebra of admissible functionals.

\begin{mydef}
Let $\A_{H,N}$ be the algebra with respect to point-wise product generated by the functionals in
\begin{equation}
\{F\in C^\infty(\G_N^*;\R) : F(\cdot) = i\Tr(A_N\cdot), \enspace A_N\in \G_N\} \cup \{F\in C^\infty(\G_N^*;\R) : F(\cdot)\equiv C\in\R\}.
\end{equation}
\end{mydef}

\begin{remark}
\label{rem:AH_can}
Our definition of $\A_{H,N}$ is not canonical in the sense that one may include additional functionals in it. However, since we are really only interested in trace functionals, we will not do so in this work.
\end{remark}

\begin{remark}
\label{rem:AH_struc}
The structure of $\A_{H,N}$ will be frequently used in the following way: it will suffice to verify various identities for finite products of trace functionals and constant functionals. Moreover, by \cref{rem:con_GD} below and the Leibnitz rule for the G\^ateaux derivative, it will often suffice to check identities on trace functionals.
\end{remark}

\begin{remark}
\label{rem:con_GD}
By the linearity of the trace and the definition of the G\^ateaux derivative, a trace functional has constant G\^ateaux derivative. Similarly, a constant functional has zero G\^ateaux derivative.
\end{remark}

To define the Lie-Poisson bracket on $\A_{H,N}\times\A_{H,N}$ using the Lie bracket $\comm{\cdot}{\cdot}_{\G_N}$ constructed in \cref{ssec:N_LA}, we need the following identification of continuous linear functionals with skew-adjoint operators, given via the canonical trace pairing. We note, in particular, that $(\G_N^*)^*$ is not isomorphic to $\G_N$.

\begin{lemma}[Dual of $\G_N^*$]
\label{lem:dual_GN*}
The topological dual of $\G_N^*$, denoted by $(\G_N^*)^*$ and endowed with the strong dual topology, is isomorphic to
\begin{equation}
\widetilde{\G}_N \coloneqq \bigl\{A_N\in\bigoplus_{k=1}^N \L(\Sc_s(\R^k),\Sc_s'(\R^k)) : (A_N^{(k)})^* = - A_N^{(k)}\bigr\}.
\end{equation}
\end{lemma}
\begin{proof}
We omit the proof as it proceeds quite similarly to that of \cref{lem:dual_GN}.
\end{proof}

We continue to abuse notation by using $dF[\Gamma_N]$ to denote both the continuous linear functional and the element of $\widetilde{\G}_N$. We are now prepared to introduce the Lie-Poisson bracket $\pb{\cdot}{\cdot}_{\G_N^*}$ on $\A_{H,N}\times\A_{H,N}$.
\begin{mydef}
Let $N\in\N$. For $F,G\in \A_{H,N}$, we define
\begin{equation}
\pb{F}{G}_{\G_{N}^{*}}(\Gamma_{N}) \coloneqq i\Tr\paren*{\comm{dF[\Gamma_{N}]}{dG[\Gamma_{N}]}_{\G_{N}}\cdot \Gamma_{N}} = \sum_{k=1}^{N} i\Tr_{1,\ldots,k}\paren*{\comm{dF[\Gamma_{N}]}{dG[\Gamma_{N}]}_{\G_{N}}^{(k)} \gamma_{N}^{(k)}},
\end{equation}
for  $\Gamma_N= (\gamma_N^{(k)})_{k\in\N_{\leq N}}\in \G_N^*$.
\end{mydef}

We now turn to the second main goal of this subsection, that is, proving \cref{prop:NH_WP}, the statement of which we repeat here for the reader's convenience.

\NHWP*

We begin with the following technical lemma for the functional derivative of $\pb{\cdot}{\cdot}_{\G_N^*}$.

\begin{lemma}
\label{lem:tr_gen_H}
Suppose that $G_j\in\A_{H,N}$ is a trace functional $G_j(\Gamma_N) = i\Tr(dG_j[0]\cdot\Gamma_N)$ for $j=1,2$. Then for all $\Gamma_N  \in \G_N^*$, the G\^ateaux derivative $d\pb{G_1}{G_2}_{\G_N^*}[\Gamma_N]$ at the point $\Gamma_N$ may be identifed with the element
\begin{equation}
\comm{dG_1[0]}{dG_2[0]}_{\G_N} \in \G_N
\end{equation}
via the canonical trace pairing. If $G_1$ is a trace functional and $G_2 = G_{2,1}G_{2,2}$ is the product of two trace functionals in $\A_{H,N}$, then $d\pb{G_1}{G_2}_{\G_N^*}[\Gamma_N]$ may be identified with
\begin{equation}
G_{2,1}(\Gamma_N)\comm{dG_1[0]}{dG_{2,2}[0]}_{\G_N} + G_{2,2}(\Gamma_N)\comm{dG_1[0]}{dG_{2,1}[0]}_{\G_N}
\end{equation}
for all $\Gamma_N\in\G_N^*$ via the canonical trace pairing.
\end{lemma}
\begin{proof}
The first assertion follows readily from the definition of $\pb{G_{1}}{G_{2}}_{\G_{N}^{*}}$. To see the second assertion, observe that by the Leibnitz rule for the G\^ateaux derivative and the bilinearity of the bracket $\comm{\cdot}{\cdot}_r$,
\begin{align*}
\comm{dG_{1}[\Gamma_N]^{(\ell)}}{dG_{2}[\Gamma_N]^{(j)}}_{r} &= G_{2,1}(\Gamma_N)\comm{dG_{1}[0]^{(\ell)}}{dG_{2,2}[0]^{(j)}}_{r} + G_{2,2}(\Gamma_N)\comm{dG_{1}[0]^{(\ell)}}{dG_{2,1}[0]^{(j)}}_{r}.
\end{align*}
Hence using \cref{prop:Lie_hi_form} and introducing the notation
\begin{equation}
\label{eq:comb_not}
C_{\ell jkrN} \coloneqq \frac{N C_{\ell,N}C_{j,N}}{C_{k,N}\prod_{m=1}^{r-1}(N-k+m)}, \qquad r_0 \coloneqq \max\{1,\min\{\ell,j\}-(N-\max\{\ell,j\})\},
\end{equation}
we obtain that
\begin{align}
&\comm{dG_{1}[\Gamma_N]}{dG_{2}[\Gamma_N]}_{\G_{N}}^{(k)} \nonumber \\
&= \sum_{{1\leq \ell,j\leq N}\atop {\min\{\ell+j-1,N\}=k}} \Sym_k\biggl(\sum_{r=r_0}^{\min\{\ell,j\}} C_{\ell jkrN} \comm{dG_1[\Gamma_N]^{(\ell)}}{dG_2[\Gamma_N]^{(j)}}_r\biggr) \nonumber\\
&=G_{2,1}(\Gamma_N)\sum_{{1\leq \ell,j\leq N}\atop {\min\{\ell+j-1,N\}=k}}\Sym_k\biggl(\sum_{r=r_0}^{\min\{\ell,j\}}C_{\ell jkrN} \comm{dG_1[0]^{(\ell)}}{dG_{2,2}[0]^{(j)}}_r\biggr) \nonumber\\
&\phantom{=} G_{2,2}(\Gamma_N)\sum_{{1\leq \ell,j\leq N}\atop {\min\{\ell+j-1,N\}=k}}\Sym_k\biggl(\sum_{r=r_0}^{\min\{\ell,j\}}C_{\ell jkrN} \comm{dG_1[0]^{(\ell)}}{dG_{2,1}[0]^{(j)}}_r\biggr) \nonumber\\
&=G_{2,1}(\Gamma_N)\comm{dG_1[0]}{dG_{2,2}[0]}_{\G_N}^{(k)} + G_{2,2}(\Gamma_N)\comm{dG_1[0]}{dG_{2,1}[0]}_{\G_N}^{(k)},
\end{align}
where the ultimate equality follows from another application of \cref{prop:Lie_hi_form}.
\end{proof}

We divide our proof of \cref{prop:NH_WP} into several lemmas. We first show that $\pb{\cdot}{\cdot}_{\G_N^*}$ is well-defined and is a Lie bracket satisfying the Leibnitz rule.

\begin{lemma}
\label{lem:H_LA_P1}
The formula
\begin{equation}
\pb{F}{G}_{\G_N^*}(\Gamma_N) \coloneqq i\Tr\paren*{\comm{dF[\Gamma_N]}{dG[\Gamma_N]}_{\G_N}\cdot\Gamma_N}, \qquad \forall \Gamma_N\in\G_N^*
\end{equation}
defines a map $\A_{H,N}\times\A_{H,N}\rightarrow \A_{H,N}$ which satisfies property \ref{item:wp_P1} in \cref{def:WP}.
\end{lemma}
\begin{proof}
We first show that for $F,G\in\A_{H,N}$, one has $\pb{F}{G}_{\G_{N}^{*}}\in\A_{H,N}$. Recall that $\A_{H,N}$ is generated by constant functionals and trace functionals, hence using the Leibnitz rule, bilinearity of $\comm{\cdot}{\cdot}_{\G_{N}}$, and the linearity of the trace, it suffices to consider the case where $F,G$ are both trace functionals. Indeed, elements of $\A_{H,N}$ are finite linear combinations of finite products of trace functionals and constant functionals, hence using that the derivative of constant functionals is zero, upon applying the Leibnitz rule, the elements of the product which are not differentiated can be treated as scalars when evaluated at a point $\Gamma_N$ and hence can be pulled out of the Lie bracket and then out of the trace by bilinearity. 

When $F,G$ are both trace functionals, $dF[\Gamma_N]$ and $dG[\Gamma_N]$ are constant in $\Gamma_N$ by \cref{rem:con_GD}, hence
\begin{equation}
\pb{F}{G}_{\G_{N}^{*}}(\Gamma_N) = i\Tr(\comm{dF[0]}{dG[0]}_{\G_{N}}\cdot\Gamma_N), \qquad \forall \,\Gamma_N \in \G_{N}^{*}.
\end{equation}
So, we only need to show that the right-hand side defines an element of $\A_{H,N}$. Since $dF[0]$ and $dG[0]$ both belong to $\G_{N}$, it follows from \cref{prop:NH_LA} that $\comm{dF[0]}{dG[0]}_{\G_{N}}\in\G_{N}$. Hence, $\pb{F}{G}_{\G_{N}^{*}} \in \A_{H,N}$, which completes the proof of the claim.

Bilinearity and anti-symmetry of $\pb{\cdot}{\cdot}_{\G_N^*}$ are immediate from the bilinearity and anti-symmetry of $\comm{\cdot}{\cdot}_{\G_{N}}$, so it remains to verify the Jacobi identity. Let $F,G,H\in \A_{H,N}$. As we argued above, it suffices to consider the case where $G$ and $H$ are trace functionals and $F$ is a product of two trace functionals, that is, $F=F_{1}F_{2}$, where $F_{1},F_{2} \in\A_{H,N}$ are such that
\begin{equation}
F_{j}(\Gamma_N) = i\Tr\paren*{dF_{j}[0]\cdot \Gamma_N}, \qquad \forall \Gamma_N \in \G_{N}^*, \enspace j=1,2.
\end{equation}
Thus, we need to show that for all $\Gamma_N\in G_N^{*}$,
\begin{align}
0&=\pb{F}{\pb{G}{H}_{\G_{N}^{*}}}_{\G_{N}^{*}}(\Gamma_N) + \pb{G}{\pb{H}{F}_{\G_{N}^{*}}}_{\G_{N}^{*}}(\Gamma_N) + \pb{H}{\pb{F}{G}_{\G_{N}^{*}}}_{\G_{N}^{*}}(\Gamma_N) \nonumber\\
&=i\Tr\paren*{\comm{dF[\Gamma_N]}{d\pb{G}{H}_{\G_{N}^{*}}[\Gamma_N]}_{\G_{N}}\cdot\Gamma_N} + i\Tr\paren*{\comm{dG[\Gamma_N]}{d\pb{H}{F}_{\G_{N}^{*}}[\Gamma_N]}_{\G_{N}}\cdot\Gamma_N} \nonumber\\
&\phantom{=} + i\Tr\paren*{\comm{dH[\Gamma_N]}{d\pb{F}{G}_{\G_{N}^{*}}[\Gamma_N]}_{\G_{N}}\cdot\Gamma_N}.
\end{align}
We show the desired equality by direct computation:

First, since $dF[\Gamma_N]=F_{1}(\Gamma_N)dF_{2}[0] + F_{2}(\Gamma_N)dF_{1}[0]$, where we use that $F_{1}$ and $F_{2}$ have constant G\^ateaux derivatives by \cref{rem:con_GD}, it follows from the linearity of the trace that
\begin{align}
i\Tr\paren*{\comm{dF[\Gamma_N]}{d\pb{G}{H}_{\G_{N}^{*}}[\Gamma_N]}_{\G_{N}}\cdot\Gamma_N} &= iF_{1}(\Gamma_N)\Tr\paren*{\comm{dF_{2}[0]}{d\pb{G}{H}_{\G_{N}^{*}}[\Gamma_N]}_{\G_{N}}\cdot\Gamma_N} \nonumber\\
&\phantom{=} + iF_{2}(\Gamma_N)\Tr\paren*{\comm{dF_{1}[0]}{d\pb{G}{H}_{\G_{N}^{*}}[\Gamma_N]}_{\G_{N}}\cdot\Gamma_N} \nonumber\\
&=iF_{1}(\Gamma_N)\Tr\paren*{\comm{dF_{2}[0]}{\comm{dG[0]}{dH[0]}_{\G_{N}}}_{\G_{N}}\cdot\Gamma_N} \nonumber\\
&\phantom{=} + iF_{2}(\Gamma_N)\Tr\paren*{\comm{dF_{1}[0]}{\comm{dG[0]}{dH[0]}_{\G_{N}}}_{\G_{N}}\cdot\Gamma_N},
\end{align}
where we use \cref{lem:tr_gen_H} to obtain the ultimate equality.

Next, since $F$ is a product of two trace functionals, we have by \cref{lem:tr_gen_H} that
\begin{equation}
d\pb{H}{F}_{\G_{N}^{*}}[\Gamma_N] = F_{1}(\Gamma_N)\comm{dH[0]}{dF_{2}[0]}_{\G_{N}} + F_{2}(\Gamma_N)\comm{dH[0]}{dF_{1}[0]}_{\G_{N}}, \qquad \forall \Gamma_N\in \G_{N}^{*}.
\end{equation}
Hence by bilinearity of the Lie bracket and linearity of the trace,
\begin{align}
i\Tr\paren*{\comm{dG[\Gamma_N]}{d\pb{H}{F}_{\G_{N}^{*}}[\Gamma_N]}_{\G_{N}}\cdot\Gamma_N} &= iF_{1}(\Gamma_N)\Tr\paren*{\comm{dG[0]}{\comm{dH[0]}{dF_{2}[0]}_{\G_{N}}}_{\G_{N}}\cdot\Gamma_N} \nonumber\\
&\phantom{=} +  iF_{2}(\Gamma_N)\Tr\paren*{\comm{dG[0]}{\comm{dH[0]}{dF_{1}[0]}_{\G_{N}}}_{\G_{N}}\cdot\Gamma_N}.
\end{align}

Finally, similarly to the preceding case,
\begin{equation}
d\pb{F}{G}_{\G_{N}^{*}}[\Gamma_N] = F_{1}(\Gamma_N)\comm{dF_{2}[0]}{dG[0]}_{\G_{N}} + F_{2}(\Gamma_N)\comm{dF_{1}[0]}{dG[0]}_{\G_{N}},
\end{equation}
and therefore,
\begin{align}
i\Tr\paren*{\comm{dH[\Gamma_N]}{d\pb{F}{G}_{\G_{N}^{*}}[\Gamma_N]}_{\G_{N}}\cdot\Gamma_N} &= iF_{1}(\Gamma_N)\Tr\paren*{\comm{dH[0]}{\comm{dF_{2}[0]}{dG[0]}_{\G_{N}}}_{\G_{N}}\cdot\Gamma_N} \nonumber\\
&\phantom{=} + iF_{2}(\Gamma_N)\Tr\paren*{\comm{dH[0]}{\comm{dF_{1}[0]}{dG[0]}_{\G_{N}}}_{\G_{N}}\cdot\Gamma_N}.
\end{align}

Combining the preceding identities, we obtain that
\begin{align}
&i\Tr\paren*{\comm{dF[\Gamma_N]}{d\pb{G}{H}_{\G_{N}^{*}}[\Gamma_N]}_{\G_{N}}\cdot\Gamma_N} + i\Tr\paren*{\comm{dG[\Gamma_N]}{d\pb{H}{F}_{\G_{N}^{*}}[\Gamma_N]}_{\G_{N}}\cdot\Gamma_N} \nonumber\\
&\phantom{=}\qquad  + i\Tr\paren*{\comm{dH[\Gamma_N]}{d\pb{F}{G}_{\G_{N}^{*}}[\Gamma_N]}_{\G_{N}}\cdot\Gamma_N} \nonumber\\
&\phantom{=}= iF_{1}(\Gamma_N)\Tr\left(\left(\comm{dF_{2}[0]}{\comm{dG[0]}{dH[0]}_{\G_{N}}}_{\G_{N}} + \comm{dG[0]}{\comm{dH[0]}{dF_{2}[0]}_{\G_{N}}}_{\G_{N}} \right.\right. \nonumber\\
&\phantom{=}\hspace{30mm} \left.\left.+ \comm{dH[0]}{\comm{dF_{2}[0]}{dG[0]}_{\G_{N}}}_{\G_{N}}\right) \cdot\Gamma_N\right) \nonumber\\
&\phantom{=}\qquad + iF_{2}(\Gamma_N)\Tr\left(\left(\comm{dF_{1}[0]}{\comm{dG[0]}{dH[0]}_{\G_{N}}}_{\G_{N}} + \comm{dG[0]}{\comm{dH[0]}{dF_{1}[0]}_{\G_{N}}}_{\G_{N}} \right.\right.\nonumber\\
&\phantom{=}\hspace{35mm} \left.\left.+ \comm{dH[0]}{\comm{dF_{1}[0]}{dG[0]}_{\G_{N}}}_{\G_{N}}\right) \cdot\Gamma_N\right) \nonumber\\
&\phantom{=} =0,
\end{align}
where the ultimate equality follows from the fact that both lines in the penultimate equality vanish by virtue of the Jacobi identity of the Lie bracket $\comm{\cdot}{\cdot}_{\G_{N}}$.

Finally, we claim that $\pb{\cdot}{\cdot}_{\G_{N}^{*}}$ satisfies the Leibnitz rule:
\begin{equation}
\pb{FG}{H}_{\G_{N}^*}(\Gamma_N) = G(\Gamma_N)\pb{F}{H}_{\G_{N}^{*}}(\Gamma_N) + F(\Gamma_N)\pb{G}{H}_{\G_{N}^{*}}(\Gamma_N), \qquad \forall \Gamma_N\in \G_{N}^{*}.
\end{equation}
Since $d(FG)[\Gamma_N]=F(\Gamma_N)dG[\Gamma_N] + G(\Gamma_N)dF[\Gamma_N]$ by the Leibnitz rule for the G\^ateaux derivative, we see that
\begin{align}
\pb{FG}{H}_{\G_{N}^{*}}(\Gamma_N) &= i\Tr\paren*{\comm{d(FG)[\Gamma_N]}{dH[\Gamma_N]}_{\G_{N}}\cdot\Gamma_N} \nonumber\\
&=iF(\Gamma_N)\Tr\paren*{\comm{dG[\Gamma_N]}{dH[\Gamma_N]}_{\G_{N}}\cdot\Gamma_N} + iG(\Gamma_N)\Tr\paren*{\comm{dF[\Gamma_N]}{dH[\Gamma_N]}_{\G_{N}}\cdot\Gamma_N} \nonumber\\
&=F(\Gamma_N)\pb{G}{H}_{\G_{N}^{*}}(\Gamma_N) + G(\Gamma_N)\pb{F}{H}_{\G_{N}^{*}}(\Gamma_N),
\end{align}
where the penultimate equality follows by bilineariy of the Lie bracket and linearity of the trace and the ultimate equality follows from the definition of the Poisson bracket.
\end{proof}

We next verify that $\A_{H,N}$ satisfies the non-degeneracy property \ref{item:wp_P2}.

\begin{lemma}
\label{lem:H_LA_P2}
$\A_{H,N}$ satisfies property \ref{item:wp_P2} in \cref{def:WP}.
\end{lemma}
\begin{proof}
Let $\Gamma_N \in\G_{N}^{*}$ and $v \in  T_{\Gamma_N}\G_{N}^{*}$, and note that $T_{\Gamma_N}\G_{N}^{*} =\G_{N}^{*}$. Suppose that $dF[\Gamma_N](v) = 0$ for all $F\in\A_{H,N}$. We will show that $v = 0$.

Consider functionals of the form $F_{f,k_0}(\cdot ) \coloneqq i\Tr\paren*{A_{N,k_0}\cdot}$,
\begin{equation}
A_{N,k_0}^{(k)} \coloneqq
\begin{cases}
-i\ket*{f^{(k_0)}}\bra*{f^{(k_0)}}, & {k=k_0} \\
0, & {\text{otherwise}}
\end{cases},
\end{equation}
for $k_0\in\N_{\leq N}$ and $f^{(k_0)}\in\Sc_s(\R^{k_0})$. By \cref{rem:con_GD}, we have $dF_{f,k_0}[\Gamma_N](\cdot) = F_{f,k_0}(\cdot)$, so if $v = (v^{(k)})_{k \in \N_{\leq N}}  \in \G_{N}^{*} $ is as above, we have by definition of the trace that
\begin{equation}
 F_{f,k_0}(v) = \ip{v^{(k_0)} f^{(k_0)}}{f^{(k_0)}} = 0.
\end{equation}
Since $v^{(k)}$ extends uniquely to a bounded operator on $L_{s}^{2}(\R^{k})$ and $\Sc_{s}(\R^{k})$ is dense in $L_{s}^{2}(\R^{k})$, it follows from a standard polarization argument that $v^{(k)}=0$ for all $k \in \N_{\leq N}$, which completes the proof.
\end{proof}

Lastly, we show the existence of a unique Hamiltonian vector $X_H$ for $H\in \A_{H,N}$ with respect to the Poisson structure $\pb{\cdot}{\cdot}_{\G_N^*}$. With this last (most difficult) step, the proof of \cref{prop:NH_WP} will be complete.

\begin{lemma}
\label{lem:H_WP_VF}
$(\G_N^*, \A_{H,N},\pb{\cdot}{\cdot}_{\G_N^*})$ satisfies property \ref{item:wp_P3} in \cref{def:WP}. Furthermore, if $H\in\A_{H,N}$, then we have the following formula for the Hamiltonian vector field $X_H$: 
\begin{equation}
\begin{split}
X_H(\Gamma_N)^{(\ell)} &= \sum_{j=1}^{N} \sum_{r=r_0}^{\min\{\ell,j\}} C_{\ell jkrN}' \Tr_{\ell+1,\ldots,k }\paren*{\comm{\sum_{\ul{\alpha}_r\in P_r^\ell} dH[\Gamma_N]_{(\ul{\alpha}_r,\ell+1,\ldots,\min\{\ell+j-r,k \})}^{(j)}}{\gamma_N^{(k)}}},
\end{split}
\end{equation} 
where
\begin{align*}
k \coloneqq \min\{\ell+j-1,N\},  \qquad r_0 \coloneqq \max\{1,\min\{\ell,j\}-(N-\max\{\ell,j\})\}
\end{align*}
and where
\begin{equation*}
 C_{\ell jk rN}' \coloneqq {j\choose r}\frac{N C_{\ell,N}C_{j,N}}{C_{k,N} \prod_{m=1}^{r-1}(N-k+m)},
\end{equation*}
for $C_{\ell, N}, C_{k,N}$ as in \eqref{eq:Ckn_def}.
\end{lemma}
\begin{proof}
Given $F,H\in\A_{H,N}$, we first identify a candidate vector field $X_H$ by directly computing $\pb{F}{H}_{\G_N^*}$. Once we have found the candidate and verified its smoothness as a map $\G_N^*\rightarrow\G_N^*$, the proof is complete by the uniqueness guaranteed by \cref{rem:hvf_u}.

By definition of the Poisson bracket on $\G_N^*$, we have that
\begin{align}
\pb{F}{H}_{\G_N^*}(\Gamma_N) &= i\Tr\paren*{\comm{dF[\Gamma_N]}{dH[\Gamma_N]}_{\G_N}\cdot \Gamma_N} \nonumber\\
&=i\sum_{k=1}^N \Tr_{1,\ldots,k}\paren*{\comm{dF[\Gamma_N]}{dH[\Gamma_N]}_{\G_N}^{(k)}\gamma_N^{(k)}},
\end{align}
for $\Gamma_N=(\gamma_N^{(k)})_{k=1}^{N} \in \G_N^*$. Using the linearity of the $\Sym_k$ operator, we have by the formula from \cref{prop:Lie_hi_form} that
\begin{align*}
\comm{dF[\Gamma_N]}{dH[\Gamma_N]}_{\G_N}^{(k)} &= \sum_{{1\leq \ell,j\leq N}\atop {\min\{\ell+j-1,N\}=k}}\sum_{r=r_0}^{\min\{\ell,j\}} C_{\ell jkrN} \Sym_k\paren*{\comm{dF[\Gamma_N]^{(\ell)}}{dH[\Gamma_N]^{(j)}}_r},
\end{align*}
and
\begin{equation*}
\begin{split}
\Sym_k\paren*{\comm{dF[\Gamma_N]^{(\ell)}}{dH[\Gamma_N]^{(j)}}_r} &= \Sym_k\biggl({j\choose r} dF[\Gamma_N]_{(1,\ldots,\ell)}^{(\ell)}\biggl(\sum_{\ul{\alpha}_r\in P_r^\ell}dH[\Gamma_N]_{(\ul{\alpha}_r,\ell+1,\ldots,\ell+j-r)}^{(j)}\biggr)\biggr)  \\
&\phantom{=} - \Sym_k\biggl({\ell\choose r} dH[\Gamma_N]_{(1,\ldots,j)}^{(j)}\biggl(\sum_{\ul{\alpha}_r\in P_r^j} dF[\Gamma_N]_{(\ul{\alpha}_r,j+1,\ldots,j+\ell-r)}^{(\ell)}\biggr)\biggr),
\end{split}
\end{equation*}
where we have used the combinatorial notation $C_{\ell jkrN}$ defined in \cref{eq:comb_not}. Recall from \cref{rem:tbyt} that we are justified in writing
\begin{equation}
dH[\Gamma_N]_{(1,\ldots,j)}^{(j)}\biggl(\sum_{\ul{\alpha}_r\in P_r^j} dF[\Gamma_N]_{(\ul{\alpha}_r,j+1,\ldots,j+\ell-r)}^{(\ell)}\biggr) = \sum_{\ul{\alpha}_r\in P_r^j} dH[\Gamma_N]_{(1,\ldots,j)}^{(j)}dF[\Gamma_N]_{(\ul{\alpha}_r,j+1,\ldots,j+\ell-r)}^{(\ell)}.
\end{equation}
Let $(m_1,\ldots,m_{j-r})$ be the increasing arrangement of the set $\N_{\leq j}\setminus \{\alpha_1,\ldots,\alpha_r\}$. Defining the permutation $\tau\in\mathbb{S}_k$ by the formula
\begin{equation}
\tau(a) \coloneqq
\begin{cases}
i, & {a=\alpha_i \text{ for $1\leq i\leq  r$}} \\
a-j+r, & {j+1\leq a \leq j+\ell-r} \\
\ell+i, & {a = m_i \text{ for $1\leq i\leq j-r$}} \\
a, & {\text{otherwise}}
\end{cases},
\end{equation}
we find that for each $\ul{\alpha}_r\in P_r^j$ fixed,
\begin{equation}
\paren*{dH[\Gamma_N]_{(1,\ldots,j)}^{(j)}dF[\Gamma_N]_{(\ul{\alpha}_r,j+1,\ldots,j+\ell-r)}^{(\ell)}}_{(\tau(1),\ldots,\tau(k))} = dH[\Gamma_N]_{(1,\ldots,r,\ell+1,\ldots,\ell+j-r)}^{(j)}dF[\Gamma_N]_{(1,\ldots,\ell)}^{(\ell)}.
\end{equation}
Since the $\Sym_k$ operator is $\mathbb{S}_k$-invariant, it then follows that
\begin{equation}
\Sym_k\paren*{dH[\Gamma_N]_{(1,\ldots,j)}^{(j)}dF[\Gamma_N]_{(\ul{\alpha}_r,\ell+1,\ldots,\ell+j-r)}^{(\ell)}} = \Sym_k\paren*{dH[\Gamma_N]_{(1,\ldots,r,\ell+1,\ldots,\ell+j-r)}^{(j)}dF[\Gamma_N]_{(1,\ldots,\ell)}^{(\ell)}}.
\end{equation}
Consequently, using that $|P_r^j| = {j\choose r}r!$, we obtain that
\begin{equation}
\begin{split}
&\Sym_k\biggl({\ell\choose r} dH[\Gamma_N]_{(1,\ldots,j)}^{(j)}\biggl(\sum_{\ul{\alpha}_r\in P_r^j} dF[\Gamma_N]_{(\ul{\alpha}_r,j+1,\ldots,j+\ell-r)}^{(\ell)}\biggr)\biggr) \\
&={\ell\choose r}{j\choose r}r!\Sym_k\paren*{dH[\Gamma_N]_{(1,\ldots,r,\ell+1,\ldots,\ell+j-r)}^{(j)}dF[\Gamma_N]_{(1,\ldots,\ell)}^{(\ell)}}.
\end{split}
\end{equation}
Now given $\ul{\alpha}_r\in P_r^\ell$, let $(m_1,\ldots,m_{\ell-r})$ be the increasing arrangement of the set $\N_{\leq \ell} \setminus\{\alpha_1,\ldots,\alpha_r\}$. We recycle notation to define a new permutation $\tau\in\Ss_k$ by
\begin{equation}
\tau(i) \coloneqq 
\begin{cases}
\alpha_i, & {1\leq i\leq r} \\
m_{i-r}, & {r+1\leq i\leq \ell}\\
i, & {\text{otherwise}}
\end{cases}.
\end{equation}
Then
\begin{equation}
\begin{split}
&\Sym_k\paren*{\paren*{dH[\Gamma_N]_{(1,\ldots,r,\ell+1,\ldots,\ell+j-r)}^{(j)}dF[\Gamma_N]_{(1,\ldots,\ell)}^{(\ell)}}_{(\tau(1),\ldots,\tau(k))}}\\
&= \Sym_k\paren*{dH[\Gamma_N]_{(\ul{\alpha}_r,\ell+1,\ldots,\ell+j-r)}^{(j)}dF[\Gamma_N]_{(1,\ldots,\ell)}^{(\ell)}},
\end{split}
\end{equation}
where we can ``undo'' the permutation $\tau$'s effect on $dF[\Gamma_N]_{(1,\ldots,\ell)}^{(\ell)}$ by its $\Ss_\ell$-invariance. Using that $|P_r^\ell|={\ell\choose r}r!$, we obtain that
\begin{equation}
\begin{split}
&{\ell\choose r}{j\choose r}r!\Sym_k\paren*{dH[\Gamma_N]_{(1,\ldots,r,\ell+1,\ldots,\ell+j-r)}^{(j)}dF[\Gamma_N]_{(1,\ldots,\ell)}^{(\ell)}} \\
&={j\choose r}\sum_{\ul{\alpha}_r\in P_r^\ell} \Sym_k\paren*{dH[\Gamma_N]_{(\ul{\alpha}_r,\ell+1,\ldots,\ell+j-r)}^{(j)}dF[\Gamma_N]_{(1,\ldots,\ell)}^{(\ell)}}.
\end{split}
\end{equation}
Substituting the preceding identity into the expression $\Tr_{1,\ldots,k}(\comm{dF[\Gamma_N]}{dH[\Gamma_N]}_{\G_N}^{(k)}\gamma_N^{(k)})$ and using \cref{lem:tr_bos} to eliminate the $\Sym_k$ operator, we obtain that
\begin{equation}
\begin{split}
&i\Tr_{1,\ldots,k}\paren*{\comm{dF[\Gamma_N]}{dH[\Gamma_N]}_{\G_N}^{(k)}\gamma_N^{(k)}} \\
&=i\sum_{\min\{\ell+j-1,N\}=k}\sum_{r=r_0}^{\min\{\ell,j\}}C_{\ell jkrN}{j\choose r} \sum_{{\ul{\alpha}_r\in P_r^\ell}} \Biggl(\Tr_{1,\ldots,k}\paren*{dF[\Gamma_N]_{(1,\ldots,\ell)}^{(\ell)}dH[\Gamma_N]_{(\ul{\alpha}_r,\ell+1,\ldots,\ell+j-r)}^{(j)}\gamma_N^{(k)} } \\
&\hspace{75mm} - \Tr_{1,\ldots,k}\paren*{dH[\Gamma_N]_{(\ul{\alpha}_r,\ell+1,\ldots,\ell+j-r)}^{(j)}dF[\Gamma_N]_{(1,\ldots,\ell)}^{(\ell)}\gamma_N^{(k)}} \Biggr).
\end{split}
\end{equation}
Since $dH[\Gamma_N]_{(\ul{\alpha}_r,\ell+1,\ldots,\ell+j-r)}^{(j)}$ is skew-adjoint and therefore by duality extends to an element in $\L(\Sc_s'(\R^k),\Sc'(\R^k))$, it follows from the cyclicity property of \cref{prop:gtr_prop}\ref{item:gtr_cyc} that
\begin{equation}
\begin{split}
&\Tr_{1,\ldots,k}\paren*{dH[\Gamma_N]_{(\ul{\alpha}_r,\ell+1,\ldots,\ell+j-r)}^{(j)}dF[\Gamma_N]_{(1,\ldots,\ell)}^{(\ell)}\gamma_N^{(k)}}\\
&=\Tr_{1,\ldots,k}\paren*{dF[\Gamma_N]_{(1,\ldots,\ell)}^{(\ell)}(\gamma_N^{(k)}dH[\Gamma_N]_{(\ul{\alpha}_r,\ell+1,\ldots,\ell+j-r)}^{(j)})}.
\end{split}
\end{equation}
Since
\begin{equation}
dH[\Gamma_N]_{(\ul{\alpha}_r,\ell+1,\ldots,\ell+j-r)}^{(j)}\gamma_N^{(k)}, \,\,\, \gamma_N^{(k)}dH[\Gamma_N]_{(\ul{\alpha}_r,\ell+1,\ldots,\ell+j-r)}^{(j)} \in \L(\Sc_s'(\R^k),\Sc(\R^k)),
\end{equation}
the usual partial trace $\Tr_{\ell+1,\ldots,k}$ of each of these operators exists and defines an element of $\L(\Sc_s'(\R^\ell),\Sc(\R^\ell))$. Moreover, since $dH[\Gamma_N]^{(j)}$ and $\gamma_N^{(k)}$ are skew- and self-adjoint, respectively, these partial traces are self-adjoint.

Returning to the expression $i\Tr(\comm{dF[\Gamma_N]}{dH[\Gamma_N]}_{\G_N}\cdot\Gamma_N)$ and interchanging the order of the $k$ and $\ell$ summations, we see that
\begin{equation*}
\begin{split}
&\sum_{k=1}^N i\Tr_{1,\ldots,k}\paren*{\comm{dF[\Gamma_N]}{dH[\Gamma_N]}_{\G_N}^{(k)}\gamma_N^{(k)}} \\
&=i\sum_{\ell=1}^{N}\sum_{j=1}^{N} \sum_{r=r_0}^{\min\{\ell,j\}} C_{lj\tl{k}rN}'  \Biggl( \Tr_{1,\ldots,\ell}\biggl(dF[\Gamma_N]^{(\ell)}\biggl(\sum_{\ul{\alpha}_r\in P_r^\ell}\Tr_{\ell+1,\ldots,\tl{k}}\biggl(dH[\Gamma_N]_{(\ul{\alpha}_r,\ell+1,\ldots,\min\{\ell+j-r,\tl{k}\})}^{(j)}\gamma_N^{(\tl{k})}\biggr)\biggr)\biggr) \\
&\hspace{50mm} - \Tr_{1,\ldots,\ell}\biggl(dF[\Gamma_N]^{(\ell)}\biggl(\sum_{\ul{\alpha}_r\in P_r^\ell}\Tr_{\ell+1,\ldots,\tl{k}}\biggl(\gamma_N^{(\tl{k})}dH[\Gamma_N]_{(\ul{\alpha}_r,\ell+1,\ldots,\min\{\ell+j-r,\tl{k}\})}^{(j)}\biggr)\biggr)\biggr) \Biggr),
\end{split}
\end{equation*}
where
\begin{align}
\tl{k} &\coloneqq \min\{\ell+j-1,N\},\\
C_{\ell j\tl{k}rN}' &\coloneqq \frac{N C_{\ell,N}C_{j,N}}{C_{\tl{k},N}\prod_{m=1}^{r-1}(N-\tl{k}+m)}{j\choose r}.
\end{align}
Note that since $\gamma_N^{(\tl{k})}$ admits a decomposition
\begin{equation}
\gamma_N^{(\tl{k})} = \sum_{m=1}^\infty \lambda_m \ket*{f_m^{(\tl{k})}}\bra*{f_m^{(\tl{k})}},
\end{equation}
where $\sum_{m=1}^\infty |\lambda_m|\leq 1$ and $f_m^{(\tl{k})}, g_m^{(\tl{k})}$ converge to zero in $\Sc_s(\R^{\tl{k}})$, we see that
\begin{equation}
\begin{split}
&\Tr_{\ell+1,\ldots,\tl{k}}\paren*{\gamma_N^{\tl{k}}dH[\Gamma_N]_{(\ul{\alpha}_r,\ell+1,\ldots,\min\{\ell+j-r,\tl{k}\})}^{(j)}}\\
&=\sum_{m=1}^\infty \lambda_m \ip{f_m^{(\tl{k})}}{dH[\Gamma_N]_{(\ul{\alpha}_r,\ell+1,\ldots,\min\{\ell+j-r,\tl{k}\})}^{(j)}f_m^{(\tl{k})}},
\end{split}
\end{equation}
which is independent of the choice of extension of $dH[\Gamma_N]^{(j)}$ to domain $\Sc(\R^j)$ by the permutation invariance of each $f_m^{(\tl{k})}$. Furthermore, the operator
\begin{equation}
\sum_{\ul{\alpha}_r\in P_r^\ell}\Tr_{\ell+1,\ldots,\tl{k}}\paren*{\gamma_N^{(\tl{k})}dH[\Gamma_N]_{(\ul{\alpha}_r,\ell+1,\ldots,\min\{\ell+j-r,\tl{k}\})}^{(j)}}
\end{equation}
is invariant under the $\Ss_\ell$ action, since $P_r^{\ell}$ is invariant under the $\Ss_{\ell}$ group action. Hence, it maps into $\Sc_s(\R^\ell)$, and its left-composition with $dF[\Gamma_N]^{(\ell)}$ is well-defined.

Using the bilinearity of the generalized trace, we obtain the candidate Hamiltonian vector field
\begin{equation}
\begin{split}
X_H(\Gamma_N)^{(\ell)} &\coloneqq \sum_{j=1}^{N} \sum_{r=r_0}^{\min\{\ell,j\}} C_{\ell j\tl{k}rN}' \sum_{\ul{\alpha}_r\in P_r^\ell}\Biggl(\Tr_{\ell+1,\ldots,\tl{k}}\paren*{dH[\Gamma_N]_{(\ul{\alpha}_r,\ell+1,\ldots,\min\{\ell+j-r,\tl{k}\})}^{(j)}\gamma_N^{(\tl{k})}} \\
&\hspace{50mm} - \Tr_{\ell+1,\ldots,\tl{k}}\paren*{\gamma_N^{(\tl{k})}dH[\Gamma_N]_{(\ul{\alpha}_r,\ell+1,\ldots,\min\{\ell+j-r,\tl{k}\})}^{(j)}} \Biggr).
\end{split}
\end{equation}

We now verify that $X_H$, as defined above, is a smooth map $\G_N^*\rightarrow \G_N^*$, so that we may conclude the proof by \cref{rem:con_GD}. We claim that the right-hand side of the preceding identity defines a continuous linear (hence, smooth) map 
\begin{equation}
\G_N^*\rightarrow \bigoplus_{k=1}^N \L(\Sc_s'(\R^k),\Sc_s(\R^k)).
\end{equation}
Linearity is obvious, and the map is continuous from 
\[
\G_N^*\rightarrow \bigoplus_{k=1}^{N} \L(\Sc_s'(\R^k),\Sc(\R^k))
\]
by \cref{prop:partial_trace}. That we may replace the target $\Sc(\R^k)$ by the bosonic subspace $\Sc_s(\R^k)$ is a consequence of the following facts: $P_r^\ell$ is invariant under the $\mathbb{S}_\ell$ group action, $dH[\Gamma_N]^{(j)}$ is $\mathbb{S}_j$-invariant, and $\gamma_N^{(\tl{k})}$ is a fortiori $\mathbb{S}_\ell$-invariant. The self-adjointness of $X_H(\Gamma_N)^{(\ell)}$ follows from the skew- and self-adjointness of $dH[\Gamma_N]^{(j)}$ and $\gamma_N^{(\tl{k})}$, respectively, and the adjoint properties of the generalized partial trace.
\end{proof}

\subsection{Density matrix maps as Poisson morphisms}
\label{ssec:dm_pomo}
We close this section with the observations that the well-known operations of forming a density matrix out of a wave function and forming an $N$-hierarchy of reduced density matrices from an $N$-body density matrix respect the geometric structure we have developed, in the sense that these operations define Poisson morphisms.

We first define the \emph{density matrix map} or \emph{ket-bra map} from $N$-body bosonic wave functions to $N$-body bosonic density matrices.

\begin{mydef}[Density matrix map]
We define the \emph{density matrix map} or \emph{ket-bra map} by
\begin{equation}
\iota_{DM,N}: \mathcal{S}_{s}(\R^{N}) \rightarrow \g_N^* \qquad \iota_{DM,N}(\Phi_{N}) \coloneqq \ket{\Phi_{N}}\bra{\Phi_{N}} = \Phi_{N} \otimes \ol{\Phi_{N}}.
\end{equation}
\end{mydef}

It is easy to verify that $\iota_{DM,N}$ is a smooth map from $\mathcal{S}_{s}(\R^{N})$ to $\g_{N}^{*}$. We now show that the density matrix map is a Poisson map. To prove this property, we recall from \cref{def:po_map} the requirement that $\iota_{DM,N}^*\A_{DM,N}\subset \A_{\Sc}$. If $F$ is smooth, then the smoothness of $\iota_{DM,N}$ implies by the chain rule that $f = F \circ \iota_{DM,N} \in C^{\infty}(\mathcal{S}_{s}(\R^{N});\R)$. However, it is not a priori clear that $f\in \A_{\Sc}$, where we recall that  $A_{\Sc}\subset C^{\infty}(\Sc(\R^{N});\R)$ is defined by
\begin{equation}\label{equ:Asc}
\A_{\Sc} \coloneqq \bigl \{ H \,:\, \grad_{s}H \in C^{\infty}(\Sc(\R^N);\Sc(\R^N)) \bigr\} ,
\end{equation}
In the sequel, we will use the notation $\A_{\Sc, N}$ to make the dependence on $N$ explicit. 

\begin{lemma}\label{lem:f_adm}
Let $N\in\N$. For any $F\in\A_{DM,N}$, the functional $f \coloneqq F\circ \iota_{DM,N} \in C^{\infty}(\mathcal{S}_{s}(\R^{N});\R)$ belongs to $\A_{\Sc, N}$. Furthermore,
\begin{equation}
\grad_s f(\Phi_N) = dF[\iota_{DM,N}(\Phi_N)](\Phi_N), \qquad \forall \Phi_N\in\Sc_s(\R^N),
\end{equation}
where we identify $dF[\iota_{DM,N}(\Phi_N)]$ as a skew-adjoint operator by \cref{rem:GD_Op_DM}.
\end{lemma}
\begin{proof}
Observe from the chain rule that for $\Phi_{N},\delta\Phi_{N}\in\mathcal{S}_{s}(\R^{N})$,
\begin{align}
df[\Phi_{N}](\delta\Phi_{N}) &= dF[\iota_{DM,N}(\Phi_{N})]\paren*{d\iota_{DM,N}[\Phi_{N}](\delta\Phi_{N})} \nonumber\\
&=dF[\iota_{DM,N}(\Phi_{N})]\paren*{\ket{\Phi_{N}}\bra{\delta\Phi_{N}} + \ket{\delta\Phi_{N}}\bra{\Phi_{N}}},
\end{align}
where we use the elementary computation
\begin{equation}
\label{diota_comp}
d\iota_{DM,N}[\Phi_{N}](\delta\Phi_{N}) =\ket*{\Phi_{N}}\bra*{\delta\Phi_{N}} + \ket*{\delta\Phi_{N}}\bra*{\Phi_{N}}.
\end{equation}
Identifying the functional $dF[\iota_{DM,N}(\Phi_{N})](\cdot)$ with a skew-adjoint DVO given by $dF[\iota_{DM,N}(\Phi_{N})]$ as in \cref{rem:GD_Op_DM}, we have that
\begin{align*}
dF[\iota_{DM,N}(\Phi_{N})]\paren*{\ket{\Phi_{N}}\bra{\delta\Phi_{N}} + \ket{\delta\Phi_{N}}\bra{\Phi_{N}}} &= i\Tr_{1,\ldots,N}\paren*{dF[\iota_{DM,N}(\Phi_{N})]\paren*{\ket{\Phi_{N}}\bra{\delta\Phi_{N}} + \ket{\delta\Phi_{N}}\bra{\Phi_{N}}}} \nonumber\\
&=i\ip{\delta\Phi_N}{dF[\iota_{DM,N}(\Phi_N)]\Phi_N} + i\ip{\Phi_N}{dF[\iota_{DM,N}[\Phi_N]\delta\Phi_N} .
\end{align*}
Since $dF[\iota_{DM,N}(\Phi_{N})]$ is skew-adjoint, the preceding expression equals
\begin{align*}
i\ip{\delta\Phi_N}{dF[\iota_{DM,N}(\Phi_N)]\Phi_N}- i\ip{dF[\iota_{DM,N}(\Phi_N)]\Phi_N}{\delta\Phi_N} &=-2\Im \ip{\delta\Phi_N}{dF[\iota_{DM,N}(\Phi_N)]\Phi_N} \nonumber\\
&=\omega_{L^2}(dF[\iota_{DM,N}(\Phi_N)]\Phi_N, \delta\Phi_N).
\end{align*}
We claim that the map $\Phi_N \mapsto dF[\iota_{DM,N}(\Phi_N)]\Phi_N$ is a smooth map of $\Sc_s(\R^N)$ to itself, which justifies our preceding manipulations. Indeed, suppose first that $F\in\A_{DM,N}$ is a trace functional. Then $dF[\iota_{DM,N}(\Phi_N)]=dF[0]$, and therefore the claim follows since $dF[0]$ is a continuous linear map of $\Sc_s(\R^N)$ to itself by definition of $\A_{DM,N}$. The general case then follows by the Leibnitz rule for the G\^ateaux derivative. Therefore, the functional $f$ has symplectic $L^2$ gradient 
\[
\grad_s f(\Phi_N) = dF[\iota_{DM,N}(\Phi_N)]\Phi_N,
\]
 and $\grad_s f$ is a smooth map of $\Sc_s(\R^N)$ to itself, which implies that $f\in\A_{\Sc,N}$.
\end{proof}

We recall from \eqref{pb_intro} the definition for $\pb{\cdot}{\cdot}_{L^{2}}$, and we consider the rescaled Poisson bracket
\begin{equation}
\pb{\cdot}{\cdot}_{L^{2},N} \coloneqq N\pb{\cdot}{\cdot}_{L^{2}}.
\end{equation}

\begin{prop}
\label{prop:DM_po}
Let $N\in\N$. Then 
\begin{equation}
\iota_{DM,N}: (\mathcal{S}_{s}(\R^{N}), \A_{\Sc,N}, \pb{\cdot}{\cdot}_{L^{2},N}) \rightarrow (\g_{N}^{*}, \A_{DM,N}, \pb{\cdot}{\cdot}_{\g_{N}^{*}})
\end{equation}
is a Poisson map.
\end{prop}
\begin{proof}
As observed above, the smoothness of $\iota_{DM,N}$ is evident, and by \cref{lem:f_adm}, $F\circ \iota_{DM,N} \in \A_{\Sc,N}$ for any $F\in\A_{DM,N}$. Hence, it remains for us to show that for all $F,G\in \A_{DM,N}$,
\begin{equation}
\pb{F\circ\iota_{DM,N}}{G\circ \iota_{DM,N}}_{L^{2},N}(\Phi_N) = \pb{F}{G}_{\g_{N}^{*}}\circ \iota_{DM,N}(\Phi_N), \qquad \forall \,\Phi_N\in\mathcal{S}_{s}(\R^{N}).
\end{equation}
For convenience, we introduce the notation $f\coloneqq F\circ \iota_{DM,N}$ and $g\coloneqq G\circ\iota_{DM,N}$.  We first consider the expression $\pb{f}{g}_{L^{2},N}(\Phi_N)$. Observe that by definition of the Poisson bracket $\pb{\cdot}{\cdot}_{L^2,N}$,
\begin{align}
\pb{f}{g}_{L^{2},N}(\Phi_{N}) &=  N\omega_{L^2}(\grad_s f(\Phi_{N}),\grad_s g(\Phi_{N}))  \nonumber \\
&= 2N\Im\ip{dF[\iota_{DM,N}(\Phi_N)]\Phi_N}{dG[\iota_{DM,N}(\Phi_N)]\Phi_N}.
\end{align}
Now using the skew-adjointness of $dG[\iota_{DM,N}(\Phi_{N})]$ and $dF[\iota_{DM,N}(\Phi_{N})]$, we conclude that the last expression equals
\begin{align}
& iN\paren*{\ip{\Phi_N}{dF[\iota_{DM,N}(\Phi_N) dG[\iota_{DM,N}(\Phi_N)]\Phi_N} - \ip{\Phi_N}{dG[\iota_{DM,N}(\Phi_N)]dF[\iota_{DM,N}(\Phi_N)]\Phi_N}}  \nonumber\\
&=i\Tr_{1,\ldots,N}\paren*{\comm{dF[\iota_{DM,N}(\Phi_N)]}{dG[\iota_{DM,N}(\Phi_N)]}_{\g_N}\ket*{\Phi_N}\bra*{\Phi_N}} \nonumber\\
&=\pb{F}{G}_{\g_{N}^{*}} \circ\iota_{DM,N}(\Phi_{N}),
\end{align}
which is exactly what we wanted to show.
\end{proof}

We next show that there is a linear homomorphism of Lie algebras $\G_{N}\rightarrow \g_{N}$ induced by the embeddings $\{\epsilon_{k,N}\}_{k\in\N_{\leq N}}$. We will then combine this fact with a duality argument to prove that the reduced density matrix operation is a Poisson mapping 
\begin{equation}
(\g_N^*,\A_{DM,N},\pb{\cdot}{\cdot}_{\g_N^*}) \rightarrow (\G_N^*, \A_{H,N}, \pb{\cdot}{\cdot}_{\G_N^*}).
\end{equation}

\begin{prop}
\label{prop:LA_hom_sum}
For any $N\in \N$, the map
\begin{equation}
\iota_{\ep,N}: \G_{N} \rightarrow \g_{N}, \qquad \iota_{\ep,N}(A_N) \coloneqq \sum_{k=1}^{N} \epsilon_{k,N}(A_N^{(k)}),
\end{equation}
is a continuous linear homomorphism of Lie algebras.
\end{prop}
\begin{proof}
Continuity and linearity are evident from the continuity and linearity of the maps $\ep_{k,N}$ (recall \cref{lem:ep_con}). To show that $\iota_{sum,N}$ is a homomorphism of Lie algebras, we need to show that for any 
\begin{equation}
A_N=(A_N^{(k)})_{k\in\N_{\leq N}},\,\, B_N=(B_N^{(k)})_{k\in\N_{\leq N}}\in\G_{N},
\end{equation}
we have that
\begin{equation}
\iota_{\ep,N}\paren*{\comm{A_N}{B_N}_{\G_{N}}} = \comm{\iota_{\ep,N}(A_N)}{\iota_{\ep,N}(B_N)}_{\g_{N}}.
\end{equation}

Consider the left-hand side expression. By the definition of the map $\iota_{\ep,N}$, the definition of the Lie bracket $\comm{\cdot}{\cdot}_{\G_{N}}$ from \cref{lie_def}, and \cref{lem:hi_fil}, we obtain that
\begin{align}
\iota_{\ep,N}\paren*{\comm{A_N}{B_N}_{\G_{N}}} &= \sum_{k=1}^{N} \epsilon_{k,N}\paren*{\comm{A_N}{B_N}_{\G_{N}}^{(k)}} \nonumber\\
&=\sum_{k=1}^{N}\epsilon_{k,N}(C_N^{(k)}) \nonumber\\
&=\sum_{k=1}^{N} \sum_{{1\leq \ell,j\leq N}\atop {\min\{\ell+j-1,N\}=k}} \comm{\epsilon_{\ell,N}(A_N^{(\ell)})}{\epsilon_{j,N}(B_N^{(j)})}_{\g_{N}} \nonumber.
\end{align}
Using the partition
\begin{equation}
\{(\ell,j) \in (\N_{\leq N})^2 \} = \bigcup_{k=1}^{N} \{(\ell,j)\in  (\N_{\leq N})^2: \min\{\ell+j-1,N\}=k\},
\end{equation}
we see that
\begin{equation}
\sum_{k=1}^{N} \sum_{{1\leq \ell,j\leq N}\atop {\min\{\ell+j-1,N\}=k}} \comm{\epsilon_{\ell,N}(A_N^{(\ell)})}{\epsilon_{j,N}(B_N^{(j)})}_{\g_{N}}  = \sum_{\ell=1}^{N}\sum_{j=1}^{N} \comm{\epsilon_{\ell,N}(A_N^{(\ell)})}{\epsilon_{j,N}(B_N^{(j)})}_{\g_{N}}.
\end{equation}
By the definition of the map $\iota_{\ep,N}$ and the bilinearity of Lie brackets, we observe that
\begin{equation}
\sum_{\ell=1}^{N}\sum_{j=1}^{N}\comm{\epsilon_{\ell,N}(A_N^{(\ell)})}{\epsilon_{j,N}(B_N^{(j)})}_{\g_{N}} = \comm{\iota_{\ep,N}(A_N)}{\iota_{\ep,N}(B_N)}_{\g_{N}} ,
\end{equation}
which completes the proof.
\end{proof}

Finally, we show that there is a canonical Poisson mapping of $\g_{N}^{*} \rightarrow \G_{N}^{*}$ given by taking the sequence of reduced density matrices.

\begin{prop}[RDM Map is Poisson]\label{prop:RDM_Po}
The map $\iota_{RDM,N}: \g_{N}^{*} \rightarrow \G_{N}^{*}$ given by
\begin{equation}
\iota_{RDM,N}(\Psi_{N}) \coloneqq \Gamma_{N} = (\gamma_{N}^{(k)})_{k\in\N_{\leq N}}, \qquad \gamma_{N}^{(k)} \coloneqq \Tr_{k+1,\ldots,N}\paren*{\Psi_{N}}
\end{equation}
is a Poisson map.
\end{prop}

To prove \cref{prop:RDM_Po}, we will show that $\iota_{RDM,N}$ is the dual of the map $\iota_{sum,N}$, which, by \cref{prop:LA_hom_sum}, we know is a continuous linear homomorphism of Lie algebras. We then appeal to the following general result, the statement of which we have taken from \cite[Proposition 10.7.2]{MR2013}.

\begin{lemma}
Let $(\g,\comm{\cdot}{\cdot}_{\g})$ and $(\h,\comm{\cdot}{\cdot}_{\h})$ be Lie algebras. Let $\alpha:\g \rightarrow \h$ be a linear map. Then the map $\alpha$ is a homomorphism of Lie algebras if and only if its dual map $\alpha^{*}: \h^{*} \rightarrow \g^{*}$ is a (linear) Poisson map.
\end{lemma}
\begin{proof}[Proof of \cref{prop:RDM_Po}]
As stated above, we want to show that the reduced density matrix $\iota_{RDM,N}$ is the dual of the map
\begin{equation}
\iota_{\ep,N}:\G_{N} \rightarrow \g_{N}, \qquad A_N = (A_N^{(1)},\ldots,A_N^{(N)}) \mapsto \sum_{k=1}^{N} \epsilon_{k,N}(A_N^{(k)}).
\end{equation}
Indeed, observe that for $\Psi_{N}\in \g_{N}^{*}$ and $A_N=(A_N^{(k)})_{k\in\N_{\leq N}}\in \G_{N}$, we see from unpacking the definition of $\iota_{\ep,N}$ and using the bilinearity of the generalized trace that
\begin{equation}
 \iota_{\ep,N}^{*}(\Psi_{N})(A_N) = i\Tr_{1,\ldots,N}\paren*{ \iota_{\ep,N}(A_N)\Psi_{N}} =\sum_{k=1}^{N} i\Tr_{1,\ldots,N}\paren*{\epsilon_{k,N}(A_N^{(k)})\Psi_{N}}.
\end{equation}
Unpacking the definition \eqref{eps_def} of the map $\epsilon_{k,N}(A_N^{(k)})$ and using the bilinearity of the generalized trace again, we see that
\begin{equation}
\sum_{k=1}^{N} i\Tr_{1,\ldots,N}\paren*{\epsilon_{k,N}(A_N^{(k)})\Psi_{N}} = \sum_{k=1}^{N}\sum_{\ul{p}_k\in P_{k}^{N}} iC_{k,N}\Tr_{1,\ldots,N}\paren*{A_{N,(p_1,\ldots,p_k)}^{(k)}\Psi_N}.
\end{equation}
Hence using that $\Psi_{N}$ is bosonic and \cref{lem:tr_bos}, we have that
\begin{align}
\Tr_{1,\ldots,N}\paren*{A_{N,(p_1,\ldots,p_k)}^{(k)}\Psi_N} = \Tr_{1,\ldots,N}\paren*{A_{N,(1,\ldots,k)}^{(k)}\Psi_N} &= \Tr_{1,\ldots,k}\paren*{A_N^{(k)} \Tr_{k+1,\ldots,N}\paren*{\Psi_{N}}} \nonumber\\
&= \Tr_{1,\ldots,k}\paren*{A_N^{(k)}\gamma_{N}^{(k)}},
\end{align}
where the ultimate equality follows by definition of $\gamma_N^{(k)}$. Since $|P_{k}^{N}| = 1/C_{k,N}$, we conclude that
\begin{equation}
\iota_{\ep,N}^*(\Psi_N)(A_N) = \sum_{k=1}^{N} i\Tr_{1,\ldots,k}\paren*{A_N^{(k)}\gamma_{N}^{(k)}} = i\Tr\paren*{A_N\cdot\iota_{RDM,N}(\Psi_{N})},
\end{equation}
which completes the proof of the proposition.
\end{proof}

\section{Geometric structure for infinity hierarchies} \label{sec:geom}
In this section, we compute the limit of the $N$-body Lie algebra $(\G_N,\comm{\cdot}{\cdot}_{\G_N})$ as $N\rightarrow\infty$. We then show that in this limit, the higher-order contractions appearing in formula \eqref{eq:N_LB_form} vanish. Consequently, the domain of definition of the Lie bracket may be enlarged, for which we construct the Lie algebra $(\G_{\infty},\comm{\cdot}{\cdot}_{\G_{\infty}})$ of observable $\infty$-hierarchies and dually, the weak Lie-Poisson manifold $(\G_{\infty}^{*},\A_{\infty},\pb{\cdot}{\cdot}_{\G_{\infty}^{*}})$ of density matrix $\infty$-hierarchies.

\subsection{The limit of $\G_N$ as $N\rightarrow\infty$}
In order to pass from the $N$-particle setting to the $\infty$-particle setting, we first study the limit of the Lie algebra $(\G_N,\comm{\cdot}{\cdot}_{\G_N})$ as $N\rightarrow \infty$.

Via the natural inclusion map, we can identify $\G_N$ as the subspace of the locally convex direct sum
\begin{equation}\label{colim2}
\mathfrak{F}_\infty \coloneqq \bigcup_{N=1}^\infty \G_N =  \bigoplus_{k=1}^\infty \g_k
\end{equation}
consisting of elements $A=(A^{(k)})_{k\in\N}$, where $A^{(k)} = 0$ for $k\geq N+1$. In our next result, \cref{prop:LB_lim}, we establish a formula for the limiting bracket structure for $\G_\infty$.

\LBlim*

\begin{proof}
Let $k\in\N$. For $M\gg k$, we have by \cref{prop:Lie_hi_form} and the linearity of the map $\epsilon_{k,N}$ that
\begin{align}
&\sum_{{ \ell,j\geq 1}\atop {\ell+j-1=k}} \epsilon_{k,M}^{-1}\paren*{\comm{\epsilon_{\ell,M}(A^{(\ell)})}{\epsilon_{j,M}(B^{(j)})}_{\g_{M}}} \nonumber\\
&=\sum_{{ \ell,j\geq 1}\atop {\ell+j-1=k}}\Sym_k\paren*{\sum_{r=1}^{\min\{\ell,j\}} \frac{M C_{\ell,M}C_{j,M}}{C_{k,M} \prod_{a=1}^{r-1}(M-k+a)} \comm{A^{(\ell)}}{B^{(j)}}_r} \nonumber\\
&=\sum_{{ \ell,j\geq 1}\atop {\ell+j-1=k}}\Sym_k\paren*{\frac{MC_{\ell,M}C_{j,M}}{C_{k,M}} \comm{A^{(\ell)}}{B^{(j)}}_1}\nonumber\\
&\phantom{=} \hspace{24mm}+ \sum_{{ \ell,j\geq 1}\atop {\ell+j-1=k}}\Sym_k\paren*{\sum_{r=2}^{\min\{\ell,j\}}\frac{MC_{\ell,M}C_{j,M}}{C_{k,M}\prod_{a=1}^{r-1}(M-k+a)} \comm{A^{(\ell)}}{B^{(j)}}_r} \nonumber\\
&\eqqcolon \mathrm{Term}_{1,M}+\mathrm{Term}_{2,M}.
\end{align}

We first consider $\mathrm{Term}_{1,M}$. Since
\begin{equation*}
\lim_{M\rightarrow\infty}\frac{MC_{\ell,M}C_{j,M}}{C_{k,M}}=  \lim_{M\rightarrow\infty}\frac{M \prod_{a=1}^{k} (M+1-a)}{(\prod_{a=1}^\ell (M+1-a)) (\prod_{a=1}^j (M+1-a))} =\lim_{M\rightarrow\infty} \frac{M^{k+1}}{M^{\ell+j}} = 1,
\end{equation*}
we see that
\begin{equation}
\mathrm{Term}_{1,M} \rightarrow \sum_{\ell,j\geq 1;\ell+j-1=k}\Sym_k\paren*{\comm{A^{(\ell)}}{B^{(j)}}_1},
\end{equation}
as $M\rightarrow \infty$, in $\g_k$.

We next consider $\mathrm{Term}_{2,M}$. Let $2\leq r\leq \min\{\ell,j\}$. Since
\begin{align}
\lim_{M\rightarrow\infty} \frac{MC_{\ell,N}C_{j,M}}{C_{k,M}\prod_{a=1}^{r-1} (M-k+a)} &=  \lim_{M\rightarrow\infty}\frac{M \prod_{a=1}^{k} (M+1-a)}{(\prod_{a=1}^\ell (M+1-a)) (\prod_{a=1}^j (M+1-a)) (\prod_{a=1}^{r-1} (M-k+a))} \nonumber\\
&= \lim_{M\rightarrow\infty} \frac{M^{k+1}}{M^{\ell+j+r-1}}  \nonumber\\
&=\lim_{M\rightarrow\infty} M^{1-r} \nonumber\\
&=0,
\end{align}
we see that
\begin{equation}
\Sym_k\paren*{\frac{MC_{\ell,M}C_{j,M}}{C_{k,M}\prod_{a=1}^{r-1}(M-k+a)} \comm{A^{(\ell)}}{B^{(j)}}_r} \rightarrow 0,
\end{equation}
as $M\rightarrow\infty$, in $\g_k$. Summing over the ranges $2\leq r\leq \min\{\ell,j\}$ and $\ell+j-1=k$, for a total of finitely many terms, we conclude that
\begin{equation}
\mathrm{Term}_{2,M} \rightarrow 0,
\end{equation}
as $M\rightarrow \infty$, in $\g_k$, proving the result.
\end{proof}

\subsection{The Lie algebra $\G_{\infty}$ of observable $\infty$-hierarchies}\label{ssec:geo_LA}
As mentioned in the introduction, the simplified form of $\comm{\cdot}{\cdot}_{\G_\infty}$ allows us to take advantage of the good mapping property and extend this bracket to a map on a much larger real topological vector space, which we redefine $\G_\infty$ to be, to obtain a Lie algebra of observable $\infty$-hierarchies. We rigorously construct this extension now. 

We define $\g_{k,gmp}$ to be
\begin{align}\label{gkgmp_def}
\g_{k,gmp} \coloneqq \{A^{(k)} \in \L_{gmp}(\Sc_{s}(\R^{k}), \Sc_{s}'(\R^{k})) : A^{(k)} = - (A^{(k)})^{*}\}.
\end{align}
In words, $\g_{k,gmp}$ is the real, locally convex space consisting of skew-adjoint elements of $\L_{gmp}(\Sc_{s}(\R^{k}),\Sc_{s}'(\R^{k}))$. We will hereafter refer to the elements of $\g_{k, gmp}$ as \emph{$k$-particle or $k$-body observables}.  We define the locally convex direct sum 
\begin{align}\label{Ginf_def}
\G_{\infty} \coloneqq \bigoplus_{k=1}^{\infty}\g_{k,gmp} .
\end{align}
We refer to the elements of $\G_{\infty}$ as \emph{observable $\infty$-hierarchies}. For  
\[
A = (A^{(k)})_{k\in\N},\,\, B = (B^{(k)})_{k\in\N} \in \G_\infty,
\]
 we define
\begin{equation}\label{eq:C_def}
\begin{split}
&\comm{A}{B}_{\G_{\infty}} \coloneqq C=(C^{(k)})_{k\in\N},\\
&C^{(k)} \coloneqq \Sym_k\biggl(\sum_{{ \ell,j\geq 1}\atop {\ell+j-1=k}} \comm{A^{(\ell)}}{B^{(j)}}_1\biggr),
\end{split}
\end{equation}
where $\Sym_k$ denotes the bosonic symmetrization operator defined in \cref{sec:pre}, which we recall is given by
\begin{equation}\label{eq:sym_repeat}
\Sym_k(A^{(k)}) \coloneqq \frac{1}{k!} \sum_{\pi\in\Ss_{k}} A^{(k)}_{(\pi(1), \ldots, \pi(k))}, \quad A^{(k)}_{(\pi(1), \ldots, \pi(k))} = \pi \circ A^{(k)}_{1, \ldots, k} \circ \pi^{-1}
\end{equation}
and where $\comm{A^{(\ell)}}{B^{(j)}}_1$ is given according to \eqref{eq:comm_r} by
\begin{equation}
\begin{split}
\comm{A^{(\ell)}}{B^{(j)}}_1 &= j A^{(\ell)}\circ_1 B^{(j)} - \ell B^{(j)}\circ_1 A^{(\ell)} \\
&= j A_{(1,\ldots,\ell)}^{(\ell)}\biggl(\sum_{\alpha=1}^\ell B_{(\alpha,\ell+1,\ldots,\ell+j-1)}^{(j)}\biggr) - \ell B^{(j)}\biggl(\sum_{\alpha=1}^j A_{(\alpha,j+1,\ldots,j+\ell-1)}^{(\ell)}\biggr).
\end{split}
\end{equation}

The main goal of this section is to establish the existence of a Lie algebra of observable $\infty$-hierarchies, namely, to prove \cref{prop:G_inf_br}:

\LA*

The construction follows closely our $N$-body approach in \cref{sec:geom_N}; however, there are new technical difficulties that have to be considered. Indeed, $\G_\infty$ contains more singular objects than $\G_N$, and we have to heavily exploit the good mapping property in order to handle this issue. We remind the reader the enlarged definition of $\G_\infty$, as opposed to simply the union of the $\G_N$, is necessary to accommodate the observable $\infty$-hierarchy $-i\W_{GP}$ which generates the GP Hamiltonian functional.

We first need to establish that the Lie bracket given by \eqref{eq:C_def} is well-defined on $\G_\infty$. To this end, we must begin by giving meaning to the composition
\begin{equation}\label{eq:comp}
A_{(1,\ldots,\ell)}^{(\ell)}\paren*{\sum_{\alpha=1}^\ell B_{(\alpha,\ell+1,\ldots,\ell+j-1)}^{(j)}}
\end{equation}
as an operator in $\L(\Sc(\R^{k}),\Sc'(\R^{k}))$, for which it will be convenient to proceed term-wise by extending $A^{(\ell)}$ and $B^{(j)}$ to operators defined on the entire space $\Sc(\R^\ell)$ and $\Sc(\R^j)$, respectively, as described in \cref{rem:tbyt}.\footnote{We will see later that the choice of extension is immaterial.} For general $A^{(\ell)}\in\L(\Sc(\R^{\ell}),\Sc'(\R^{\ell}))$ and $B^{(j)}\in \L(\Sc(\R^{j}),\Sc'(\R^{j}))$, such a composition may not be well-defined, see \cref{not_well_def}, and hence we appeal to the good mapping property of \cref{def:gmp} to give meaning to \eqref{eq:comp}. It will be useful in the sequel to observe that the definition of the good mapping property says the following: let  $A^{(\ell)} \in \L(\Sc(\R^{\ell}),\Sc'(\R^{\ell}))$ and $(f^{(\ell)},g^{(\ell)})\in\Sc(\R^{\ell})\times\Sc(\R^{\ell})$, and for fixed $x_\alpha'\in \R$, consider the distribution in $\Sc'(\R)$ defined by
\begin{equation}
\phi \mapsto \ipp*{A^{(\ell)}f^{(\ell)},\paren*{\phi\otimes_\alpha g^{(\ell)}(\cdot,x_{\alpha}',\cdot)}}_{\Sc'(\R^\ell)-\Sc(\R^\ell)},
\end{equation}
where
\begin{equation}
\label{eq:def_ot_a}
\paren*{\phi\otimes_\alpha g^{(\ell)}(\cdot,x_\alpha',\cdot)}(\ul{y}_\ell) \coloneqq \phi(y_\alpha)g^{(\ell)}(\ul{y}_{1;\alpha-1}, x_\alpha',\ul{y}_{\alpha+1;\ell}),\qquad \ul{y}_\ell \in\R^{\ell}.
\end{equation}
Then $A^{(\ell)}\in \L_{gmp}(\Sc(\R^{\ell}),\Sc'(\R^{\ell}))$ if the element of $\Sc(\R;\Sc'(\R))$\footnote{ Given a Hausdorff locally convex space $E$, we let $\Sc(\R^d;E)$ denote the space of functions $f\in C^\infty(\R^d;E)$ such that for each pair of $d$-dimensional polynomials $P$ and $Q$ with complex coefficients, the union $\bigcup_{x\in\R^d} \{P(x)Q(\p_x)f(x)\}$ is contained in a bounded subset of $E$. We endow $\Sc(\R^d;E)$ with the topology of uniform convergence of the functions $P(x)Q(\p_x)f(x)$, for all $P$ and $Q$.} defined by
\begin{equation}
x_{\alpha}' \mapsto \ipp*{A^{(\ell)}f^{(\ell)}, (\cdot)\otimes_\alpha g^{(\ell)}(\cdot,x_\alpha',\cdot)}_{\Sc'(\R^\ell)-\Sc(\R^\ell)},
\end{equation}
may be identified with a (necessarily unique) Schwartz function $\Phi(f^{(\ell)},g^{(\ell)})$ in $\Sc(\R^{2})$ by
\begin{equation}
\ipp*{A^{(\ell)}f^{(\ell)}, \phi\otimes_\alpha g^{(\ell)}(\cdot,x_\alpha',\cdot)}_{\Sc'(\R^\ell)-\Sc(\R^\ell)} = \int_{\R}dx_{\alpha}\Phi(f,g)(x_{\alpha},x_{\alpha}') \phi(x_{\alpha}), \qquad x_\alpha'\in\R,
\end{equation}
and the assignment $\Phi: \Sc(\R^{\ell}) \times \Sc(\R^{\ell}) \rightarrow \Sc(\R^{2})$ is continuous.

\begin{lemma}[$\circ_\alpha^\beta$ contraction]\label{lem:gmp}
Let $i,j\in\N$, let $k\coloneqq i+j-1$, and let $(\alpha,\beta)\in\N_{\leq i}\times\N_{\leq j}$. Then there exists a bilinear map, continuous in the first entry,
\begin{equation}
\circ_\alpha^\beta : \L(\Sc(\R^{i}),\Sc'(\R^{i})) \times \L_{gmp}(\Sc(\R^{j}),\Sc'(\R^{j})) \rightarrow \L(\Sc(\R^{k}),\Sc'(\R^{k})),
\end{equation}
such that $A^{(i)}\circ_\alpha^\beta B^{(j)}$ corresponds to
\begin{equation}
A^{(i)}\circ_\alpha^\beta B^{(j)} = A_{(1,\ldots,i)}^{(i)}B_{(i+1,\ldots,i+\beta-1,\alpha,i+\beta,\ldots,k)}^{(j)},
\end{equation}
when $A^{(i)}\in \L(\Sc(\R^i),\Sc(\R^i))$ and $B^{(j)}\in\L(\Sc(\R^j),\Sc(\R^j))$ or $A^{(i)}\in \L(\Sc(\R^i),\Sc'(\R^i))$ and $B^{(j)}\in \L(\Sc'(\R^j),\Sc'(\R^j))$. If we replace the domain space $\L(\Sc(\R^{i}),\Sc'(\R^{i}))$ for the first entry by $\L_{gmp}(\Sc(\R^{i}),\Sc'(\R^{i}))$, then the bilinear map
\begin{equation}
\circ_\alpha^\beta: \L_{gmp}(\Sc(\R^{i}),\Sc'(\R^{i}))\times \L_{gmp}(\Sc(\R^{j}),\Sc'(\R^{j})) \rightarrow \L_{gmp}(\Sc(\R^{k}),\Sc'(\R^{k}))
\end{equation}
is continuous in the first entry.
\end{lemma}

\begin{remark}\label{rem:comm_def}
Using this lemma and bosonic symmetry, we note that we can rewrite our definition of $\comm{\cdot}{\cdot}_1$ from \eqref{eq:comm_r} using the contractions $\circ_\alpha^\beta$ as follows: Let $i,j\in\N$ and set $k\coloneqq i+j-1$. We extend $\comm{\cdot}{\cdot}_1$ to be the bilinear, continuous in the first entry, map
\begin{equation}\label{def:ab_brak}
\begin{split}
&\comm{\cdot}{\cdot}_1: \L_{gmp}(\Sc(\R^{i}),\Sc'(\R^{i})) \times \L_{gmp}(\Sc(\R^{j}),\Sc'(\R^{j})) \rightarrow \L_{gmp}(\Sc(\R^{k}),\Sc'(\R^{k})) \\
&\phantom{=}\qquad (A^{(i)},B^{(j)}) \mapsto \sum_{\alpha=1}^i\sum_{\beta=1}^j A^{(i)} \circ_\alpha^\beta B^{(j)} - B^{(j)}\circ_\beta^\alpha A^{(i)},
\end{split}
\end{equation}
for $\circ_\alpha^\beta$ and $\circ_\beta^\alpha$ as in \cref{lem:gmp}.
\end{remark}

\begin{proof}[Proof of \cref{lem:gmp}]
We first show that for fixed $f\in\Sc(\R^{k})$, there is a well-defined element
\begin{equation}
\label{comp}
(A^{(i)}\circ_\alpha^\beta B^{(j)})(f)\in\Sc'(\R^k)
\end{equation}
corresponding to
\begin{equation}
\label{eq:comp_cor}
A_{(1,\ldots,i)}^{(i)}B_{(i+1,\ldots,i+\beta-1,\alpha,i+\beta,\ldots,k)}^{(j)}(f).
\end{equation}
Let $g\in\Sc(\R^{k})$. Now it follows from the assumption that $B^{(j)}$ has the good mapping property and \cref{rem:gmp} that the bilinear map
\begin{equation}
\begin{split}
(\tl{f},\tl{g}) \mapsto \,\,&\ipp*{B_{(2,\ldots,\beta,1,\beta+1,\ldots,j)}^{(j)}(\tl{f}(\ux_{\alpha-1},\cdot,\ux_{\alpha+1;i},\cdot)),(\cdot)\otimes \tl{g}(\ux_i',\cdot)}_{\Sc'(\R^j)-\Sc(\R^j)},
\end{split}
\end{equation}
which is a priori a bilinear continuous map
\begin{equation}
\Sc(\R^k)\times\Sc(\R^k) \rightarrow \Sc_{(\ux_{\alpha-1},\ux_{\alpha+1;i},\ux_i')}(\R^{\alpha-1}\times\R^{i-\alpha}\times\R^i; \Sc_{x_\alpha}'(\R)),
\end{equation}
is identifiable with a unique smooth map
\begin{equation}
\Phi_{B^{(j)},\alpha,\beta}: \Sc(\R^k)\times\Sc(\R^k) \rightarrow \Sc_{(\ux_i;\ux_i')}(\R^{2i}).
\end{equation}
Since we have the canonical isomorphism
\begin{equation}
\L(\Sc(\R^{i}),\Sc'(\R^{i}))  \cong \Sc'(\R^{2i})
\end{equation}
by the Schwartz kernel theorem, we therefore define the composition \eqref{comp} by
\begin{equation}\label{eq:comp_def}
\begin{split}
&\ipp{(A^{(i)}\circ_\alpha^\beta B^{(j)})f,g}_{\Sc'(\R^{k})-\Sc(\R^{k})} \coloneqq \ipp*{K_{A^{(i)}},\Phi_{B^{(j)},\alpha,\beta}(f,g)^t}_{\Sc'(\R^{2i})-\Sc(\R^{2i})},
\end{split}
\end{equation}
where 
\begin{equation*}
\Phi_{B^{(j)},\alpha,\beta}(f,g)^t(\ux_i;\ux_i') = \Phi_{B^{(j)},\alpha,\beta}(f,g)(\ux_i';\ux_i), \qquad (\ux_i,\ux_i')\in\R^{2i}.
\end{equation*}

Hence, taking \eqref{eq:comp_def} as the definition of \eqref{comp} for $f\in\Sc(\R^{k})$, we have defined an evidently linear map 
\begin{equation}
A^{(i)}\circ_\alpha^\beta B^{(j)} : \Sc(\R^{k}) \rightarrow \Sc'(\R^{k}).
\end{equation}
The continuity of this map follows from its definition as a composition of continuous maps. Bilinearity of $\circ_\alpha^\beta$ in $A^{(i)}$ and $B^{(j)}$ is obvious. Moreover, it is clear that if $B^{(j)}$ has the good mapping property, then $A^{(i)}\circ_\alpha^\beta B^{(j)}$ has the good mapping property. Lastly, the reader can check from the distributional Fubini-Tonelli theorem that our definition of $A^{(i)}\circ_\alpha^\beta B^{(j)}$ coincides with the composition \eqref{eq:comp_cor} in the cse where $A^{(i)}\in \L(\Sc(\R^i),\Sc(\R^i))$ and $B^{(j)}\in\L(\Sc(\R^j),\Sc(\R^j))$ or $A^{(i)}\in \L(\Sc(\R^i),\Sc'(\R^i))$ and $B^{(j)}\in \L(\Sc'(\R^j),\Sc'(\R^j))$.

We now prove that the map
\begin{equation}
(\cdot)\circ_\alpha^\beta(\cdot): \L(\Sc(\R^{i}),\Sc'(\R^{i}))\times \L_{gmp}(\Sc(\R^{j}),\Sc'(\R^{j})) \rightarrow \L_{gmp}(\Sc(\R^{k}),\Sc'(\R^{k}))
\end{equation}
is continuous in the first entry, that is, for fixed $B^{(j)}\in \L_{gmp}(\Sc(\R^{i}),\Sc'(\R^{i}))$, the map
\begin{equation}
\L(\Sc(\R^{i}),\Sc'(\R^{i})) \rightarrow \L_{gmp}(\Sc(\R^{k}),\Sc'(\R^{k})), \qquad A^{(i)} \mapsto A^{(i)}\circ_\alpha^\beta B^{(j)}
\end{equation}
is continuous. By considerations of symmetry, it suffices to consider the case $(\alpha,\beta)= (1,1)$. To this end, it suffices to show that given a bounded subset $\mathfrak{R}^{(k)} \subset \Sc(\R^{k})$, there exists a bounded subset $\mathfrak{R}^{(i)} \subset \Sc(\R^{i})$ such that
\begin{equation}
\sup_{f^{(k)},g^{(k)}\in\mathfrak{R}^{(k)}} \left|\ip{(A^{(i)}\circ_1^1 B^{(j)}) f^{(k)}}{g^{(k)}}\right| \lesssim \sup_{f^{(i)},g^{(i)}\in\mathfrak{R}^{(i)}} \left|\ip{A^{(i)}f^{(i)}}{g^{(i)}}\right|.
\end{equation}
To see how to obtain the desired seminorm, first observe that
\begin{align}
\left|\ip{(A^{(i)}\circ_1^1 B^{(j)})f^{(k)}}{g^{(k)}}\right| &= \left|\ipp*{K_{A^{(i)}},\Phi_{B^{(j)},1,1}(f^{(k)},g^{(k)})^t}_{\Sc'(\R^{2i})-\Sc(\R^{2i})}\right| \nonumber\\
&=\left|\Tr_{1,\ldots,i}\paren*{A^{(i)}\Phi_{B^{(j)},1,1}(f^{(k)},g^{(k)})}\right|,
\end{align}
where the ultimate equality follows from the definition of the generalized trace (recall \cref{def:gen_trace}) and we commit an abuse of notation by using $\Phi_{B^{(j)},1,1}(f^{(k)},g^{(k)})$ to denote the operator in $\L(\Sc'(\R^{i}),\Sc(\R^{i}))$ defined by this integral kernel. Since $\mathfrak{R}^{(k)}$ is bounded, the image $\Phi_{B^{(j)},1,1}(\mathfrak{R}^{(k)}\times\mathfrak{R}^{(k)})$ is a bounded subset of $\Sc(\R^{2i}) \cong \L(\Sc'(\R^{i}),\Sc(\R^{i}))$, and since $A^{(i)}$ is continuous, it follows that
\begin{equation}
\sup_{\gamma^{(i)}\in\Phi_{B^{(j)},1,1}(\mathfrak{R}^{(k)}\times\mathfrak{R}^{(k)})} \left|\Tr_{1,\ldots,i}\paren*{A^{(i)}\gamma^{(i)}}\right| <\infty.
\end{equation}
Hence, there exists an element $\gamma_{0}^{(i)}\in\Phi_{B^{(j)},1,1}(\mathfrak{R}^{(k)}\times\mathfrak{R}^{(k)})$ such that
\begin{equation}
\left|\Tr_{1,\ldots,i}\paren*{A^{(i)}\gamma_{0}^{(i)}}\right| \geq \frac{1}{2}\sup_{\gamma^{(i)}\in\Phi_{B^{(j)},1,1}(\mathfrak{R}^{(k)}\times\mathfrak{R}^{(k)})} \left|\Tr_{1,\ldots,i}\paren*{A^{(i)}\gamma^{(i)}}\right|.
\end{equation}
Since each element of $\Sc(\R^{2i})$ can be written as $\sum_{\ell=1}^{\infty}\lambda_{\ell} f_{\ell}^{(i)}\otimes g_{\ell}^{(i)}$, where $\sum_{\ell=1}^{\infty}|\lambda_{\ell}|\leq 1$, and $f_{\ell}^{(i)},g_{\ell}^{(i)}$ are sequences in $\Sc(\R^i)$ converging to zero, we see from the separate continuity of the generalized trace that
\begin{align}
\left|\Tr_{1,\ldots,i}\paren*{A^{(i)}\gamma_{0}^{(i)}}\right| &\leq \sum_{\ell=1}^{\infty} |\lambda_{\ell}| \left|\Tr_{1,\ldots,i}\paren*{A^{(i)}(f_{0,\ell}^{(i)}\otimes g_{0,\ell}^{(i)})}\right| \nonumber\\
&\leq \sup_{f^{(i)},g^{(i)}\in\{f_{0,\ell'}^{(i)},g_{0,\ell'}^{(i)}\}_{\ell'=1}^{\infty}}\left|\ipp{A^{(i)}f^{(i)},g^{(i)})}_{\Sc'(\R^i)-\Sc(\R^i)}\right| .
\end{align}
We claim that $\{f_{0,\ell}^{(i)}, g_{0,\ell}^{(i)}\}_{\ell=1}^{\infty}$ is a bounded subset of $\Sc(\R^{i})$, which then completes the proof. Indeed, this follows readily from the fact that $f_{0,\ell}^{(i)},g_{0,\ell}^{(i)}$ converge to zero.
\end{proof}

\begin{remark}\label{rem:sc}
If we restrict the domain of the map $\circ_\alpha^\beta$ to the space
\begin{equation*}
\L_{gmp,*}(\Sc(\R^{i}),\Sc'(\R^{i})) \times \L_{gmp,*}(\Sc(\R^j),\Sc'(\R^j))
\end{equation*}
consisting of distribution-valued operators satisfying the good mapping property such that their adjoints also satisfy the good mapping property, which we endow with the subspace topology, then it follows by duality that $\circ_\alpha^\beta$ is separately continuous on this space.
\end{remark}

\begin{remark}
\label{rem:bos_sym}
If $B^{(j)} \in \L_{gmp}(\Sc_s(\R^j),\Sc_s'(\R^j))$, then it follows from bosonic symmetry that for any $(\alpha,\beta) \in \N_{\leq i}\times\N_{\leq j}$,
\begin{equation}
A^{(i)}\circ_\alpha^\beta B^{(j)} = A^{(i)}\circ_\alpha^1 B^{(j)}.
\end{equation}
\end{remark}

\begin{remark}\label{rem:ext_wd}
If $A^{(i)}\in \L(\Sc_s(\R^i),\Sc'(\R^i))$ and $B^{(j)}\in \L_{gmp}(\Sc_s(\R^j),\Sc_s'(\R^j))$, then given two extensions $A_1^{(i)}, A_2^{(i)}\in \L(\Sc(\R^i),\Sc'(\R^i))$ of $A^{(i)}$, we claim that
\begin{equation}
\sum_{\alpha=1}^i A_1^{(i)}\circ_\alpha^1 B^{(j)} = \sum_{\alpha=1}^i A_2^{(i)}\circ_\alpha^1 B^{(j)} \in \L(\Sc_s(\R^k),\Sc'(\R^k)).
\end{equation}
Indeed, for $f\in \Sc_s(\R^k), g\in \Sc(\R^k)$, we have that
\begin{align}
\sum_{\alpha=1}^i \ipp{g,(A_1^{(i)}\circ_\alpha^1 B^{(j)})f}_{\Sc(\R^k)-\Sc'(\R^k)} &=\sum_{\alpha=1}^i \ipp*{K_{A_1^{(i)}}, \Phi_{B^{(j)},\alpha,1}(f,g)^t}_{\Sc'(\R^{2i})-\Sc(\R^{2i})}.
\end{align}
Since each $\Phi_{B^{(j)},\alpha,1}(f,g)\in\Sc(\R^{2i})$ and $f\in\Sc_s(\R^k)$, we see that
\begin{equation}
\sum_{\alpha=1}^i \Phi_{B^{(j)},\alpha,1}(f,g)(\pi(\ux_i);\ux_i') = \sum_{\alpha=1}^i \Phi_{B^{(j)},\alpha,1}(f,g)(\ux_i;\ux_i'), \qquad (\ux_i,\ux_i') \in \R^{2i},
\end{equation}
for any permutation $\pi\in\Ss_i$. Consequently, for fixed $\ux_i'\in\R^i$, the function $\sum_{\alpha=1}^i \Phi_{B^{(j)},\alpha,1}(f,g)(\cdot,\ux_i')$ belongs to $\Sc_s(\R^i)$ on which the two extensions $A_1^{(i)}$ and $A_2^{(i)}$ agree. It then follows from the Schwartz kernel theorem that 
\begin{equation}
\begin{split}
\ipp*{K_{A_1^{(i)}}, \paren*{\sum_{\alpha=1}^i \Phi_{B^{(j)},\alpha,1}(f,g)}^t}_{\Sc'(\R^{2i})-\Sc(\R^{2i})} &= \ipp*{K_{A_2^{(i)}}, \paren*{\sum_{\alpha=1}^i \Phi_{B^{(j)},\alpha,1}(f,g)}^t}_{\Sc'(\R^{2i})-\Sc(\R^{2i})},
\end{split}
\end{equation}
and therefore
\begin{equation}
\sum_{\alpha=1}^i \ipp{g,(A_1^{(i)}\circ_\alpha^1 B^{(j)})f}_{\Sc(\R^k)-\Sc'(\R^k)} = \sum_{\alpha=1}^i \ipp{g,(A_2^{(i)}\circ_\alpha^1 B^{(j)})f}_{\Sc(\R^k)-\Sc'(\R^k)},
\end{equation}
which establishes our claim.
\end{remark}

By \cref{lem:gmp},
\begin{equation}
A^{(\ell)}\circ_\alpha^\beta B^{(j)} \in \L_{gmp}(\Sc(\R^{k}),\Sc'(\R^{k})), \qquad  \textup{for }\ell+j-1=k.
\end{equation}
Hence, by definition of the bracket $\comm{\cdot}{\cdot}_1$ and \cref{rem:comm_def},
\begin{equation}
\sum_{{ \ell,j\geq 1}\atop {\ell+j-1=k}}\comm{A^{(\ell)}}{B^{(j)}}_1 \in \L_{gmp}(\Sc_s(\R^{k}),\Sc'(\R^{k})).
\end{equation}
Thus it remains to show two properties: first that the symmetrization of an operator preserves the good mapping property, which will then establish that $C^{(k)} \in \L_{gmp}(\Sc_s(\R^{k}),\Sc_s'(\R^{k}))$, where $C^{(k)}$ is defined according to \eqref{eq:C_def}, and second that $C^{(k)}$ is skew-adjoint. We begin with the following lemma which establishes the desired property of the symmetrization operators.

\begin{lemma}\label{sym_gmp}
If $A=(A^{(k)})_{k\in\N}\in\bigoplus_{k=1}^{\infty}\L_{gmp}(\Sc(\R^{k}),\Sc'(\R^{k}))$, then 
\[
\Sym(A)\in \bigoplus_{k=1}^{\infty}\L_{gmp}(\Sc_{s}(\R^{k}),\Sc_{s}'(\R^{k})).
\]
\end{lemma}
\begin{proof}
It suffices to show that for each $k\in\N$, if $A^{(k)}\in\L_{gmp}(\Sc(\R^{k}),\Sc'(\R^{k}))$, then 
\[
\Sym_k(A^{(k)})\in\L_{gmp}(\Sc_{s}(\R^{k}),\Sc_{s}'(\R^{k})).
\]
Let $\alpha\in\N_{\leq k}$. We need to show that the map
\begin{equation}
\begin{split}
&\Sc_{s}(\R^{k})\times\Sc_{s}(\R^{k}) \rightarrow \Sc(\R;\Sc'(\R)) \\
&\phantom{=} (f^{(k)},g^{(k)})\mapsto \ipp*{\Sym_k(A^{(k)})(f^{(k)}), (\cdot)\otimes_\alpha g(\cdot,x_\alpha',\cdot)}_{\Sc'(\R^k)-\Sc(\R^k)}
\end{split}
\end{equation}
may be identified with a continuous map $\Sc_{s}(\R^{k})\times\Sc_{s}(\R^{k})\rightarrow \Sc(\R^{2})$. By definition of the $\Sym_k$ operator and bilinearity of the distributional pairing, we have that
\begin{align}
&\ipp*{\Sym_k(A^{(k)})f^{(k)}, (\cdot)\otimes_{\alpha} g^{(k)}(\cdot,x_\alpha',\cdot)}_{\Sc'(\R^k)-\Sc(\R^k)} \nonumber\\
&= \frac{1}{k!}\sum_{\pi\in\Ss_{k}} \ipp*{A_{(\pi(1),\ldots,\pi(k))}^{(k)}f^{(k)}, (\cdot)\otimes_\alpha g^{(k)}(\cdot,x_\alpha',\cdot)}_{\Sc'(\R^k)-\Sc(\R^k)}. \label{eq:gmp_sym_sum}
\end{align}
By definition of the notation $A_{(\pi(1),\ldots,\pi(k))}^{(k)} = \pi\circ A^{(k)}_{1, \ldots, k}\circ\pi^{-1}$, we have that
\begin{align}
&\ipp*{A_{(\pi(1),\ldots,\pi(k))}^{(k)}f^{(k)}, (\cdot)\otimes_\alpha g^{(k)}(\cdot,x_\alpha',\cdot)}_{\Sc'(\R^k)-\Sc(\R^k)} \nonumber\\
&= \ipp*{A^{(k)}(f^{(k)}\circ\pi^{-1})\circ\pi, (\cdot)\otimes_\alpha g^{(k)}(\cdot,x_\alpha',\cdot)}_{\Sc'(\R^k)-\Sc(\R^k)} \nonumber\\
&= \ipp*{A^{(k)}(f^{(k)})\circ\pi, (\cdot)\otimes_\alpha g^{(k)}(\cdot,x_\alpha',\cdot)}_{\Sc'(\R^k)-\Sc(\R^k)},
\end{align}
where the ultimate equality follows from the assumption $f^{(k)}\in\Sc_{s}(\R^{k})$. Let $\phi\in\Sc(\R)$ be a test function. Then by definition of the permutation of a distribution,
\begin{align}
\ipp*{A^{(k)}(f^{(k)})\circ\pi, \phi\otimes_\alpha g^{(k)}(\cdot,x_\alpha',\cdot)}_{\Sc'(\R^k)-\Sc(\R^k)} = \ipp*{A^{(k)}f^{(k)}, (\phi\otimes_\alpha g^{(k)}(\cdot,x_\alpha',\cdot))\circ\pi^{-1}}_{\Sc'(\R^k)-\Sc(\R^k)}.
\end{align}
Observing that
\begin{equation}
((\phi\otimes_\alpha g^{(k)}(\cdot,x_\alpha',\cdot))\circ\pi^{-1})(\ux_k) =  g^{(k)}(x_{\pi^{-1}(1)},\ldots,x_{\pi^{-1}(\alpha-1)},x_{\alpha}',x_{\pi^{-1}(\alpha+1)},\ldots,x_{\pi^{-1}(k)})\phi(x_{\pi^{-1}(\alpha)}), \quad \ux_k\in\R^k,
\end{equation}
upon setting $j\coloneqq \pi^{-1}(\alpha)$ and using the bosonic symmetry of $g^{(k)}$, we obtain that
\begin{equation}
((\phi\otimes_\alpha g^{(k)}(\cdot,x_\alpha',\cdot))\circ\pi^{-1})(\ux_k) = g^{(k)}(\ux_{j-1},x_\alpha',\ux_{j+1;k})\phi(x_j) = (\phi \otimes_j g^{(k)}(\cdot,x_\alpha',\cdot))(\ux_k).
\end{equation}
Since $A^{(k)}$ has the good mapping property, we have that
\begin{equation}
\begin{split}
\ipp*{A^{(k)}f^{(k)}, \phi\otimes_j g^{(k)}(\cdot,x_\alpha',\cdot)}_{\Sc'(\R^k)-\Sc(\R^k)} = \ipp*{\Phi_{A^{(k)},j}(f^{(k)},g^{(k)})(\cdot,x_\alpha'), \phi}_{\Sc'(\R)-\Sc(\R)},
\end{split}
\end{equation}
where $\Phi_{A^{(k)},j}:\Sc(\R^k)\times\Sc(\R^k)\rightarrow \Sc(\R^2)$ is a continuous bilinear map. Since $\Sc_{s}(\R^{k})$ continuously embeds (trivially) in $\Sc(\R^{k})$ and since $\alpha\in\N_{\leq k}$ was arbitrary, we conclude that \eqref{eq:gmp_sym_sum} is identifiable with a finite sum of continuous bilinear maps $\Sc_{s}(\R^{k})\times\Sc_{s}(\R^{k})\rightarrow\Sc(\R^{2})$, and the proof of the lemma is complete.
\end{proof}

Finally, to conclude our proof that the Lie bracket is well-defined, we only need to verify that $C^{(k)}$ defined according to \eqref{eq:C_def} is skew-adjoint. This is a consequence of \cref{rem:comm_def}, \cref{rem:ext_wd}, and the following lemma.

\begin{lemma}
Let $i,j\in\N$, and define $k\coloneqq i+j-1$. Let $A^{(i)} \in \L_{gmp}(\Sc(\R^i),\Sc'(\R^i))$ and $B^{(j)}\in\L_{gmp}(\Sc(\R^j),\Sc'(\R^j))$ be skew-adjoint distribution-valued operators. Then for any $(\alpha,\beta)\in\N_{\leq i}\times\N_{\leq j}$,
\begin{equation}
(A^{(i)}\circ_\alpha^\beta B^{(j)})^* = (B^{(j)}\circ_\beta^\alpha A^{(i)})_{(i+1,\ldots,i+\beta-1,\alpha,i+\beta,\ldots,k,1,\ldots,i)} \in \L_{gmp}(\Sc(\R^k),\Sc'(\R^k)).
\end{equation}
\end{lemma}
\begin{proof}
By considerations of symmetry, it suffices to consider the case where $(\alpha,\beta)=(1,1)$. Recalling the definition of the adjoint of a distribution-valued operator, see \cref{lem:dvo_adj}, we need to show that
\begin{equation}
\begin{split}
&\ipp*{(B^{(j)}\circ_1^1 A^{(i)})_{(1,i+1,\ldots,k,2,\ldots,i)}g, \bar{f}}_{\Sc'(\R^k)-\Sc(\R^k)}\\
&=\ol{\ipp*{(A^{(i)}\circ_1^1 B^{(j)})f, \ol{g}}_{\Sc'(\R^k)-\Sc(\R^k)}},
\end{split}
\end{equation}
for any $f,g\in\Sc(\R^k)$. By \cref{lem:ext_sa},
\begin{equation*}
A^{(i)}_{(1,\ldots,i)} \enspace \text{and} \enspace B^{(j)}_{(1,i+1,\ldots,k)}
\end{equation*}
are both skew-adjoint elements of $\L_{gmp}(\Sc(\R^{k}),\Sc'(\R^{k}))$. Now by density of linear combinations of pure tensors, linearity, and the continuity of the operators $A^{(i)}_{(1,\ldots,i)}$, $B^{(j)}_{(1,i+1,\ldots,k)}$, and $A^{(i)}\circ_1^1 B^{(j)}$, it suffices to consider the expression
\begin{equation}
\ol{\ipp*{(A^{(i)}\circ_1^1 B^{(j)})f, \ol{g}}_{\Sc'(\R^k)-\Sc(\R^k)}}
\end{equation}
in the case where $f,g\in\Sc(\R^{k})$ are pure tensors of the form
\begin{equation}
f=\bigotimes_{a=1}^k f_a \enspace \text{and} \enspace g=\bigotimes_{a=1}^k g_a,
\end{equation}
respectively, where $f_1,\ldots,f_k,g_1,\ldots,g_k\in\Sc(\R)$. Recalling the definition \eqref{eq:comp_def} for $A^{(i)}\circ_1^1 B^{(j)}$, we have that
\begin{align*}
\ol{\ipp*{(A^{(i)}\circ_1^1 B^{(j)})f, \bar{g}}_{\Sc'(\R^k)-\Sc(\R^k)}} &=\ol{ \ipp*{K_{A^{(i)}}, \Phi_{B^{(j)},1,1}(f,\bar{g})^t}_{\Sc'(\R^{2i})-\Sc(\R^{2i})}}.
\end{align*}
An examination of the $\Phi_{B^{(j)}}(f,\bar{g})$ together with the tensor product structure of $f$ and $g$ reveals that
\begin{equation}
\label{eq:bphi_Bj}
\begin{split}
\Phi_{B^{(j)},1,1}(f,\bar{g})(\ux_i;\ux_i') &= \underbrace{(\bigotimes_{a=2}^i f_a)}_{\eqqcolon f^{(i-1)}}(\ux_{2;i}) \underbrace{(\bigotimes_{a=1}^i \ol{g_a})}_{\eqqcolon \ol{g^{(1)}}\otimes \ol{g^{(i-1)}}}(\ux_i')\\
&\phantom{=}\quad \times\ipp*{B^{(j)}\paren*{f_1\otimes \bigotimes_{a=i+1}^{k} f_{a}}, (\cdot) \otimes \bigotimes_{a=i+1}^k \ol{g_a}}_{\Sc'(\R^j)-\Sc(\R^j)}(x_1).
\end{split}
\end{equation}
Since $B^{(j)}$ has the good mapping property, it follows that the element of $\Sc_{x_1}'(\R)$ defined by the second factor in the right-hand side of \eqref{eq:bphi_Bj} is in fact an element of $\Sc(\R)$, which we denote by
\begin{equation}
\label{eq:lphi_Bj}
\phi_{B^{(j)},1}\paren*{f_1\otimes\bigotimes_{a=i+1}^k f_a, \bigotimes_{a=i+1}^k \ol{g_a}} \eqqcolon \phi_{B^{(j)},1}(f^{(j)}, \ol{g^{(j-1)}}).
\end{equation}
Thus, using \eqref{eq:lphi_Bj} and \eqref{eq:bphi_Bj}, we can write
\begin{equation}
\Phi_{B^{(j)},1,1}(f,\bar{g})(\ux_i;\ux_i') = \phi_{B^{(j)},1}(f^{(j)},\ol{g^{(j-1)}})(x_1) f^{(i-1)}(\ux_{2;i}) \ol{g^{(1)}}(x_1') \ol{g^{(i-1)}}(\ux_{2;i}'), \qquad (\ux_i,\ux_i')\in\R^{2i},
\end{equation}
and
\begin{align}
&\ol{\ipp*{K_{A^{(i)}}, \Phi_{B^{(j)},1,1}(f,\bar{g})^t}_{\Sc'(\R^{2i})-\Sc(\R^{2i})}} \nonumber\\
&=\ol{\ipp*{A^{(i)}\paren*{\phi_{B^{(j)},1}(f^{(j)},\ol{g^{(j-1)}}) \otimes f^{(i-1)}}, \ol{g^{(1)}\otimes g^{(i-1)}}}_{\Sc'(\R^i)-\Sc(\R^i)}}
\end{align}
by the Schwartz kernel theorem. Since $A^{(i)}$ is skew-adjoint, we have that this last expression equals
\begin{align}
&-\ipp*{A^{(i)}\paren*{g^{(1)}\otimes g^{(i-1)}}, \ol{\phi_{B^{(j)},1}(f^{(j)}, \ol{g^{(j-1)}})\otimes f^{(i-1)}}}_{\Sc'(\R^i)-\Sc(\R^i)}. \label{eq:Ai_eval}
\end{align}
Now since $A^{(i)}$ also has the good mapping property by assumption, the element of $\Sc_{x_1}'(\R)$ defined by
\begin{equation}
-\ipp*{A^{(i)}\paren*{g^{(1)}\otimes g^{(i-1)}}, (\cdot) \otimes \ol{f^{(i-1)}}}_{\Sc'(\R^i)-\Sc(\R^i)}
\end{equation}
is identifiable with a unique element of $\Sc_{x_1}(\R)$, which we denote by
\begin{equation}
\label{eq:lphi_Ai}
-\phi_{A^{(i)},1}(g^{(1)}\otimes g^{(i-1)}, \ol{f^{(i-1)}}).
\end{equation}
Using \eqref{eq:lphi_Ai}, we see that
\begin{equation}
\label{eq:AB_phi_pair}
\eqref{eq:Ai_eval} = -\int_{\R}dx \phi_{A^{(i)},1}(g^{(1)}\otimes g^{(i-1)}, \ol{f^{(i-1)}})(x) \ol{\phi_{B^{(j)},1}(f^{(j)},\ol{g^{(j-1)}})(x)}.
\end{equation}
After unpacking the definition of the Schwartz function $\phi_{B^{(j)},1}(f^{(j)},\ol{g^{(j-1)}})$ given in \eqref{eq:bphi_Bj} and \eqref{eq:lphi_Bj}, it follows that
\begin{align}
\eqref{eq:AB_phi_pair} &= \ol{\ipp*{B^{(j)}f^{(j)} , \ol{\phi_{A^{(i)},1}(g^{(1)}\otimes g^{(i-1)},\ol{f^{(i-1)}}) \otimes g^{(j-1)}}}_{\Sc'(\R^j)-\Sc(\R^j)}} \nonumber\\
&=\ipp*{B^{(j)}\paren*{\phi_{A^{(i)},1}(g^{(1)}\otimes g^{(i-1)}, \ol{f^{(i-1)}}) \otimes g^{(j-1)}} , \ol{f^{(j)}} }_{\Sc'(\R^j)-\Sc(\R^j)} \nonumber\\
&=\ipp*{K_{B^{(j)}}, \paren*{\paren*{\phi_{A^{(i)},1}(g^{(1)}\otimes g^{(i-1)},\ol{f^{(i-1)}})\otimes g^{(j-1)}} \otimes \ol{f^{(j)}} }^t}_{\Sc'(\R^{2j})-\Sc(\R^{2j})}, \label{eq:Bj_goal_LHS}
\end{align}
where we use the skew-adjointness of $B^{(j)}$ to obtain the penultimate equality and the Schwartz kernel theorem to obtain the ultimate equality.

Our goal now is to show that
\begin{equation}
\label{eq:Phi_perm}
\begin{split}
&\paren*{\phi_{A^{(i)},1}(g^{(1)}\otimes g^{(i-1)}, \ol{f^{(i-1)}}) \otimes g^{(j-1)}}\otimes \ol{f^{(j)}} (\ux_j;\ux_j') \\
&=\Phi_{A^{(i)},1,1}(g\circ\pi, \bar{f}\circ\pi)(\ux_j;\ux_j')
\end{split}
\end{equation}
where $\pi\in\Ss_k$ is the permutation
\begin{equation}
\pi(a) =
\begin{cases}
1, & {a=1} \\
a+j-1, &{2\leq a\leq i} \\
a-i+1, & {i+1\leq a\leq k}.
\end{cases}
\end{equation}
With \eqref{eq:Phi_perm}, we then have by definition of the composite distribution $B^{(j)}\circ_1^1 A^{(i)}$, see \eqref{eq:comp_def}, and the notation
\begin{equation*}
(B^{(j)}\circ_1^1 A^{(i)})_{(1,i+1,\ldots,k,2,\ldots,i)},
\end{equation*}
see \cref{prop:ext_k}, that
\begin{align}
\eqref{eq:Bj_goal_LHS} &= \ipp*{K_{B^{(j)}}, \Phi_{A^{(i)},1,1}(g\circ\pi, \bar{f}\circ\pi)^t}_{\Sc'(\R^{2j})-\Sc(\R^{2j})} \nonumber\\
&=\ipp*{(B^{(j)}\circ_1^1 A^{(i)})(g\circ\pi), \bar{f}\circ\pi}_{\Sc'(\R^k)-\Sc(\R^k)} \nonumber\\
&=\ipp*{ (B^{(j)}\circ_1^1 A^{(i)})_{(1,i+1,\ldots,k,2,\ldots,i)}g,\bar{f}}_{\Sc'(\R^k)-\Sc(\R^k)},
\end{align}
which is exactly what we needed to show.

Turning to \eqref{eq:Phi_perm}, observe that
\begin{equation}
(g\circ\pi)(\ux_k) = g(x_1,x_{j+1},\ldots,x_k,x_2,\ldots,x_j) = g_1(x_1) (\bigotimes_{a=2}^{i} g_a)(\ux_{j+1;k}) (\bigotimes_{a=i+1}^k g_a)(\ux_{2;j}),
\end{equation}
and similarly for $(\bar{f}\circ\pi)$. By the same analysis as in \eqref{eq:bphi_Bj}, it then follows that
\begin{align}
\Phi_{A^{(i)},1,1}(g\circ\pi,\bar{f}\circ\pi)(\ux_j;\ux_j') &= (\bigotimes_{a=i+1}^k g_a)(\ux_{2;j}) (\bigotimes_{a=i+1}^k \ol{f_a})(\ux_{2;j}')\ol{f_1}(x_1') \nonumber\\
&\phantom{=}\quad \times \ipp*{A^{(i)}(\bigotimes_{a=1}^i g_a), (\cdot)\otimes \bigotimes_{a=2}^{i-1} \ol{f_a}}_{\Sc'(\R^i)-\Sc(\R^i)}(x_1) \nonumber\\
&=\phi_{A^{(i)},1}(g^{(1)}\otimes g^{(i-1)}, \ol{f^{(i-1)}})(x_1) g^{(j-1)}(\ux_{2;j}) f^{(j)}(x_j'),
\end{align}
as desired.
\end{proof}

We now turn to the proof of \cref{prop:G_inf_br}.

\begin{proof}[Proof of \protect{\cref{prop:G_inf_br}}]
We first verify the Lie bracket properties \ref{item:LA_1}-\ref{item:LA_3} in \cref{def:la}. Bilinearity and anti-symmetry are immediate from the linearity of the bosonic symmetrization $\Sym$ operator, see \cref{eq:sym_repeat} above, and the bilinearity and anti-symmetry of the bracket $\comm{\cdot}{\cdot}_{1}$. 

To verify the Jacobi identity
\begin{equation}
\comm{A}{\comm{B}{C}}^{(k)} + \comm{C}{\comm{A}{B}}^{(k)} + \comm{B}{\comm{C}{A}}^{(k)} = 0,
\end{equation}
we use our convergence result \cref{prop:LB_lim} together with the fact that $\comm{\cdot}{\cdot}_{\G_N}$ is a Lie bracket by \cref{prop:NH_LA}. Let $A, B, C \in \G_\infty$, where $A=(A^{(k)})_{k\in\N}, B=(B^{(k)})_{k\in\N}, C=(C^{(k)})_{k\in\N}$. Note that since $\G_\infty$ is a direct sum, there exists an $N_0\in\N$ such that $A^{(k)}=B^{(k)}=C^{(k)}=0$ for $k\geq N_0$. Now by mollifying and truncating the Schwartz kernels of the $k$-particle components $A^{(k)}, B^{(k)}, C^{(k)}$, we obtain approximating sequences
\begin{equation}
A_{n_1}\coloneqq (A_{n_1}^{(k)})_{k\in\N}, \,\,B_{n_2}\coloneqq (B_{n_2}^{(k)})_{k\in\N}, \,\,C_{n_3}\coloneqq (C_{n_3}^{(k)})_{k\in\N} \in \G_\infty \cap \bigoplus_{k=1}^\infty \L(\Sc_s'(\R^k),\Sc_s(\R^k))
\end{equation}
such that for all $(n_1,n_2,n_3)\in\N^3$, $A_{n_1}^{(k)}=B_{n_2}^{(k)}=C_{n_3}^{(k)}=0\in\g_{k, gmp}$ for $k\geq N_0$. In particular, $A_{n_1}, B_{n_2}, C_{n_3}\in \G_M$ for any integer $M\geq N_0$. Now for such $M$, we know from the Jacobi identity for $\comm{\cdot}{\cdot}_{\G_M}$ that
\begin{equation}
\comm{A_{n_1}}{\comm{B_{n_2}}{C_{n_3}}_{\G_M}}_{\G_M}  + \comm{C_{n_3}}{\comm{A_{n_1}}{B_{n_2}}_{\G_M}}_{\G_M} + \comm{B_{n_2}}{\comm{C_{n_3}}{A_{n_1}}_{\G_M}}_{\G_M}=0 \in \G_M \subset \G_\infty.
\end{equation}
Consequently, for fixed $(n_1,n_2,n_3)\in\N^3$, we obtain from \cref{prop:LB_lim} that
\begin{align}
0&= \lim_{M\rightarrow\infty} \left(\comm{A_{n_1}}{\comm{B_{n_2}}{C_{n_3}}_{\G_M}}_{\G_M}  + \comm{C_{n_3}}{\comm{A_{n_1}}{B_{n_2}}_{\G_M}}_{\G_M} + \comm{B_{n_2}}{\comm{C_{n_3}}{A_{n_1}}_{\G_M}}_{\G_M} \right) \nonumber\\
&= \comm{A_{n_1}}{\comm{B_{n_2}}{C_{n_3}}_{\G_\infty}}_{\G_\infty}  + \comm{C_{n_3}}{\comm{A_{n_1}}{B_{n_2}}_{\G_\infty}}_{\G_\infty} + \comm{B_{n_2}}{\comm{C_{n_3}}{A_{n_1}}_{\G_\infty}}_{\G_\infty}.
\end{align}
Next, using three applications of the separate continuity of the bracket $\comm{\cdot}{\cdot}_{\G_\infty}$ established below, we have that
\begin{align}
\comm{A}{\comm{B}{C}_{\G_\infty}}_{\G_\infty} &= \lim_{n_1\rightarrow\infty}\lim_{n_2\rightarrow\infty}\lim_{n_3\rightarrow\infty} \comm{A_{n_1}}{\comm{B_{n_2}}{C_{n_3}}_{\G_\infty}}_{\G_\infty}, \\
\comm{C}{\comm{A}{B}_{\G_\infty}}_{\G_\infty} &= \lim_{n_1\rightarrow\infty}\lim_{n_2\rightarrow\infty}\lim_{n_3\rightarrow\infty}\comm{C_{n_3}}{\comm{A_{n_1}}{B_{n_2}}_{\G_\infty}}_{\G_\infty}, \\
\comm{B}{\comm{C}{A}_{\G_\infty}}_{\G_\infty} &= \lim_{n_1\rightarrow\infty}\lim_{n_2\rightarrow\infty}\lim_{n_3\rightarrow\infty}\comm{B_{n_2}}{\comm{C_{n_3}}{A_{n_1}}_{\G_\infty}}_{\G_\infty}.
\end{align}
Summarizing our computations, we have shown that
\begin{align}
0 &= \lim_{n_1\rightarrow\infty}\lim_{n_2\rightarrow\infty}\lim_{n_3\rightarrow\infty}\lim_{M\rightarrow\infty} \biggl( \comm{A_{n_1}}{\comm{B_{n_2}}{C_{n_3}}_{\G_M}}_{\G_M} + \comm{C_{n_3}}{\comm{A_{n_1}}{B_{n_2}}_{\G_M}}_{\G_M} \nonumber\\
&\hspace{100mm} + \comm{B_{n_2}}{\comm{C_{n_3}}{A_{n_1}}_{\G_M}}_{\G_M} \biggr)\nonumber\\
&= \comm{A}{\comm{B}{C}_{\G_\infty}}_{\G_\infty} + \comm{C}{\comm{A}{B}_{\G_\infty}}_{\G_\infty} + \comm{B}{\comm{C}{A}_{\G_\infty}}_{\G_\infty},
\end{align}
which completes the proof of the Jacobi identity.

Finally, we check that the map $\comm{\cdot}{\cdot}_{\G_{\infty}}$ is separately continuous. By linearity, it suffices to show that for each fixed $\ell,j \in\N $ and fixed $\alpha\in\N_{\leq\ell}$, the binary operation $\circ_\alpha^1$ is separately continuous as a map
\begin{equation}
\circ_\alpha^1: \g_{\ell,gmp}\times\g_{j,gmp} \rightarrow \L_{gmp,*}(\Sc(\R^{k}),\Sc'(\R^{k}))
\end{equation}
where $k\coloneqq \ell+j-1$ and where the space $\L_{gmp,*}(\Sc(\R^k),\Sc'(\R^k))$ consists of distribution-valued operators satisfying the good mapping property such that their adjoints also satisfy the good mapping property, endowed with the subspace topology. This property follows from \cref{rem:sc} together with the fact that the adjoints of elements in $\g_{\ell,gmp}$ and $\g_{j,gmp}$ also satisfy the good mapping property by skew-adjointness. Thus, the proof of the proposition is complete.
\end{proof}

\subsection{Lie-Poisson manifold $\G_{\infty}^{*}$ of density matrix $\infty$-hierarchies}\label{ssec:geo_LPM}
In this subsection, we define the Poisson structure on $\G_\infty^*$, which will be used in the sequel in order to establish Hamiltonian properties of the GP hierarchy. Since many of the proofs from \cref{ssec:LP_N} carry over with trivial modification, as they do not make use of the good mapping property, we focus instead in this section on the parts of the construction which require the good mapping property. To begin, we define the real topological vector space
\begin{equation}\label{Ginf_star_def}
\G_{\infty}^{*} \coloneqq \{\Gamma = (\gamma^{(k)})_{k \in \N} \in\prod_{k=1}^{\infty} \L(\Sc_{s}'(\R^{k}),\Sc_{s}(\R^{k})) : \gamma^{(k)} = (\gamma^{(k)})^{*} \enspace \forall k\in\N\},
\end{equation}
endowed with the product topology.\footnote{We remark that $\G_\infty^*$ is the projective limit of the spaces $\{\G_N^*\}_{N\in\N}$ directed with respect to reverse inclusion.} Analogous to \cref{lem:dual_GN}, it holds that $\G_\infty^*$ is isomorphic to the dual of $(\G_\infty)^*$.

\begin{lemma}[Dual of $\G_{\infty}$]\label{lem:G_inf_dual}
The topological dual of $\G_{\infty}$, denoted by $(\G_{\infty})^{*}$ and endowed with the strong dual topology, is isomorphic to $\G_{\infty}^{*}$.
\end{lemma}

We now need to established the existence of a Poisson structure on $\G_{\infty}^{*}$. We start by specifying a unital sub-algebra of $C^\infty(\G_\infty^*; \R)$.

\begin{mydef}\label{a_inf_def}
Let $\A_{\infty}$ be the algebra with respect to point-wise product generated by functionals in
\begin{equation}
\begin{split}
&\{F\in C^{\infty}(\G_{\infty}^{*};\R) : F(\cdot) = i\Tr( A\cdot),\,\,  A\in\G_{\infty}\}  \cup \{F\in C^{\infty}(\G_{\infty}^{*};\R) : F(\cdot) \equiv C\in\R\}.
\end{split}
\end{equation}
\end{mydef}
In other words, $\A_{\infty}$ is the algebra (under point-wise product) generated by constants and the image of $\G_{\infty}$ under the canonical embedding into $(\G_{\infty}^*)^{*}$. We note that our previous remarks \cref{rem:AH_can}, \cref{rem:AH_struc}, \cref{rem:con_GD} carry over with $\A_{H,N}$ replaced by $\A_{\infty}$.

\medskip
We now wish to define the Lie-Poisson bracket $\pb{\cdot}{\cdot}_{\G_{\infty}^{*}}$ on $\A_{\infty}\times \A_{\infty}$ using the Lie bracket constructed in \cref{ssec:geo_LA}. In order to so, we first need an identification of continuous linear functionals as skew-adjoint operators, which follows from \cref{lem:dual_GN*}.

\begin{lemma}[Dual of $\G_{\infty}^{*}$]\label{dual_dual}
The topological dual of $\G_{\infty}^{*}$, denoted by $(\G_{\infty}^{*})^{*}$ and endowed with the strong dual topology, is isomorphic to
\begin{equation}
\widetilde{\G}_{\infty}\coloneqq \{A\in \bigoplus_{k=1}^{\infty}\L(\Sc_{s}(\R^{k}), \Sc_{s}'(\R^{k})) : (A^{(k)})^{*} = -A^{(k)}\},
\end{equation}
equipped with the subspace topology induced by $\bigoplus_{k=1}^{\infty}\L(\Sc_{s}(\R^{k}),\Sc_{s}'(\R^{k}))$, via the canonical bilinear form
\begin{equation}
i\Tr(A\cdot \Gamma) = i\sum_{k=1}^{\infty}\Tr_{1,\ldots,k}(A^{(k)}\gamma^{(k)}), \qquad \Gamma=(\gamma^{(k)})_{k\in\N}\in \G_{\infty}^{*}.
\end{equation}
\end{lemma}

\begin{remark}
The previous lemma implies that, given a smooth real-valued functional $F:\G_{\infty}^{*}\rightarrow \R$ and a point $\Gamma\in\G_{\infty}^{*}$, we may identify the continuous linear functional $dF[\Gamma]$, given by the G\^ateaux derivative of $F$ at $\Gamma$, as a skew-adjoint element of $\bigoplus_{k=1}^{\infty}\L(\mathcal{S}_{s}(\mathbb{R}^{k}), \Sc_{s}'(\R^{k}))$. We will abuse notation by denoting this element by $dF[\Gamma]=(dF[\Gamma]^{(k)})_{k\in\N}$.
\end{remark}

We are now prepared to introduce the Lie-Poisson bracket $\pb{\cdot}{\cdot}_{\G_{\infty}^{*}}$ on $\A_{\infty}\times \A_{\infty}$.

\begin{mydef}
\label{def:Ginf_pb}
For $F, G \in \A_\infty$, we define
\begin{equation}\label{equ:poisson_def}
\pb{F}{G}_{\G_{\infty}^{*}}(\Gamma)  \coloneqq i\Tr\paren*{\comm{dF[\Gamma]}{dG[\Gamma]}_{\G_{\infty}}\cdot\Gamma}, \qquad \forall \Gamma\in \G_{\infty}^{*}.
\end{equation}
\end{mydef}

\begin{remark}[Existence of Casimirs]
The functional $F(\Gamma) \coloneqq \Tr_{1}(\gamma^{(1)})$ is a Casimir\footnote{i.e. it Poisson commutes with every functional in $\A_{\infty}$.} for the Poisson bracket $\pb{\cdot}{\cdot}_{\G_{\infty}^{*}}$. Consequently, the Poisson bracket $\pb{\cdot}{\cdot}_{\G_{\infty}^{*}}$ is not canonically induced by a symplectic structure on $\G_{\infty}^{*}$.
\end{remark}

We now turn to our ultimate goal of this subsection, that is, proving the following: 

\LP*

Properties \ref{item:wp_P1} and \ref{item:wp_P2} in \cref{def:WP} for weak Poisson manifolds are readily proved using the same arguments in the proofs of \cref{lem:H_LA_P1} and \cref{lem:H_LA_P2}, respectively, together with the following technical result, which in turn follows from the same argument as in \cref{lem:tr_gen_H}. We omit the details of the verification of these properties. 

\begin{lemma}\label{lem:pb_func}
Suppose that $G_{j}\in\A_{\infty}$ is a trace functional $G_{j}(\Gamma) = i\Tr(dG_{j}[0]\cdot\Gamma)$ for $j=1,2$. Then for all $\Gamma\in\G_{\infty}^{*}$, the G\^ateaux derivative $d\pb{G_{1}}{G_{2}}_{\G_{\infty}^{*}}[\Gamma]$ at the point $\Gamma$ may be identified with the element
\begin{equation}
\comm{dG_{1}[0]}{dG_{2}[0]}_{\G_{\infty}}\in \G_{\infty}
\end{equation}
via the canonical trace pairing. If $G_{1}$ is a trace functional and $G_{2}=G_{2,1}G_{2,2}$ is the product of two trace functionals in $\A_{\infty}$, then $d\pb{G_{1}}{G_{2}}_{\G_{\infty}^{*}}[\Gamma]$ may be identified with
\begin{equation}
G_{2,1}(\Gamma)\comm{dG_{1}[0]}{dG_{2,2}[0]}_{\G_{\infty}} + G_{2,2}(\Gamma)\comm{dG_{1}[0]}{dG_{2,1}[0]}_{\G_{\infty}}
\end{equation}
for all $\Gamma\in\G_{\infty}^{*}$ via the canonical trace pairing.
\end{lemma}

Property \ref{item:wp_P3} is more delicate: to show that the Hamiltonian vector field is well-defined, we have to exploit the good mapping property. Analogous to the proof of \cref{prop:G_inf_br}, rather than prove directly the well-definedness of the Hamiltonian vector field, we can use our earlier investment of work in proving \cref{lem:H_WP_VF}, which gives an explicit formula for the $N$-body vector field, together with our convergence result \cref{prop:LB_lim} and an approximation argument.

\begin{lemma}\label{lem:WP_P3}
$(\G_\infty^*, \A_{\infty}, \pb{\cdot}{\cdot}_{\G_{\infty}^{*}})$ satisfies property \ref{item:wp_P3} in \cref{def:WP}. Furthermore, if $H\in\A_{\infty}$, then we have the following formula for the Hamiltonian vector field $X_{H}$:
\begin{equation}
\label{eq:Ginf_vf}
\begin{split}
X_{H}(\Gamma)^{(\ell)} &= \sum_{j=1}^{\infty}j\Tr_{\ell+1,\ldots,\ell+j-1}\paren*{\comm{\sum_{\alpha=1}^{\ell} dH[\Gamma]^{(j)}_{(\alpha,\ell+1,\ldots,\ell+j-1)}}{\gamma^{(\ell+j-1)}} }.
\end{split}
\end{equation}
\end{lemma}

\begin{proof}
Let $F,H\in\A_{\infty}$. In order to find a candidate Hamiltonian vector field, we compute $\pb{F}{H}_{\G_{\infty}^{*}}$ using an approximation to reduce to the case where $F$ and $G$ belong to $\A_{H,N}$, for all $N$ sufficiently large, together with the $N$-hierarchy Hamiltonian vector field result \cref{lem:H_WP_VF} and our convergence result \cref{prop:LB_lim}. Once we have found a candidate, we then verify that the vector field is a smooth map $\G_{\infty}^{*}\rightarrow\G_{\infty}^{*}$, which then completes the proof by the uniqueness guaranteed by \cref{rem:hvf_u}.

By definition of $\A_\infty$, the functionals $F$ and $H$ are finite linear combinations of finite products of trace functionals generated by elements in $\G_\infty$:
\begin{equation}
F(\Gamma) = \sum_{a=1}^{M_F} C_{a,F}\prod_{b=1}^{M_{a,F}} i\Tr(A_{b,F}\cdot \Gamma), \qquad H(\Gamma) = \sum_{a=1}^{M_H}C_{a,H}\prod_{b=1}^{M_{a,H}} i\Tr(A_{b,H}\cdot\Gamma),
\end{equation}
where $M_F,M_H,M_{a,F},M_{a,H}\in\N$, $C_{a,F},C_{a,H}\in\R$, and $A_{b,F}=(A_{b,F}^{(k)})_{k\in\N}, A_{b,H}=(A_{b,H}^{(k)})_{k\in\N} \in\G_\infty$. Additionally, since $\G_\infty$ is a direct sum, there exists an integer $N_0\in\N$ such that for each $1\leq a\leq M_F$ and $1\leq b\leq M_{a,F}$, 
\begin{equation}
A_{b,F}^{(k)}= 0\in\g_{k, gmp}, \qquad \forall 1\leq k\leq N_0
\end{equation}
and similarly for $A_{b,H}^{(k)}$. So by mollifying and truncating the Schwartz kernels of each $A_{b,F}^{(k)}, A_{b,H}^{(k)}$, we obtain approximating sequences $A_{n,b,F}\coloneqq (A_{n,b,F}^{(k)})_{k\in\N}$ and $A_{n,b,H}\coloneqq (A_{n,b,H}^{(k)})_{k\in\N}$, such that 
\begin{equation}
A_{n,b,F}, A_{n,b,H} \in \G_\infty \cap \bigoplus_{k=1}^\infty \L(\Sc_s'(\R^k),\Sc_s(\R^k)),
\end{equation}
$A_{n,b,F}\rightarrow A_{b,F}$, and $A_{n,b,H}\rightarrow A_{b,H}$ in $\G_\infty$ as $n\rightarrow\infty$. In particular, each $A_{n,b,F}, A_{n,b,H}\in\G_M$ for every integer $M\geq N_0$. Now using the approximants $A_{n,b,F}$ and $A_{n,b,H}$, we can define sequences $(F_n)_{n\in\N}$ and $(H_n)_{n\in\N}$ of functionals in $\A_{\infty}$ by
\begin{equation}
F_n(\Gamma) \coloneqq \sum_{a=1}^{M_F}C_{a,F}\prod_{b=1}^{M_{a,F}}i\Tr(A_{n,b,F}\cdot\Gamma), \qquad H_n(\Gamma) \coloneqq \sum_{a=1}^{M_H}C_{a,H}\prod_{b=1}^{M_{a,H}} i\Tr(A_{n,b,H}\cdot\Gamma),
\end{equation}
such that $F_n(\Gamma)\rightarrow F(\Gamma)$ and $H_n(\Gamma)\rightarrow H(\Gamma)$ as $n\rightarrow\infty$ uniformly on bounded subsets of $\G_\infty^*$. Lastly, note that by the Leibnitz rule for the G\^{a}teaux derivative,
\begin{equation}
dF_n[\Gamma], dH_n[\Gamma]\in \G_M, \qquad \forall M\geq N_0
\end{equation}
and $dF_n[\Gamma]\rightarrow dF[\Gamma]$ and $dH_n[\Gamma]\rightarrow dH[\Gamma]$ in $\bigoplus_{k=1}^\infty \L(\Sc_s(\R^k),\Sc_s'(\R^k))$, as $n\rightarrow\infty$, uniformly on bounded subsets of $\G_\infty^*$.

Now by separate continuity of the Lie bracket $\comm{\cdot}{\cdot}_{\G_\infty}$ and the separate continuity of the generalized trace (see \cref{prop:gtr_prop}), we obtain from the definition of $\pb{\cdot}{\cdot}_{\G_\infty^*}$ that
\begin{align}
\pb{F}{H}_{\G_\infty^*}(\Gamma) &= i\Tr\paren*{\comm{dF[\Gamma]}{dH[\Gamma]}_{\G_\infty}\cdot\Gamma} \nonumber\\
&=i\lim_{n_1\rightarrow\infty}\lim_{n_2\rightarrow\infty} \Tr\paren*{\comm{dF_{n_1}[\Gamma]}{dH_{n_2}[\Gamma]}_{\G_\infty}\cdot\Gamma} \nonumber\\
&=\lim_{n_1\rightarrow\infty}\lim_{n_2\rightarrow \infty} \pb{F_{n_1}}{H_{n_2}}_{\G_\infty^*}(\Gamma),
\end{align}
for each $\Gamma\in\G_\infty^*$. Since
\begin{equation}
\label{eq:aa_zero}
dF_{n_1}[\Gamma]^{(k)} = dH_{n_2}[\Gamma]^{(k)} = 0\in\g_{k, gmp}, \qquad \forall k\geq N_0, \enspace (n_1,n_2)\in\N^2, \enspace \Gamma\in\G_\infty^*,
\end{equation}
it follows from an examination of the definition of $\comm{dF_{n_1}[\Gamma]}{dH_{n_2}[\Gamma]}_{\G_\infty}$ that
\begin{equation}
\comm{dF_{n_1}[\Gamma]}{dH_{n_2}[\Gamma]}_{\G_\infty}^{(k)} = 0\in \g_{k, gmp}, \qquad \forall k\geq 2N_0+1, \enspace (n_1,n_2)\in\N^2, \enspace \Gamma\in\G_\infty^*.
\end{equation}
Therefore, if $\Gamma=(\gamma^{(k)})_{k\in\N}\in\G_\infty^*$, then letting $\Gamma_M \coloneqq (\gamma^{(k)})_{k=1}^M$ be the projection onto an element of $\G_M^*$, for $M\geq 2N_0+1$, we see that
\begin{align}
\Tr\paren*{\comm{dF_{n_1}[\Gamma]}{dH_{n_2}[\Gamma]}_{\G_\infty}\cdot\Gamma} &= \Tr\paren*{\comm{dF_{n_1}[\Gamma]}{dH_{n_2}[\Gamma]}_{\G_\infty}\cdot\Gamma_{2N_0+1}} \nonumber\\
&=\Tr\paren*{\comm{dF_{n_1}[\Gamma_{2N_0+1}]}{dH_{n_2}[\Gamma_{2N_0+1}]}_{\G_\infty}\cdot\Gamma_{2N_0+1}}.
\end{align}
For each $(n_1,n_2)\in\N^2$, we have by \cref{prop:LB_lim} and the separate continuity of the generalized trace that
\begin{align}
\Tr\paren*{\comm{dF_{n_1}[\Gamma_{2N_0+1}]}{dH_{n_2}[\Gamma_{2N_0+1}]}_{\G_\infty}\cdot\Gamma_{2N_0+1}} &= \lim_{M\rightarrow\infty}\Tr\paren*{\comm{dF_{n_1}[\Gamma_{2N_0+1}]}{dH_{n_2}[\Gamma_{2N_0+1}]}_{\G_M}\cdot\Gamma_{2N_0+1}}.
\end{align}
For $M\gg 2_{N_0+1}$, we have by \cref{lem:H_WP_VF} that
\begin{align}
i\Tr\paren*{\comm{dF_{n_1}[\Gamma_{2N_0+1}]}{dH_{n_2}[\Gamma_{2N_0+1}]}_{\G_M}\cdot\Gamma_{2N_0+1}} &= \pb{F_{n_1}}{H_{n_2}}_{\G_M^*}(\Gamma_{2N_0+1}) \nonumber\\
&=\sum_{\ell=1}^{N_0} i\Tr_{1,\ldots,\ell}\paren*{dF_{n_1}[\Gamma_{2N_0+1}]^{(\ell)} X_{H_{n_2},\G_M^*}(\Gamma_{2N_0+1})^{(\ell)}}, \label{eq:nz_con}
\end{align}
where
\begin{equation}
\label{eq:X_Hn2}
\begin{split}
&X_{H_{n_2},\G_M^*}(\Gamma_{2N_0+1})^{(\ell)} \\
&=\sum_{j=1}^M\sum_{r=r_0}^{\min\{\ell,j\}}C_{\ell jkr M}'\Tr_{\ell+1,\ldots,k}\paren*{\comm{\sum_{\ul{\alpha}_r\in P_r^\ell} dH_{n_2}[\Gamma_{2N_0+1}]_{(\ul{\alpha}_r,\ell+1,\ldots,\min\{\ell+j-r,k\})}^{(j)}}{\gamma_{2N_0+1}^{(k)}}}
\end{split}
\end{equation}
and where
\begin{align}
k &\coloneqq \min\{\ell+j-1,M\}, \quad r_0 \coloneqq \max\{1,\min\{\ell,j\}-(M-\max\{\ell,j\})\},
\end{align}
and
\begin{align}
C_{\ell j k r M}' \coloneqq {j\choose r}\frac{M C_{\ell,M}C_{j,M}}{C_{k,M}\prod_{m=1}^{r-1}(M-k+m)}.
\end{align}
Since $dF_{n_1}[\Gamma_{2N_0+1}]^{(\ell)}=0\in\g_\ell$ and $dH_{n_2}[\Gamma_{2N_0+1}]^{(j)}=0\in\g_j$, for $\ell,j\geq N_0$, we see upon substituting the right-hand side of \eqref{eq:X_Hn2} into \eqref{eq:nz_con} that, for any $M\geq 2N_0+1$, only pairs $(\ell,j)$ satisfying $\ell+j-1\leq M$ give a nonzero contribution to the resulting expression. Similarly, only pairs $(\ell,j)$ such that $r_0=1$ give a nonzero contribution to \eqref{eq:nz_con}. Therefore, we may write
\begin{equation}
\label{eq:X_H_form}
\begin{split}
&X_{H_{n_2},\G_M^*}(\Gamma_{2N_0+1})^{(\ell)}\\
&= \sum_{j=1}^M \sum_{r=1}^{\min\{\ell,j\}} C_{\ell j k r M}'\Tr_{\ell+1,\ldots,\ell+j-1}\paren*{\comm{\sum_{\ul{\alpha}_r\in P_r^\ell} dH_{n_2}[\Gamma_{2N_0+1}]_{(\ul{\alpha}_r,\ell+1,\ldots,\ell+j-r)}^{(j)}}{\gamma_{2N_0+1}^{(\ell+j-1)}}}.
\end{split}
\end{equation}
By the analysis from the proof of \cref{prop:LB_lim}, we have that
\begin{equation}
\lim_{M\rightarrow\infty} C_{\ell j k r M}' =
\begin{cases}
j, & {r=1} \\
0, & {2\leq r\leq \min\{\ell,j\}}
\end{cases}.
\end{equation}
Since the summands in \eqref{eq:X_H_form} are zero for $j\geq N_0$, it then follows that
\begin{equation}
X_{H_{n_2},\G_M^*}(\Gamma_{2N_0+1})^{(\ell)} \xrightarrow[M\rightarrow\infty]{\g_\ell^*} \underbrace{\sum_{j=1}^\infty j\Tr_{\ell+1,\ldots,\ell+j-1}\paren*{\comm{\sum_{\alpha=1}^\ell dH_{n_2}[\Gamma_{2N_0+1}]_{(\alpha,\ell+1,\ldots,\ell+j-1)}^{(j)}}{\gamma_{2N_0+1}^{(\ell+j-1)}}}}_{\eqqcolon X_{H_{n_2},\G_\infty^*}(\Gamma_{2N_0+1})^{(\ell)}}.
\end{equation}
The preceding convergence result implies, by the separate continuity of the generalized trace, that for fixed $(n_1,n_2)\in\N^2$,
\begin{equation}
\begin{split}
&\lim_{M\rightarrow\infty} \sum_{\ell=1}^{N_0}i\Tr_{1,\ldots,\ell}\paren*{dF_{n_1}[\Gamma_{2N_0+1}]^{(\ell)} X_{H_{n_2},\G_M^*}(\Gamma_{2N_0+1})^{(\ell)}}\\
&= \sum_{\ell=1}^{N_0}i\Tr_{1,\ldots,\ell}\paren*{dF_{n_1}[\Gamma_{2N_0+1}]^{(\ell)} X_{H_{n_2},\G_\infty^*}(\Gamma_{2N_0+1})^{(\ell)}}.
\end{split}
\end{equation}
Recalling from \eqref{eq:aa_zero} that $dH_{n_2}[\Gamma_{2N_0+1}]^{(j)}=dH_{n_2}[\Gamma]^{(j)}$, for all $j\in\N$, and 
\[
\gamma_{2N_0+1}^{(\ell+j-1)}=\gamma^{(\ell+j-1)}, \quad \textup{for $\ell+j-1\leq 2N_0+1$},
\]
 by definition of the projection $\Gamma_{2N_0+1}$, we obtain that
\begin{equation}
X_{H_{n_2},\G_\infty^*}(\Gamma_{2N_0+1})^{(\ell)} = \underbrace{\sum_{j=1}^\infty j\Tr_{\ell+1,\ldots,\ell+j-1}\paren*{\comm{\sum_{\alpha=1}^\ell dH_{n_2}[\Gamma]_{(\alpha,\ell+1,\ldots,\ell+j-1)}^{(j)}}{\gamma^{(\ell+j-1)}}}}_{\eqqcolon X_{H_{n_2}}(\Gamma)^{(\ell)}},
\end{equation}
for $\ell\in\N_{\leq N_0}$. Similarly, by \eqref{eq:aa_zero}, $dF_{n_1}[\Gamma_{2N_0+1}]^{(\ell)} = dF_{n_1}[\Gamma]^{(\ell)}$, and so we have that
\begin{equation}
\sum_{\ell=1}^{N_0}i\Tr_{1,\ldots,\ell}\paren*{dF_{n_1}[\Gamma_{2N_0+1}]^{(\ell)} X_{H_{n_2},\G_\infty^*}(\Gamma_{2N_0+1})^{(\ell)}} = \sum_{\ell=1}^{N_0} i\Tr_{1,\ldots,\ell}\paren*{dF_{n_1}[\Gamma]^{(\ell)} X_{H_{n_2}}(\Gamma)^{(\ell)}}.
\end{equation}

We now proceed to the analysis of the iterative limits $n_2\rightarrow\infty$ followed by $n_1\rightarrow\infty$. Since 
\[
dH_{n_2}[\Gamma] \rightarrow dH[\Gamma]
\]
 in $\G_\infty$, as $n_2\rightarrow\infty$, it follows from \cref{prop:ext_k} and the universal property of the tensor product that the $(\ell+j-1)$-particle extensions
\begin{equation}
dH_{n_2}[\Gamma]_{(\alpha,\ell+1,\ldots,\ell+j-1)}^{(j)} \longrightarrow  dH[\Gamma]_{(\alpha,\ell+1,\ldots,\ell+j-1)}^{(j)},
\end{equation}
in $\L_{gmp}(\Sc(\R^{\ell+j-1}),\Sc'(\R^{\ell+j-1}))$ as $M \to \infty$.
for $\Gamma\in\G_\infty^*$ fixed. The continuity of the commutator bracket, the good mapping property, and the separate continuity of the generalized trace imply that
\begin{equation}
X_{H_{n_2}}(\Gamma) \longrightarrow  X_H(\Gamma).
\end{equation}
in $\prod_{k=1}^\infty \L(\Sc_s'(\R^k),\Sc_s(\R^k))$ as $n_2 \to \infty$. Moreover, the continuity of the adjoint operation (see \cref{lem:dvo_adj}) and the self-adjointness of $X_{H_{n_2}}(\Gamma)$ imply that $X_{H}(\Gamma)$ is self-adjoint, hence an element of $\G_\infty^*$. We note that writing $X_H(\Gamma)$ is a slight abuse of notation since we have not yet verified that $X_H$ satisfies all of the desired properties, but this limit, $X_H$, will be our candidate Hamiltonian vector field from the statement of the lemma.

For each $n_1\in\N$ fixed, the separate continuity of the generalized trace and the fact that $dF_{n_1}[\Gamma]^{(\ell)}=0$, for $\ell\geq N_0$, then implies
\begin{equation}
\lim_{n_2\rightarrow\infty} i\Tr\paren*{dF_{n_1}[\Gamma]\cdot X_{H_{n_2}}(\Gamma)} = i\Tr\paren*{dF_{n_1}[\Gamma]\cdot X_H(\Gamma)}.
\end{equation}
Since $dF_{n_1}[\Gamma]\rightarrow dF[\Gamma]$ in $\G_\infty$, as $n_1\rightarrow\infty$, by construction of the approximations $F_{n_1}$, another application of the separate continuity of the generalized trace yields
\begin{equation}
\lim_{n_1\rightarrow\infty} i\Tr\paren*{dF_{n_1}[\Gamma]\cdot X_H(\Gamma)} = i\Tr\paren*{dF[\Gamma]\cdot X_H(\Gamma)}.
\end{equation}

After a little bookkeeping, we have shown that for every $\Gamma\in\G_\infty^*$,
\begin{align}
\pb{F}{G}_{\G_\infty^*}(\Gamma) &= \lim_{n_1\rightarrow\infty}\lim_{n_2\rightarrow\infty}\lim_{M\rightarrow\infty} i\Tr\paren*{\comm{dF_{n_1}[\Gamma_{2N_0+1}]}{dH_{n_2}[\Gamma_{2N_0+1}]}_{\G_M}\cdot\Gamma_{2N_0+1}} \nonumber\\
&=\lim_{n_1\rightarrow\infty}\lim_{n_2\rightarrow\infty}\lim_{M\rightarrow\infty} i\Tr\paren*{dF[\Gamma_{2N_0+1}]\cdot X_{H_{n_2},\G_M}(\Gamma_{2N_0+1})} \nonumber\\
&=\lim_{n_1\rightarrow\infty}\lim_{n_2\rightarrow\infty} i\Tr\paren*{dF_{n_1}[\Gamma]\cdot X_{H_{n_2}}(\Gamma)} \nonumber\\
&=i\Tr\paren*{dF[\Gamma]\cdot X_{H}(\Gamma)}.
\end{align}

We now verify that $X_H$ is a smooth map $\G_{\infty}^{*}\rightarrow\G_{\infty}^{*}$ in order to conclude by \cref{rem:hvf_u}. It remains only to check the smoothness property. If $H$ is a trace functional, then since $dH[\Gamma]^{(j)}=dH[0]^{(j)}$ satisfies the good mapping property, the desired conclusion is immediate. The general case then follows by the Leibnitz rule for the G\^ateaux derivative, since constant functionals and trace functionals generate $\A_\infty$.
\end{proof}

\subsection{The Poisson morphism $\iota: \Sc(\R) \rightarrow \G_{\infty}^{*}$}
\label{ssec:po_emb}
We now turn to the proof of \cref{thm:pomo}. We recall that we are considering the map
\begin{equation}
\iota: \mathcal{S}(\R) \rightarrow \G_{\infty}^{*}, \qquad \iota(\phi) \coloneqq \paren*{\ket*{\phi^{\otimes k}}\bra*{\phi^{\otimes k}}}_{k\in\N},
\end{equation}
which sends a $1$-particle wave function to a density matrix $\infty$-hierarchy. We recall the definition
\begin{align*}
\A_{\Sc} = \bigl \{ H \,:\, \grad_{s}H \in C^{\infty}(\Sc(\R);\Sc(\R)) \bigr\} \subset C^{\infty}(\Sc(\R);\R).
\end{align*}
and we restate \cref{thm:pomo} here for the reader's convenience.

\pomo*

We recall that although we set $d=1$ in the proof, it works in any dimension. To prove \cref{thm:pomo}, we will need the following technical result which gives a formula for the G\^ateaux derivative of $\iota$.
\begin{lemma}[Formula for $d\iota$]\label{lem:diota}
Let $\phi,\psi\in\mathcal{S}(\R)$. Then for all $k\in\N$,
\begin{align}
\label{eq:diota}
d\iota[\phi](\psi)^{(k)} &= \sum_{m=1}^{k} \ket*{\phi^{\otimes (m-1)}\otimes \psi \otimes \phi^{\otimes (k-m)}}\bra*{\phi^{\otimes k}} + \sum_{m=1}^k \ket*{\phi^{\otimes k}}\bra*{\phi^{\otimes m-1}\otimes \psi \otimes \phi^{\otimes (k-m)}}.
\end{align}
\end{lemma}
\begin{proof}
The desired formula follows readily from the product rule.
\end{proof}

\begin{remark}
\label{rem:diota_sym}
We record here the observation that for $\phi\in\Sc(\R)$ fixed, each sum in \eqref{eq:diota} has co-domain $\L(\Sc_s'(\R^k),\Sc_s(\R^k))$. We will use this observation throughout the proof of \cref{thm:pomo} below.
\end{remark}

\begin{proof}[Proof of \cref{thm:pomo}]
Smoothness of $\iota$ follows readily from \cref{lem:diota} and induction on $k$, therefore, it remains to check that
\begin{enumerate}[(i)]
\item\label{item:alg_pback}
$\iota^{*}\A_{\infty}\subset \A_{\Sc}$,
\item\label{item:pb_pback}
$\iota^{*}\pb{\cdot}{\cdot}_{\G_{\infty}^{*}} = \pb{\iota^{*}\cdot}{\iota^{*}\cdot}_{\Sc(\R)}$.
\end{enumerate} 

We prove assertion \ref{item:alg_pback}. Let $F\in\A_{\infty}$. We need to show that $f\coloneqq F\circ \iota\in\A_{\Sc}$, that is, we need to show the symplectic $L^{2}$ gradient of $f$ exists and is a smooth $\Sc(\R)$-valued map. To this end, observe that by the chain rule, for any $\phi,\delta\phi\in\Sc(\R)$, we have
\begin{align}
df[\phi](\delta\phi) &= dF[\iota(\phi)]\paren*{d\iota[\phi](\delta\phi)} \nonumber\\
&= i\Tr\paren*{dF[\iota(\phi)] \cdot d\iota[\phi](\delta\phi)} \nonumber\\
&= i\sum_{k=1}^{\infty}\Tr_{1,\ldots,k}\paren*{dF[\iota(\phi)]^{(k)} d\iota[\phi]^{(k)}(\delta\phi)}, \label{chain}
\end{align}
where the penultimate equality follows from the identification of $dF[\iota(\phi)]$ as  an element of $\widetilde{\G_\infty}$, the bi-dual of $\G_\infty$, via the canonical trace pairing and the ultimate equality follows from the definition of the dot product. Now applying \cref{lem:diota} and the bilinearity of the generalized trace, we see that
\begin{align}
\Tr_{1,\ldots,k}\paren*{dF[\iota(\phi)]^{(k)} d\iota[\phi]^{(k)}(\delta\phi)} &= \Tr_{1,\ldots,k}\paren*{dF[\iota(\phi)]^{(k)}\paren*{\sum_{m=1}^k \ket*{\phi^{\otimes (m-1)}\otimes \delta\phi\otimes \phi^{\otimes (k-m)}}\bra*{\phi^{\otimes k}}}} \nonumber\\
&\phantom{=}+ \Tr_{1,\ldots,k}\paren*{dF[\iota(\phi)]^{(k)}\paren*{\sum_{m=1}^k\ket*{\phi^{\otimes k}}\bra*{\phi^{\otimes (m-1)}\otimes \delta\phi\otimes \phi^{\otimes (k-m)}}}} \nonumber\\
&= \ip{\phi^{\otimes k}}{dF[\iota(\phi)]^{(k)}\paren*{\sum_{m=1}^k \phi^{\otimes (m-1)}\otimes \delta\phi\otimes \phi^{\otimes (k-m)}}} \nonumber\\
&\phantom{=} + \ip{\sum_{m=1}^k \phi^{\otimes (m-1)}\otimes\delta\phi\otimes\phi^{\otimes (k-m)}}{dF[\iota(\phi)]^{(k)}\phi^{\otimes k}}, \label{expan}
\end{align}
where the ultimate equality is just applying the definition of the generalized trace. Since $dF[\iota(\phi)]^{(k)}$ is skew-adjoint, we have that
\begin{equation}
\begin{split}
&\ip{\phi^{\otimes k}}{dF[\iota(\phi)]^{(k)}\paren*{\sum_{m=1}^k \phi^{\otimes (m-1)}\otimes \delta\phi\otimes \phi^{\otimes (k-m)}}} \\
&=-\ip{dF[\iota(\phi)]^{(k)}\phi^{\otimes k}}{\sum_{m=1}^k \phi^{\otimes (m-1)}\otimes \delta\phi\otimes \phi^{\otimes (k-m)}}.
\end{split}
\end{equation}
Since $dF[\iota(\phi)]^{(k)}$ satisfies the good mapping property, the preceding expression can be written as $-\ip{\psi_{F,k}}{\delta\phi}$, where $\psi_{F,k}\in\Sc(\R)$ is the unique Schwartz function coinciding with the bosonic tempered distribution
\begin{equation}\label{equ:psi_fm}
\ip{\sum_{\alpha=1}^k (\cdot)\otimes_\alpha \phi^{\otimes (k-1)}}{dF[\iota(\phi)]^{(k)}\phi^{\otimes k}},
\end{equation}
and we recall the notation $(\cdot)\otimes_\alpha\phi^{\otimes (k-1)}$ introduced in \eqref{eq:def_ot_a}. Similarly,
\begin{equation}
\ip{\sum_{m=1}^k \phi^{\otimes (m-1)}\otimes\delta\phi\otimes\phi^{\otimes (k-m)}}{dF[\iota(\phi)]^{(k)}\phi^{\otimes k}} = \ip{\delta\phi}{\psi_{F,k}}.
\end{equation}
Therefore, we have shown that
\begin{align}
&\ip{\phi^{\otimes k}}{dF[\iota(\phi)]^{(k)}\paren*{\sum_{m=1}^k \phi^{\otimes (m-1)}\otimes \delta\phi\otimes \phi^{\otimes (k-m)}}} \nonumber\\
&\phantom{=} + \ip{\sum_{m=1}^k \phi^{\otimes (m-1)}\otimes\delta\phi\otimes\phi^{\otimes (k-m)}}{dF[\iota(\phi)]^{(k)}\phi^{\otimes k}} \nonumber\\
&=2i\Im{\ip{\delta\phi}{\psi_{F,k}}} \nonumber\\
&=i\omega_{L^2}(\delta\phi,\psi_{F,k}) \label{expan2}
\end{align}
and consequently by \cref{chain}, \cref{expan}, \cref{expan2} and bilinearity
\begin{equation}
i\sum_{k=1}^{\infty}\Tr_{1,\ldots,k}\paren*{dF[\iota(\phi)]^{(k)}d\iota[\phi]^{(k)}(\delta\phi)} = -\sum_{k=1}^{\infty} \omega_{L^2}(\delta\phi,\psi_{F,k}) = \omega_{L^2}(\psi_F,\delta\phi),
\end{equation}
where we have defined $\psi_{F}\coloneqq \sum_{k=1}^\infty \psi_{F,k}$ and used the anti-symmetry of $\omega_{L^2}$ to obtain the ultimate equality. Note that moving the summation inside the second entry of $\omega_{L^2}$ is justified by the bilinearity of the symplectic form since $dF[\iota(\phi)]^{(k)}=0$ for all but finitely many $k$, by assumption that $F\in \A_\infty$ and the generating structure of $\A_\infty$. Consequently, $\psi_{F,k}\equiv 0$ for all but finitely many $k$. We conclude that
\begin{equation}
df[\phi](\delta\phi)  = \omega_{L^2}(\psi_F,\delta\phi),
\end{equation}
and hence, recalling the definition of the symplectic $L^2$ gradient in \cref{schwartz_deriv}, we have that 
\begin{equation}
\label{eq:grad_s_f}
\grad_s f(\phi) = \psi_F\in\Sc(\R).
\end{equation}

Lastly, using the identity \eqref{eq:grad_s_f}, we prove assertion \ref{item:pb_pback}. By definition of the Hamiltonian vector field $X_G(\iota(\phi))$ in \ref{item:wp_P3} together with \cref{lem:WP_P3}, which gives a formula for $X_G(\iota(\phi))$, we have that for $F,G\in\A_{\infty}$,
\begin{align}
&\pb{F}{G}_{\G_{\infty}^{*}}(\iota(\phi)) \nonumber\\
&= dF[\iota(\phi)](X_G(\iota(\phi))) \nonumber\\
&= i\sum_{k=1}^\infty \Tr_{1,\ldots,k}\paren*{dF[\iota(\phi)]^{(k)}\sum_{j=1}^\infty j\Tr_{k+1,\ldots,k+j-1}\paren*{\comm{\sum_{\alpha=1}^k dG[\iota(\phi)]_{(\alpha,k+1,\ldots,k+j-1)}^{(j)}}{\iota(\phi)^{(k+j-1)}}}}. \label{eq:vf_expan}
\end{align}
Observe that
\begin{equation}
dG[\iota(\phi)]_{(\alpha,k+1,\ldots,k+j-1)}^{(j)}\iota(\phi)^{(k+j-1)} = \ket*{\phi^{\otimes(k-1)}\otimes^\alpha dG[\iota(\phi)]^{(j)}(\phi^{\otimes j})}\bra*{\phi^{\otimes (k+j-1)}},
\end{equation}
where $\phi^{\otimes (k-1)}\otimes^\alpha dG[\iota(\phi)]^{(j)}(\phi^{\otimes j})$ is the tempered distribution in $\Sc'(\R^{k+j-1})$ defined by
\begin{equation}
\begin{split}
&\paren*{\phi^{\otimes (k-1)}\otimes^\alpha dG[\iota(\phi)]^{(j)}(\phi^{\otimes j})}(\ux_{k+j-1})\\
&\coloneqq \phi^{\otimes (\alpha-1)}(\ux_{\alpha-1})\phi^{\otimes (k-\alpha)}(\ux_{\alpha+1;k}) dG[\iota(\phi)]^{(j)}(x_\alpha,\ux_{k+1;k+j-1}).
\end{split}
\end{equation}
Since $dG[\iota(\phi)]^{(j)}$ has the good mapping property by assumption $G\in\A_\infty$, it follows from \cref{rem:gmp} and the definition of the generalized partial trace that
\begin{equation}
\begin{split}
&\Tr_{k+1,\ldots,k+j-1}\paren*{dG[\iota(\phi)]_{(\alpha,k+1,\ldots,k+j-1)}^{(j)}\iota(\phi)^{(k+j-1)}}\\
&=\ket*{\phi^{\otimes (\alpha-1)}\otimes \psi_{G,j,\alpha}\otimes \phi^{\otimes (k-\alpha)}}\bra*{\phi^{\otimes k}},
\end{split}
\end{equation}
where $\psi_{G,j,\alpha}\in\Sc(\R)$ is the unique Schwartz function such that
\begin{equation}
\ip{\delta\phi}{\psi_{G,j,\alpha}} = \ip{\delta\phi\otimes_\alpha \phi^{\otimes (j-1)}}{dG[\iota(\phi)]^{(j)}(\phi^{\otimes j})}, \qquad \forall \delta\phi\in\Sc(\R).
\end{equation}
Moreover, since $dG[\iota(\phi)]^{(j)}(\phi^{\otimes j})\in\Sc_s'(\R^j)$, it follows from \cref{lem:bos_td} that
\begin{equation}
\ip{\delta\phi\otimes_\alpha \phi^{\otimes (j-1)}}{dG[\iota(\phi)]^{(j)}(\phi^{\otimes j})} = \ip{\delta\phi\otimes_{\alpha'} \phi^{\otimes (j-1)}}{dG[\iota(\phi)]^{(j)}(\phi^{\otimes j})},
\end{equation}
for any $1\leq \alpha,\alpha'\leq j$, and therefore $\psi_{G,j,\alpha}=\psi_{G,j,\alpha'}$. Hence,
\begin{equation}
\label{eq:dG_expan}
\begin{split}
&\Tr_{k+1,\ldots,k+j-1}\paren*{dG[\iota(\phi)]_{(\alpha,k+1,\ldots,k+j-1)}^{(j)}\iota(\phi)^{(k+j-1)}}\\
&=\frac{1}{j}\ket*{\phi^{\otimes(\alpha-1)}\otimes \psi_{G,j}\otimes \phi^{\otimes (k-\alpha)}}\bra*{\phi^{\otimes k}},
\end{split}
\end{equation}
where $\psi_{G,j}$ is defined the same as $\psi_{F,k}$ above, except with $(F,k)$ replaced by $(G,j)$.  By completely analogous reasoning together with the skew-adjointness of $dG[\iota(\phi)]^{(j)}$, we also obtain that
\begin{equation}
\label{eq:dG_expan'}
\begin{split}
&\Tr_{k+1,\ldots,k+j-1}\paren*{\iota(\phi)^{(k+j-1)}dG[\iota(\phi)]_{(\alpha,k+1,\ldots,k+j-1)}^{(j)}}\\
&=-\frac{1}{j}\ket*{\phi^{\otimes k}}\bra*{\phi^{\otimes(\alpha-1)}\otimes \psi_{G,j}\otimes \phi^{\otimes (k-\alpha)}},
\end{split}
\end{equation}
Substituting the identities \eqref{eq:dG_expan} and \eqref{eq:dG_expan'} into \eqref{eq:vf_expan}, we obtain the expression
\begin{align}
&i\sum_{k=1}^\infty \Tr_{1,\ldots,k}\Big(dF[\iota(\phi)]^{(k)}\Big(\sum_{j=1}^\infty\sum_{\alpha=1}^k \ket*{\phi^{\otimes(\alpha-1)}\otimes\psi_{G,j}\otimes\phi^{\otimes (k-\alpha)}}\bra*{\phi^{\otimes k}} \nonumber\\
&\phantom{=}\hspace{55mm}+ \ket*{\phi^{\otimes k}}\bra*{\phi^{\otimes(\alpha-1)}\otimes\psi_{G,j}\otimes\phi^{\otimes(k-\alpha)}}\Big)\Big) \nonumber\\
&=i\sum_{j=1}^\infty \sum_{k=1}^\infty \ip{\phi^{\otimes k}}{dF[\iota(\phi)]^{(k)}\paren*{\sum_{\alpha=1}^k \phi^{\otimes (\alpha-1)}\otimes\psi_{G,j}\otimes\phi^{\otimes (k-\alpha)}}} \nonumber\\
&\phantom{=}\hspace{25mm}+ \ip{\sum_{\alpha=1}^k \phi^{\otimes(\alpha-1)}\otimes\psi_{G,j}\otimes\phi^{\otimes(k-\alpha)}}{dF[\iota(\phi)]^{(k)}\phi^{\otimes k}} \nonumber\\
&=-2\sum_{j=1}^\infty \sum_{k=1}^\infty \Im{\ip{\sum_{\alpha=1}^k \phi^{\otimes(\alpha-1)}\otimes\psi_{G,j}\otimes\phi^{\otimes (k-\alpha)}}{dF[\iota(\phi)]^{(k)}\phi^{\otimes k}}} \nonumber\\
&=-2\sum_{j=1}^\infty \sum_{k=1}^\infty \Im{\ip{\psi_{G,j}}{\psi_{F,k}}},
\end{align}
where the penultimate equality follows from the skew-adjointness of $dF[\iota(\phi)]^{(k)}$ and the ultimate equality follows from the definition of $\psi_{F,k}$. Since $\psi_{F,k}=\psi_{G,j}\equiv 0$ for all but finitely many $j,k$, we are justified in writing
\begin{equation}
-2\sum_{j=1}^\infty \sum_{k=1}^\infty \Im{\ip{\psi_{G,j}}{\psi_{F,k}}} = -2\Im{\ip{\psi_G}{\psi_F}},
\end{equation}
where $\psi_F$ is defined as above and $\psi_G\coloneqq \sum_{j=1}^\infty \psi_{G,j}$ is defined completely analogously. Recalling \eqref{l2_symp} for the definition of $\omega_{L^2}$ and identity \eqref{eq:grad_s_f} for the symplectic gradient, we obtain that
\begin{equation}
-2\Im{\ip{\psi_G}{\psi_F}} = \omega_{L^2}(\grad_s f(\phi),\grad_s g(\phi)).
\end{equation}

After a little bookkeeping, we realize that we have shown that
\begin{equation}
\pb{F}{G}_{\G_\infty^*}(\iota(\phi)) = \omega_{L^2}(\grad_s f(\phi),\grad_s g(\phi)).
\end{equation} 
Since the symplectic form $\omega_{L^2}$ canonically induces the Poisson bracket $\pb{\cdot}{\cdot}_{L^2}$ through
\begin{equation}
\pb{f}{g}_{L^2}(\phi) = \omega_{L^2}(\grad_s f(\phi), \grad_s g(\phi)),
\end{equation}
the proof of assertion \ref{item:pb_pback} is complete. 
\end{proof}

\section{GP Hamiltonian flows}\label{sec:gp_flows}

In this last section, we prove \cref{thm:BBGKY_ham} and its limiting version \cref{thm:GP_ham}. 

\subsection{BBGKY Hamiltonian Flow}
For the reader's benefit, we recall that the BBGKY Hamiltonian $\H_{BBGKY,N}$ is the trace functional given by
\begin{equation}
\H_{BBGKY,N}(\Gamma_N) = \Tr(\W_{BBGKY,N}\cdot\Gamma_N),
\end{equation}
where
\begin{equation}\label{w_bbgky}
\W_{BBGKY,N} = (-\Delta_{x}, \kappa V_{N}(X_{1}-X_{2}),0,\ldots),
\end{equation}
with $\kappa$ and $V_N$ as in \eqref{eq:N_bod_ham}. We also recall here the statement of \cref{thm:BBGKY_ham}.

\BBGKYham*

We now proceed to proving \cref{thm:BBGKY_ham}. Since by \cref{lem:H_WP_VF}, we have the formula
\begin{equation}
\label{eq:BBGKY_vf}
\begin{split}
&X_{\H_{BBGKY,N}}(\Gamma_N)^{(\ell)} \\
&=  \sum_{j=1}^{N} \sum_{r=r_0}^{\min\{\ell,j\}} C_{\ell jkrN}' \Tr_{\ell+1,\ldots,k}\paren*{\comm{\sum_{\ul{\alpha}_r\in P_r^\ell} d\H_{BBGKY,N}[\Gamma_N]_{(\ul{\alpha}_r,\ell+1,\ldots,\min\{\ell+j-r,k\})}^{(j)}}{\gamma_N^{(k)}}},
\end{split}
\end{equation}
where
\begin{align}
k &\coloneqq \min\{\ell+j-1,N\}, \quad  r_0 \coloneqq \max\{1,\min\{\ell,j\}-(N-\max\{\ell,j\})\} ,
\end{align}
and 
\[
C_{\ell jkrN}' \coloneqq \frac{N C_{\ell,N}C_{j,N}}{C_{k,N}\prod_{m=1}^{r-1}(N-k+m)}{j\choose r},
\]
our task reduces to simplifying the expression in the right-hand side of \eqref{eq:BBGKY_vf}.

To this end, we first need a formula for the G\^ateaux derivative $d\H_{BBGKY,N}$ of $\H_{BBGKY,N}$ and its identification with an observable $N$-hierarchy via the canonical trace pairing.  Indeed, let $N\in\N$. Then for any $\Gamma_N = (\gamma_N^{(k)})_{k=1}^N \in \G_N^*$, we have that
\begin{equation}
d\H_{BBGKY,N}[\Gamma_N](\delta\Gamma_N) = \Tr\paren*{\W_{BBGKY,N}\cdot \delta\Gamma_N}, \qquad \forall \delta\Gamma_N \in \G_N^*.
\end{equation}
Therefore, $d\H_{BBGKY,N}[\Gamma_N] = d\H_{BBGKY,N}[0]$ is uniquely identifiable with the observable $2$-hierarchy $-i\W_{BBGKY,N}$. As a consequence, we see that
\begin{equation}
d\H_{BBGKY,N}[\Gamma_N]_{(\ul{\alpha}_r,\ell+1,\ldots,\min\{\ell+j-r,k\})}^{(j)} = 0
\end{equation}
for $3\leq j\leq N$. Therefore, by \eqref{eq:BBGKY_vf}, we have
\begin{align}
X_{\H_{BBGKY,N}}(\Gamma_N)^{(\ell)} &= -i C_{\ell 1\ell 1N}'\sum_{\alpha=1}^\ell \comm{(-\Delta_{x_1})_{(\alpha)}}{\gamma_N^{(\ell)}} \nonumber\\
&\phantom{=} \quad  -i\kappa \sum_{r=r_0}^{\min\{\ell,2\}} C_{\ell 2krN}'\sum_{\ul{\alpha}_r\in P_r^\ell} \Tr_{\ell+1,\ldots,k}\paren*{\comm{(V_{N}(X_1-X_2))_{(\ul{\alpha}_r,\ell+1,\ldots,\min\{\ell+2-r,k\})}}{\gamma_N^{(k)}}} \nonumber\\
&\eqqcolon \mathrm{Term}_{1,\ell} + \mathrm{Term}_{2,\ell}.
\end{align}

We first consider $\mathrm{Term}_{1,\ell}$. Note that $(-\Delta_x)_{(\alpha)} = -\Delta_{x_\alpha}$. Now unpacking the definition of the normalizing constant $C_{\ell 1\ell 1N}'$, we find that
\begin{align}
C_{\ell 1\ell 1N}' &= \frac{N C_{\ell,N} C_{1,N}}{C_{\ell,N}} = N C_{1,N} = 1,
\end{align}
where the ultimate equality follows from the fact that $C_{1,N} = 1/|P_{1}^N| = 1/N$. Hence,
\begin{equation}
\mathrm{Term}_{1,\ell} = -i\sum_{\alpha=1}^\ell \comm{-\Delta_{x_\alpha}}{\gamma_N^{(\ell)}}.
\end{equation}

We next consider $\mathrm{Term}_{2,\ell}$. We divide into cases based on the values of $\ell\in\{1,\ldots,N\}$.
\begin{itemize}
\item
If $\ell=1$, then
\begin{equation}
\mathrm{Term}_{2,1} = -i\kappa C_{1221N}' \Tr_{2}\paren*{\comm{(V_N(X_1-X_2)_{(1,2)}}{\gamma_N^{(2)}}},
\end{equation}
where we use that $k = 2$. Since $(V_N(X_1-X_2))_{(1,2)} = V_N(X_1-X_{2})$, it follows that
\begin{equation}
\mathrm{Term}_{2,1} = -i\kappa C_{1221N}' \Tr_{2}\paren*{\comm{V_N(X_1-X_{2})}{\gamma_N^{(2)}}}.
\end{equation}
Unpacking the definition of the constant $C_{1221N}'$, we see that
\begin{align}
C_{1221N}' &= \frac{N C_{1,N} C_{2,N}}{C_{2,N}}{2\choose 1}= 2NC_{1,N} = 2,
\end{align}
hence,
\begin{equation}
\mathrm{Term}_{2,1} = -2i\kappa \Tr_{2}\paren*{\comm{V_N(X_1-X_{2})}{\gamma_N^{(2)}}}.
\end{equation}
\item
If $2\leq \ell\leq N-1$, then
\begin{equation}
r_0 = \max\{\min\{\ell,2\} - (N-\max\{\ell,2\}),1\} = \max\{2- (N-\ell),1\}=1
\end{equation}
and therefore
\begin{equation}
\mathrm{Term}_{2,\ell} = -i\kappa\sum_{r=1}^{2} C_{\ell 2(\ell+1)rN}' \sum_{\ul{\alpha}_r \in P_r^\ell} \Tr_{\ell+1}\paren*{\comm{V_N(X_1-X_2)_{(\ul{\alpha}_r,\ell+1)}}{\gamma_N^{(\ell+1)}}},
\end{equation}
where we use that $k=\ell+1$. If $r=1$, then
\begin{align}
\sum_{\ul{\alpha}_1 \in P_1^\ell} \Tr_{\ell+1}\paren*{\comm{V_N(X_1-X_2)_{(\ul{\alpha}_1,\ell+1)}}{\gamma_N^{(\ell+1)}}} & = \sum_{\alpha=1}^\ell \Tr_{\ell+1}\paren*{\comm{V_N(X_\alpha-X_{\ell+1})}{\gamma_N^{(\ell+1)}}},
\end{align}
and recalling \eqref{eq:Ckn_def}, we have
\begin{equation}
\label{eq:coeff_comp}
C_{\ell 2(\ell+1)1N}' = \frac{N C_{\ell,N} C_{2,N}}{C_{\ell+1,N}}{2\choose 1} = \frac{2(N-\ell)}{(N-1)}. 
\end{equation}
If $r=2$, then $\min\{\ell+2-r,k\} = \ell$, which per our notation implies that
\begin{align}
\sum_{\ul{\alpha}_r \in P_r^\ell} \Tr_{\ell+1}\paren*{\comm{V_N(X_1-X_2)_{(\ul{\alpha}_r,\ell+1)}}{\gamma_N^{(\ell+1)}}} &= \sum_{(\alpha_1,\alpha_2) \in P_2^\ell} \Tr_{\ell+1}\paren*{\comm{(V_N(X_1-X_2)_{(\alpha_1,\alpha_2)}}{\gamma_N^{(\ell+1)}}}.
\end{align}
Since $\alpha_1,\alpha_2\in\N_{\leq \ell}$ and $V_N(X_1-X_2)_{(\alpha_1,\alpha_2)} = V_N(X_{\alpha_1}-X_{\alpha_2})$, we have that
\begin{equation}
\Tr_{\ell+1}\paren*{\comm{(V_N(X_1-X_2)_{(\alpha_1,\alpha_2)}}{\gamma_N^{(\ell+1)}}} = \comm{V_N(X_{\alpha_1}-X_{\alpha_2})}{\gamma_N^{(\ell)}}.
\end{equation}
Now since $k=\ell+1$, it follows from our computation in \eqref{eq:coeff_comp} that
\begin{equation}
C_{\ell 2(\ell+1)2N}' = \frac{N C_{\ell,N}C_{2,N}}{C_{\ell+1,N} (N-k+1)}{2\choose 2} = \frac{1}{N-1}.
\end{equation}
Since $V_N(X_{\alpha_1}-X_{\alpha_2}) = V_N(X_{\alpha_2}-X_{\alpha_1})$ by the evenness of the potential $V$, it follows that
\begin{equation}
\sum_{\ul{\alpha}_2\in P_2^\ell} \comm{V_N(X_{\alpha_1}-X_{\alpha_2})}{\gamma_N^{(\ell)}} = \frac{2}{N-1}\sum_{1\leq \alpha_1<\alpha_2\leq \ell} \comm{V_N(X_{\alpha_1}-X_{\alpha_2})}{\gamma_N^{(\ell)}}.
\end{equation}
After a little bookkeeping, we obtain that
\begin{equation}
\begin{split}
\mathrm{Term}_{2,\ell} &= -i\kappa\frac{2 (N-\ell)}{N-1} \sum_{\alpha=1}^\ell \Tr_{\ell+1}\paren*{\comm{V_N(X_\alpha-X_{\ell+1})}{\gamma_N^{(\ell+1)}}} \\
&\phantom{=} -i\kappa\frac{2}{N-1}\sum_{1\leq \alpha_1<\alpha_2\leq \ell} \comm{V_N(X_{\alpha_1}-X_{\alpha_2})}{\gamma_N^{(\ell)}}.
\end{split}
\end{equation}

\item
Lastly, if $\ell=N$, then
\begin{equation}
r_0 = \max\{\min\{N,2\} - (N-\max\{N,2\}),1\} = 2.
\end{equation}
Moreover, $k= N$, so that
\begin{equation}
\mathrm{Term}_{2,N} = -i\kappa C_{N2N2N}' \sum_{\ul{\alpha}_2\in P_2^N} \comm{(V_N(X_1-X_2))_{(\ul{\alpha}_2)}}{\gamma_N^{(N)}}.
\end{equation}
Since
\begin{equation}
C_{N2N2N}' = \frac{N C_{N,N} C_{2,N}}{C_{N,N}} {2\choose 2} = \frac{1}{N-1},
\end{equation}
we can again use the evenness of the potential $V$ to conclude that
\begin{equation}
\mathrm{Term}_{2,N} = -\frac{2i\kappa}{N-1} \sum_{1\leq \alpha_1<\alpha_2\leq N} \comm{V_N(X_{\alpha_1}-X_{\alpha_2})}{\gamma_N^{(N)}}.
\end{equation}
\end{itemize}

Putting our case analysis together, we obtain
\begin{equation}
X_{\mathcal{H}_{BBGKY,N}}(\Gamma_N)^{(1)} = -i\comm{-\Delta_{x_1}}{\gamma_N^{(1)}} - 2i\kappa\Tr_2\paren*{\comm{V_N(X_1-X_2)}{\gamma_N^{(2)}}},
\end{equation}
while for $2\leq \ell\leq N-1$ we have
\begin{equation}
\begin{split}
X_{\mathcal{H}_{BBGKY,N}}(\Gamma_N)^{(\ell)} &= -i\sum_{\alpha=1}^\ell \comm{-\Delta_{x_\alpha}}{\gamma_N^{(\ell)}} - \frac{2i\kappa}{N-1}\sum_{1\leq \alpha_1<\alpha_2\leq \ell} \comm{V_N(X_{\alpha_1}-X_{\alpha_2})}{\gamma_N^{(\ell)}} \\
&\hspace{14mm}\phantom{=} -\frac{2i\kappa(N-\ell)}{N-1}\sum_{\alpha=1}^\ell \Tr_{\ell+1}\paren*{\comm{V_N(X_\alpha-X_{\ell+1})}{\gamma_N^{(\ell+1)}}},
\end{split}
\end{equation}
and finally
\begin{equation}
X_{\mathcal{H}_{BBGKY,N}}(\Gamma_N)^{(N)} = -i\sum_{\alpha=1}^N \comm{-\Delta_{x_\alpha}}{\gamma_N^{(\ell)}} -\frac{2i\kappa}{N-1}\sum_{1\leq \alpha_1<\alpha_2\leq N} \comm{V_N(X_{\alpha_1}-X_{\alpha_2})}{\gamma_N^{(N)}},
\end{equation}
which we see, upon comparison with \eqref{eq:BBGKY}, are precisely the equations for solutions to the BBGKY hierarchy, thus completing the proof.

\subsection{GP Hamiltonian Flow}\label{ssec:gp_flows_prf}

In this subsection, we prove \cref{thm:GP_ham}. For the reader's benefit, we recall that the GP Hamiltonian $\H_{GP}$ is the trace functional given by
\begin{equation}
\label{eq:GP_recall}
\H_{GP}(\Gamma) \coloneqq \Tr\paren*{\W_{GP}\cdot\Gamma}, \enspace\Gamma\in\G_\infty^*; \qquad \W_{GP}=(-\Delta_x,\kappa\delta(X_1-X_2), 0,\ldots).
\end{equation}
We recall the statement of the theorem.

\GPham*

The proof is similar to the proof that the BBGKY hierarchy is a Hamiltonian equation of motion, and \cref{thm:GP_ham} may be viewed as the $N\rightarrow\infty$ limit of \cref{thm:BBGKY_ham}. In our companion work \cite{MNPRS2_2019}, we will obtain \cref{thm:GP_ham} for the 1D cubic GP hierarchy as part of a more general theorem which connects the Hamiltonian structure of an infinte coupled system of linear equations, which we call the $n$-th GP hierarchy, to the Hamiltonian structure of the $n$-th equation of the nonlinear Schr\"odinger hierarchy, which is of fundamental interest in the study of the NLS as an integrable system (see, for instance, the survey of Palais \cite{Palais1997}). The GP hierarchy under consideration here then corresponds to the $n=3$ equation of the aforementioned family of equations.

We now proceed to proving \cref{thm:GP_ham}. Recalling equation \eqref{eq:GP} for the GP hierarchy, we need to show that
\begin{equation}
\label{eq:GP_ham_goal}
X_{\H_{GP}}(\Gamma)^{(k)} = -i\paren*{\comm{-\Delta_{\ux_k}}{\gamma^{(k)}} + 2\kappa B_{k+1}\gamma^{(k+1)}}, \qquad k\in\N,
\end{equation}
for any $\Gamma=(\gamma^{(k)})\in \G_\infty^*$, which we do by direct computation.

Let $\Gamma\in\G_\infty^*$. By application of \cref{lem:WP_P3} to $\H_{GP}$ together with the  identification
\begin{equation}
d\H_{GP}[\Gamma] = -i\W_{GP},
\end{equation}
which is immediate from the fact that $\H_{GP}$ is a trace functional, we know that
\begin{equation}
\begin{split}
X_{\H_{GP}}(\Gamma)^{(k)} &=  \sum_{j=1}^{\infty}j\Tr_{k+1,\ldots,k+j-1}\paren*{\comm{\sum_{\alpha=1}^k d\H_{GP}[\Gamma]^{(j)}_{(\alpha,k+1,\ldots,k+j-1)}}{\gamma^{(k+j-1)}} }.
\end{split}
\end{equation}
Since $-i\W_{GP}^{(j)} = 0\in \g_{j,gmp}$, for $j\geq 3$, we see from \eqref{eq:GP_recall} that the formula for $X_{\H_{GP}}(\Gamma)$ simplifies to
\begin{equation}
\begin{split}
X_{\H_{GP}}(\Gamma)^{(k)} &= -i\sum_{\alpha=1}^k \paren*{(-\Delta_{x_1})_{(\alpha)}\gamma^{(k)}  - \gamma^{(k)}(-\Delta_{x_1})_{(\alpha)}} \\
&\phantom{=} -i2\kappa \sum_{\alpha=1}^k\Tr_{k+1}\paren*{\delta(X_1-X_2)_{(\alpha,k+1)}\gamma^{(k+1)}} - \Tr_{k+1}\paren*{\gamma^{(k+1)}\delta(X_1-X_2)_{(\alpha,k+1)}},
\end{split}
\end{equation}
for $k\in\N$.

Since $(-\Delta_{x_1})_{(\alpha)} = -\Delta_{x_\alpha}$ and $\Delta_{\ux_k} = \sum_{\alpha=1}^k\Delta_{x_\alpha}$ by definition, it follows that
\begin{equation}
-i\sum_{\alpha=1}^k \paren*{(-\Delta_{x_1})_{(\alpha)}\gamma^{(k)}  - \gamma^{(k)}(-\Delta_{x_1})_{(\alpha)}} = -i\comm{-\Delta_{\ux_k}}{\gamma^{(k)}}.
\end{equation}

Since $\delta(X_1-X_2)_{(\alpha,k+1)} = \delta(X_\alpha-X_{k+1})$, it follows from \cref{prop:partial_trace} for the generalized partial trace that $\Tr_{k+1}(\delta(X_\alpha-X_{k+1})\gamma^{(k+1)})$ is the element of $\L(\Sc_s'(\R^k),\Sc(\R^k))$ with Schwartz kernel
\begin{equation}
\int_{\R}dx_{k+1}\delta(x_\alpha-x_{k+1})\gamma^{(k+1)}(\ux_{k+1};\ux_k',x_{k+1}) = \gamma^{(k+1)}(\ux_k,x_\alpha;\ux_k',x_\alpha) = B_{\alpha;k+1}^+\gamma^{(k+1)}(\ux_k;\ux_k').
\end{equation}
Similarly, $\Tr_{k+1}(\gamma^{(k+1)}\delta(X_\alpha-X_{k+1}))$ is the operator with Schwartz kernel
\begin{equation}
\int_{\R}dx_{k+1}' \delta(x_\alpha'-x_{k+1})\gamma^{(k+1)}(\ux_{k},x_{k+1}';\ux_{k+1}') = \gamma^{(k+1)}(\ux_k,x_\alpha';\ux_k',x_\alpha') = B_{\alpha;k+1}^-\gamma^{(k+1)}(\ux_k;\ux_k').
\end{equation}
Since $B_{k+1}=\sum_{\alpha=1}^k B_{\alpha;k+1}^+ - B_{\alpha;k+1}^-$ by definition, we conclude that
\begin{equation}
\begin{split}
&-2\kappa i \sum_{\alpha=1}^k\Tr_{k+1}\paren*{\delta(X_1-X_2)_{(\alpha,k+1)}\gamma^{(k+1)}} - \Tr_{k+1}\paren*{\gamma^{(k+1)}\delta(X_1-X_2)_{(\alpha,k+1)}} \\
&=-2\kappa iB_{k+1}\gamma^{(k+1)}.
\end{split}
\end{equation}

After a little bookkeeping, we see that we have shown \eqref{eq:GP_ham_goal}, thus completing the proof of \cref{thm:GP_ham}.
\appendix

\section{Locally convex spaces}\label{local_cvx}
\subsection{Calculus on locally convex spaces}
The following material is intended as a crash course on calculus in the setting of locally convex topological vector spaces. Since we are in general not dealing with Banach spaces or Banach manifolds, the usual notion of the Fr\'{e}chet derivative is not suitable for our purposes. Indeed, the prototypical example we ask the reader to keep in mind is the Schwartz space $\mathcal{S}(\R)$.

One main issue posed by this more general setting is that there are several inequivalent notions of the derivative for maps between locally convex spaces. Here, we use the definition which is typically called the G\^{a}teaux derivative, which has the property that $C^{1}$ maps are continuous,\footnote{For a notion of smoothness which allows for maps to be smooth but not continuous, we refer the reader to the monograph \cite{KMich1997}.} and hence enables us to regard the derivative of a smooth real-valued functional $f$ at a point $x\in X$, which we denote by $df[x]$, as an element of the topological dual $X^{*}$.

The following material can be found in lecture notes by Milnor \cite{Milnor1984}. Many of the definitions we record are standard, but we include them for completeness. The proofs are omitted, but can be found in \cite{Hamilton1982}.

\begin{mydef}[Topological vector space]
A real or complex \emph{topological vector space (tvs)} $X$ is a vector space over  a field $\K\in\{\R,\C\}$ with a topology $\tau$ which is Hausdorff and such that the operations of addition
\begin{equation}
+: X\times X \rightarrow X, \qquad (x,y) \mapsto x+y
\end{equation}
and scalar multiplication
\begin{equation}
\cdot: \K \times X \rightarrow X, \qquad (\lambda,x) \mapsto \lambda x
\end{equation}
are continuous (the domains are equipped with the product topology).
\end{mydef}

\begin{mydef}[Locally convex space]
A tvs $X$ is said to be \emph{locally convex} if every neighborhood $U\ni 0$ contains a neighborhood $U'\ni 0$ which is convex.
\end{mydef}

A particularly nice consequence of local convexity is the following Hahn-Banach type result.
\begin{prop}[Hahn-Banach]\label{prop:HB}
If $X$ is locally convex, then given two distinct vectors $x,y\in X$, there exists a continuous $\K$-linear map $\ell: X\rightarrow \K$ with $\ell(x)\neq \ell(y)$.
\end{prop}

\begin{mydef}[G\^{a}teaux derivative]\label{gateaux_deriv}
Let $X$ and $Y$ be locally convex $\R$-tvs, let $X_{0}\subset X$ and $Y_{0}\subset Y$ be open sets, and let $f: X_{0}\rightarrow Y_{0}$ be a continuous map. Given a point $x\in X_{0}$ and a direction $v\in X$, we define the \emph{directional derivative} or \emph{G\^{a}teaux derivative} of $f$ at $x$ in the direction $v$ to be the vector
\begin{equation}
f'(x;v) \eqqcolon f_{x}'(v) \coloneqq \lim_{t\rightarrow 0} \frac{f(x+tv) - f(x)}{t},
\end{equation}
if this limit exists. We call the map $f_{x}': X\rightarrow Y$ the \emph{derivative of $f$ at the point $x$}. We use the notation $df[x](v) \coloneqq f'(x;v)$.
\end{mydef}

\begin{mydef}[$C^{1}$ G\^{a}teaux map]
Let $X_{0},Y_{0}$, and $f$ be as above. The map $f: X_{0}\rightarrow Y_{0}$ is $C^{1}$ if $f'(x;v)$ exists for all $x\in X_{0},v \in X$ and is continuous as a map
\begin{equation}
f': X_{0}\times X\rightarrow Y,
\end{equation}
where the domain is equipped with the product topology.
\end{mydef}

The G\^{a}teaux derivative $f_{x}'$ of a map $f$ between two locally convex spaces may fail to be linear in the direction $v$. However, $C^{1}$ smoothness is enough to ensure linearity in the direction variable. We always work with $C^{\infty}$ functionals (see \cref{def:smooth}), so the requisite $C^1$ smoothness is not problematic for our purposes.

\begin{prop}[Linearity of derivative]
If $f$ is $C^{1}$, then for all $x_{0}$ fixed, the map
\begin{equation}
X\rightarrow Y, \qquad v\mapsto f'(x_{0};v)
\end{equation}
is linear.
\end{prop}

Having defined the derivative and $C^{1}$ regularity, we can inductively define higher-order derivatives and regularity.

\begin{mydef}[Higher derivatives]\label{def:smooth}
The map $f:X_{0}\rightarrow Y_{0}$ is \emph{$C^{2}$ G\^{a}teaux} if $f$ is a $C^{1}$ G\^{a}teaux map and for each $v_{1}\in X$ fixed, the map
\begin{equation}
X_{0} \rightarrow Y, \qquad x\mapsto f'(x;v_{1})
\end{equation}
is $C^1$ with G\^{a}teaux derivative
\begin{equation}
\lim_{t\rightarrow 0} \frac{f'(x+tv_{2}; v_{1}) - f'(x;v_{1})}{t}
\end{equation}
depending continuously on $(x;v_{1},v_{2})\in X_{0}\times X\times X$ equipped with the product topology. If this limit exists, we call it the \emph{second G\^{a}teaux derivative} of $f$ at $x$ in the directions $v_{1},v_{2}$ and denote it by $f''(x;v_{1},v_{2})$. We inductively define \emph{$C^{r}$ maps} $X_{0}\rightarrow Y_{0}$. If a map is $C^{r}$ for every $r\in \N$, then we say that $f$ is a \emph{$C^{\infty}$ map} or alternatively, \emph{smooth map}.
\end{mydef}

\begin{prop}[Symmetry and $r$-linearity of $f_{x_{0}}^{(r)}$]
If for $r\in\N$, the map $f$ is $C^{r}$, then for each fixed $x_{0}\in X_{0}$, the map
\begin{equation}
\underbrace{X\times\cdots \times X}_{r} \rightarrow Y, \qquad (v_{1},\ldots,v_{r}) \mapsto f^{(r)}(x_{0};v_{1},\ldots,v_{r})
\end{equation}
is $r$-linear and symmetric, i.e. for any permutation $\pi\in\Ss_{r}$,
\begin{equation}
f^{(r)}(x_{0};v_{\pi(1)},\ldots,v_{\pi(r)}) = f^{(r)}(x_{0};v_{1},\ldots,v_{r}).
\end{equation}

\end{prop}

\begin{prop}[Composition]
If $f:X_{0}\rightarrow Y_{0}$ and $g: Y_{0}\rightarrow Z_{0}$ are $C^{r}$ maps, then $g\circ f: X_{0}\rightarrow Z_{0}$ is $C^{r}$ and the derivative of $(g\circ f)$ at the point $x\in X_{0}$ is the map $g_{f(x)}' \circ f_{x}': X\rightarrow Z$.
\end{prop}

\subsection{Smooth locally convex manifolds}
In this subsection, we use the calculus reviewed in the preceding subsection to introduce the basics of smooth manifolds modeled on locally convex topological vector spaces, which is needed for the construction of the Lie-Poisson manifold structure in \cref{sec:geom}. Much of the theory parallels the finite-dimensional setting, where the model space $\R^{d}$ is now replaced by an arbitrary, possibly infinite-dimensional locally convex tvs. Consequently, many of the definitions below will be familiar to the reader with a minimal knowledge of differential topology, but we record them for completeness. As in the last subsection, we closely follow \cite{Milnor1984} in our presentation.

\begin{mydef}[Smooth manifold]
A \emph{smooth manifold} modeled on a locally convex space $V$ consists of a regular, Hausdorff topological space $M$ together with a collection of homeomorphisms $\varphi_{\alpha}:V_{\alpha}\rightarrow M_{\alpha}$ satisfying the following properties:
\begin{enumerate}[(M1)]
\item\label{item:sm_p1}
$V_{\alpha}\subset V$ is open.
\item\label{item:sm_p2}
$M_{\alpha}\subset M$ is open and $\bigcup_{\alpha} M_{\alpha}=M$.
\item\label{item:sm_p3}
$\varphi_{\beta}^{-1}\circ\varphi_{\alpha}:\varphi_{\alpha}^{-1}(M_{\alpha}\cap M_{\beta}) \rightarrow \varphi_{\beta}^{-1}(M_{\alpha}\cap M_{\beta})$ is a smooth map between open subsets of $V$. We refer to the maps $\varphi_{\alpha}$ as \emph{local coordinate systems} on $M$ and the maps $\varphi_{\alpha}^{-1}$ as \emph{coordinate charts}.
\end{enumerate}
\end{mydef}

\begin{remark}
We will sometimes say that the manifold $M$ is a \emph{Fr\'{e}chet manifold} if the locally convex model space $V$ is a Fr\'{e}chet space.
\end{remark}

Using the smooth structure together with the calculus from the last subsection, we can define the notion of a smooth map between manifolds.

\begin{mydef}[Smooth map]
If $M_{1}$ and $M_{2}$ are smooth manifolds modeled on locally convex spaces $V_{1}$ and $V_{2}$, respectively, then a continuous function $f:M_{1}\rightarrow M_{2}$ is \emph{smooth} if the composition
\begin{equation}
\varphi_{\beta,2}^{-1}\circ f\circ\varphi_{\alpha,1}: \varphi_{\alpha,1}^{-1}\paren*{M_{1, \alpha}\cap f^{-1}(M_{2,\beta})} \rightarrow V_{2,\beta}
\end{equation}
is smooth whenever $f(M_{1,\alpha})\cap M_{2,\beta}\neq\emptyset$. We say that $f$ is a \emph{diffeomorphism} if it is bijective and both $f$ and $f^{-1}$ are smooth.
\end{mydef}

\begin{mydef}[Submanifold]\label{def:subm}
A subset $N$ of a smooth locally convex manifold $M$ is a \emph{submanifold} if for each $m\in N$, there exists a chart $(M_{\alpha},\varphi_{\alpha}^{-1})$ about the point $m$, such that $\varphi_{\alpha}^{-1}(M_{\alpha}\cap N) = \varphi_{\alpha}^{-1}(M_{\alpha}) \cap W$, where $W$ is a closed subspace of the space $V$ on which $M$ is modeled.
\end{mydef}

\begin{remark}
The submanifold $N$ is smooth locally convex manifold modeled on $W$. Indeed, the reader may check that the maps $\varphi_{\alpha}|_{V_{\alpha}\cap W}: V_{\alpha}\cap W \rightarrow M_{\alpha} \cap N$ are homeomorphisms which satisfy properties \ref{item:sm_p1} - \ref{item:sm_p3}.
\end{remark}

In this work, we use the kinematic definition of tangent vectors (i.e. equivalence classes of smooth curves), as opposed to the operational definition (i.e. derivations). While these two definitions are equivalent in the finite-dimensional setting, they are in general inequivalent in the infinite-dimensional setting.

\begin{mydef}[Tangent space]
Let $\varphi_{\alpha}:V_{\alpha}\rightarrow M_{\alpha}$ be a local coordinate system on $M$ with $x_{0}\in M_{\alpha}$. Let $p_{1},p_{2}:I\rightarrow M$ be smooth maps on an open interval $I\subset \R$ with $p_{i}(0)=x_{0}$ for $i=1,2$. We say that $p_{1}\sim p_{2}$ if and only if
\begin{equation}
\frac{d}{dt}\paren*{\varphi_{\alpha}^{-1}\circ p_{1}}|_{t=0} = \frac{d}{dt}\paren*{\varphi_{\alpha}^{-1}\circ p_{2}}|_{t=0}.
\end{equation}
The reader may verify that $\sim$ defines an equivalence relation on smooth curves $p:I\rightarrow M$ with $p(0)=x_{0}$. The set of all such equivalence classes is called the \emph{tangent space at $x_{0}$}, denoted by $T_{x_{0}}M$.
\end{mydef}

\begin{mydef}[Tangent bundle]
We define the \emph{tangent bundle} $TM$ as a set by 
\[
\coprod_{x\in M}T_{x}M.
\]
We define a smooth locally convex structure on $TM$ modeled on $V\times V$ by the local coordinate systems
\begin{equation}
\psi_{\alpha}: V_{\alpha}\times V \rightarrow TM_{\alpha} \subset TM,
\end{equation}
where $\psi_{\alpha}(u,v)$ is defined to be the equivalence class containing the smooth curve $t\mapsto \varphi_{\alpha}(u+tv)$ through the point $\varphi_{\alpha}(u)\in M$. The reader may verify that $\psi_{\alpha}$ maps $\{u\}\times V$ isomorphically onto the tangent space $T_{\varphi_{\alpha}(u)}M$.
\end{mydef}

\begin{mydef}[Derivative]
Let $M_{1}$ and $M_{2}$ be smooth locally convex manifolds. A smooth map $f:M_{1}\rightarrow M_{2}$ induces a continuous map
\begin{equation}
f_{x}':T_{x}M_{1} \rightarrow T_{f(x)} M_{2}, \qquad [p_{1}] \mapsto [f\circ p_{1}]
\end{equation}
called the \emph{derivative of $f$ at $x$}. Together, the maps $f_{x}'$ induce a smooth map
\begin{equation}
f_{*}:TM_{1}\rightarrow TM_{2}, \qquad (x,v) \mapsto (f(x),f_{x}'(v))
\end{equation}
which maps $T_{x}M_{1}$ linearly into $T_{f(x)}M_{2}$.
\end{mydef}

\begin{mydef}[Smooth vector field]
A \emph{smooth vector field} on $M$ is a smooth map $X:M\rightarrow TM$ such that $X(x)\in T_{x}M$. We denote the vector space of smooth vector fields on $M$ by $\mathfrak{X}(M)$.
\end{mydef}

\section{Distribution-valued operators}\label{app:DVO}
We review and develop some properties of distribution-valued operators (DVOs), that is,  elements of $\L(\Sc(\R^{k}),\Sc'(\R^{k}))$,  which are used extensively in this work. Most of these properties are a special case of a more general theory involving topological tensor products of locally convex spaces for which we refer the reader to \cite{Schwartz1966, Horvath1966, Treves1967} for further reading.

\subsection{Adjoint}
In this subsection, we record some properties of the adjoint of a DVO as well as some properties of the map taking a DVO to its adjoint. The proofs follow more or less readily from the definition and standard arguments, and are left to the reader.

\begin{lemma}[Adjoint map]\label{lem:dvo_adj}
Let $k\in\N$, and let $A^{(k)}\in \L(\Sc(\R^{k}),\Sc'(\R^{k}))$. Then there is a unique map $(A^{(k)})^{*}\in\L(\Sc(\R^{k}),\Sc'(\R^{k}))$ such that
\begin{equation}\label{eq:adj_prop}
\ipp*{(A^{(k)})^{*}g^{(k)},  \ol{f^{(k)}}}_{\Sc'(\R^{k})-\Sc(\R^{k})} = \ol{\ipp*{A^{(k)}f^{(k)}, \ol{g^{(k)}}}}_{\Sc'(\R^{k})-\Sc(\R^{k})}, \qquad \forall f^{(k)},g^{(k)}\in\Sc(\R^{k}).
\end{equation}
Furthermore, the adjoint map
\begin{equation}
*:\L(\Sc(\R^{k}),\Sc'(\R^{k})) \rightarrow \L(\Sc(\R^{k}),\Sc'(\R^{k})), \qquad A^{(k)} \mapsto (A^{(k)})^{*}
\end{equation}
is a continuous involution. 

Additionally, for $B^{(k)}\in\L(\Sc'(\R^{k}),\Sc'(\R^{k}))$, there exists a unique linear map in $(B^{(k)})^{*}\in\L(\Sc(\R^{k}),\Sc(\R^{k}))$ such that
\begin{equation}
\ipp*{u^{(k)},\ol{(B^{(k)})^{*}g^{(k)}}}_{\Sc'(\R^{k})-\Sc(\R^{k})} = \ipp*{B^{(k)}u^{(k)}, \ol{g^{(k)}}}_{\Sc'(\R^{k})-\Sc(\R^{k})}, \quad \forall (g^{(k)},u^{(k)})\in\Sc(\R^{k})\times\Sc'(\R^{k}).
\end{equation}
Moreover, the adjoint map
\begin{equation}
*:\L(\Sc'(\R^{k}),\Sc'(\R^{k})) \rightarrow \L(\Sc(\R^{k}),\Sc(\R^{k}))
\end{equation}
is a continuous involution.
\end{lemma}

The next lemma is useful for computing the adjoint of the composition of maps. We omit the proof, which is standard.

\begin{lemma}
\label{lem:adj_comp}
Let $A^{(k)}\in\L(\Sc(\R^{k}),\Sc'(\R^{k}))$ and $B^{(k)}\in\L(\Sc'(\R^{k}),\Sc'(\R^{k}))$. Then
\begin{equation}
\paren*{B^{(k)}A^{(k)}}^{*} = (A^{(k)})^{*}(B^{(k)})^{*}.
\end{equation}
\end{lemma}

\begin{mydef}[Self- and skew-adjoint]\label{def:dvo_sa}
Given $k\in\N$, we say that an operator $A^{(k)} \in \L(\Sc(\R^{k}),\Sc'(\R^{k}))$ is self-adjoint if $(A^{(k)})^{*}=A^{(k)}$. Similarly, we say that $A^{(k)}\in \L(\Sc(\R^{k}),\Sc'(\R^{k}))$ is skew-adjoint if $(A^{(k)})^{*}=-A^{(k)}$.
\end{mydef}

\begin{remark}
Note that if $A^{(k)}\in \L(\Sc(\R^{k}),\Sc'(\R^{k}))$ is an operator mapping $\Sc(\R^{k}) \rightarrow L^{2}(\R^{k})$, then our definition of self-adjoint does \emph{not} coincide with the usual Hilbert space definition for densely defined operators, but instead with the definition of a symmetric operator.
\end{remark}

\subsection{Trace and partial trace}\label{sec:trace}
In this subsection, we generalize the trace of an operator on a separable Hilbert space to the DVO setting. First, we record some remarks to motivate our definition. Since the operator $\ket*{f}\bra*{g}$, where $f,g\in L^2(\R^N)$, has trace equal to $\ip{f}{g}$, we might try to generalize the notion of trace to pure tensors of the form $f\otimes u$, where $u\in\Sc'(\R^N)$ and $f\in\Sc(\R^N)$, by defining
\begin{equation}
\label{eq:tr_pt}
\Tr_{1,\ldots,N}\paren*{f\otimes u} = \ipp{u,f}_{\Sc'(\R^N)-\Sc(\R^N)}
\end{equation}
and hope to extend this definition to $\Sc(\R^N)\hat{\otimes}\Sc'(\R^N)$ through linearity, continuity, and density. However, the evaluation map
\begin{equation}
\Sc(\R^{N})\times \Sc'(\R^{N}) \rightarrow \C, \qquad (f,u) \mapsto \ipp{u,f}_{\Sc'(\R^{N})-\Sc(\R^{N})},
\end{equation}
is not continuous, but only separately continuous, preventing us from appealing to the universal property of the tensor product to guarantee the existence of a \emph{unique} generalized trace 
\begin{equation}
\Tr_{1,\ldots,N}:\Sc(\R^{N})\hat{\otimes}\Sc'(\R^{N}) \to \C
\end{equation}
satisfying \eqref{eq:tr_pt}.

Nonetheless, by viewing the trace as a \emph{bilinear} map and using the canonical isomorphisms
\begin{equation}
\L(\Sc(\R^{N}),\Sc'(\R^{N})) \cong \Sc'(\R^{2N}) \enspace \text{and} \enspace \L(\Sc'(\R^{N}),\Sc(\R^{N})) \cong \Sc(\R^{2N}),
\end{equation}
we can uniquely define the generalized trace of the right-composition of an operator in $\L(\Sc(\R^{N}),\Sc'(\R^{N}))$ with an operator in $\L(\Sc'(\R^{N}),\Sc(\R^{N}))$ through the pairing of their Schwartz kernels. More precisely,
\begin{equation}
\label{eq:gtr_dpair}
\Tr_{1,\ldots,N}(A^{(N)}\gamma^{(N)}) = \ipp{A^{(N)}, (\gamma^{(N)})^t}_{\Sc'(\R^{2N})-\Sc(\R^{2N})}
\end{equation}
is, with an abuse of notation, the distributional pairing of the Schwartz kernel of $A^{(N)}$, which belongs to $\Sc'(\R^{2N})$, with the Schwartz kernel of the transpose of $\gamma^{(N)}$,\footnote{$(\gamma^{(N)})^t$ is the operator $f\mapsto \int_{\R^N}d\ux_N'\gamma(\ux_N';\ux_N)f(\ux_N')$.}, which belongs to $\Sc(\R^{2N})$. Equivalently, for each fixed $A^{(N)}\in \L(\Sc(\R^N),\Sc'(\R^N))$, the Schwartz kernel theorem implies the existence of a unique linear map $\L(\Sc'(\R^N),\Sc(\R^N)) \rightarrow \C$, such that
\begin{equation} \label{equ:tr}
\Tr_{1,\ldots,N}\paren*{A^{(N)}(f\otimes g)}=\ipp{A^{(N)}f,g}_{\Sc'(\R^{N})-\Sc(\R^{N})}
\end{equation}
for all $f, g \in \Sc(\R^{N})$.

\begin{mydef}[Generalized trace]\label{def:gen_trace}
We define
\begin{equation}
\begin{split}
&\Tr_{1,\ldots,N}:\L(\Sc(\R^{N}),\Sc'(\R^{N}))\times\L(\Sc'(\R^{N}),\Sc(\R^{N}))\rightarrow\C \\
&\Tr_{1,\ldots,N}\paren*{A^{(N)}\gamma^{(N)}} \coloneqq \ipp{A^{(N)}, (\gamma^{(N)})^t}_{\Sc'(\R^{2N})-\Sc(\R^{2N})}.
\end{split}
\end{equation}
\end{mydef}

\begin{remark}
The reader can check that if $A^{(N)}\in \L(\Sc(\R^N),\Sc'(\R^N))$ and $\gamma^{(N)}\in \L(\Sc'(\R^N),\Sc(\R^N))$ are such that $A^{(N)}\gamma^{(N)}$ is a trace-class operator $\rho^{(N)}$, then our definition of the generalized trace of $A^{(N)}\gamma^{(N)}$ coincides with the usual definition of the trace of $\rho^{(N)}$ as an operator on the Hilbert space $L^2(\R^N)$.
\end{remark}

We now establish some properties of the generalized trace which are reminiscent of properties of the usual trace encountered in functional analysis.

\begin{prop}[Properties of generalized trace]
\label{prop:gtr_prop}
Let $A^{(N)}\in \L(\Sc(\R^{N}),\Sc'(\R^{N}))$, and let $\gamma^{(N)}\in\L(\Sc'(\R^{N}),\Sc(\R^{N}))$. The following properties hold:
\begin{enumerate}[(i)]
\item\label{item:gtr_sc}
$\Tr_{1,\ldots,N}$ is separately continuous.
\item\label{item:gtr_adj} We have the following identity:
\begin{equation}
\Tr_{1,\ldots,N}\paren*{(A^{(N)})^{*}\gamma^{(N)}} = \ol{\Tr_{1,\ldots,N}\paren*{A^{(N)}(\gamma^{(N)})^{*}}}.
\end{equation}
\item\label{item:gtr_cyc}
If $B^{(N)}\in \L(\Sc'(\R^{N}),\Sc'(\R^{N}))$, then $\Tr_{1,\ldots,N}$ satisfies the cyclicity property
\begin{equation}
\Tr_{1,\ldots,N}\paren*{\paren*{B^{(N)}A^{(N)}}\gamma^{(N)}} = \Tr_{1,\ldots,N}\paren*{A^{(N)}\paren*{\gamma^{(N)}B^{(N)}}}.
\end{equation}
\end{enumerate}
\end{prop}
\begin{proof}
Assertion \ref{item:gtr_sc} follows from the separate continuity of the distributional pairing $\ipp{\cdot,\cdot}_{\Sc'(\R^{2N})-\Sc(\R^{2N})}$.

To prove assertion \ref{item:gtr_adj}, it suffices by density of finite linear combinations of pure tensors together with bilinearity and separate continuity of the generalized trace to consider the case where $\gamma^{(N)}=f^{(N)}\otimes g^{(N)}$, for $f^{(N)},g^{(N)}\in\Sc(\R^{N})$. By definition of the generalized trace,
\begin{equation}
\Tr_{1,\ldots,N}\paren*{(A^{(N)})^{*}(f^{(N)}\otimes g^{(N)})} = \ipp*{(A^{(N)})^{*}f^{(N)},g^{(N)}}_{\Sc'(\R^{N})-\Sc(\R^{N})},
\end{equation}
and by definition of the adjoint in \cref{lem:dvo_adj},
\begin{equation}
\ipp*{ (A^{(N)})^{*}f^{(N)}, g^{(N)}}_{\Sc'(\R^{N})-\Sc(\R^{N})} = \ol{\ipp*{A^{(N)}\ol{g^{(N)}}, \ol{f^{(N)}}}}_{\Sc'(\R^{N})-\Sc(\R^{N})}.
\end{equation}
Since $(\gamma^{(N)})^{*} = \ol{g^{(N)}}\otimes \ol{f^{(N)}}$, the desired conclusion then follows from another application of the definition of the generalized trace.

To prove assertion \ref{item:gtr_cyc}, we note that since 
\begin{equation}
B^{(N)}A^{(N)}\in\L(\Sc(\R^{N}),\Sc'(\R^{N})), \qquad  \gamma^{(N)}B^{(N)}\in\L(\Sc'(\R^{N}),\Sc(\R^{N})),
\end{equation}
all expressions are well-defined. As before, it suffices to consider the case where $\gamma^{(N)}=f^{(N)}\otimes g^{(N)}$, for $f^{(N)},g^{(N)}\in\Sc(\R^{N})$. The proof then follows readily using the involution property of the adjoint and the definition of generalized trace.
\end{proof}

We now extend the partial trace map to our setting using our bilinear perspective.

\begin{prop}[Generalized partial trace]\label{prop:partial_trace}
Let $N\in\N$ and let $k\in\{0,\ldots,N-1\}$. Then there exists a unique bilinear, separately continuous map 
\begin{equation}
\Tr_{k+1,\ldots,N}: \L(\Sc(\R^{N}),\Sc'(\R^{N})) \times \L(\Sc'(\R^{N}),\Sc(\R^{N})) \rightarrow \L(\Sc(\R^{k}),\Sc'(\R^{k})),
\end{equation}
which satisfies
\begin{equation}\label{eq:gpt_up}
\Tr_{k+1,\ldots,N}\paren*{A^{(N)}(f^{(N)}\otimes g^{(N)})} = \int_{\R^{N-k}}d\ux_{k+1;N}(A^{(N)}f^{(N)})(\ux_{k},\ux_{k+1;N})g^{(N)}(\ux_{k}',\ux_{k+1;N}).
\end{equation}
for all $A^{(N)}\in \L(\Sc(\R^{N}),\Sc'(\R^{N}))$, and $f^{(N)}, g^{(N)} \in \Sc(\R^{N})$. That is,
\begin{equation}
\begin{split}
&\ipp*{\Tr_{k+1,\ldots,N}\paren*{A^{(N)}(f^{(N)}\otimes g^{(N)})}\phi^{(k)},\psi^{(k)}}_{\Sc'(\R^k)-\Sc(\R^k)} \\
&= \ipp*{A^{(N)}f^{(N)},\psi^{(k)}\otimes \ipp{g^{(N)},\phi^{(k)}}_{\Sc_{\ux_k}'(\R^k)-\Sc_{\ux_k}(\R^k)}}_{\Sc'(\R^N)-\Sc(\R^N)},
\end{split}
\end{equation}
for all $\phi^{(k)},\psi^{(k)}\in\Sc(\R^k)$.
\end{prop}

\begin{remark}
Our notation $\Tr_{k+1, \ldots, N}$ implies a partial trace over the variables with indices belonging to the index set $\{ i \,:\, k+1 \leq i \leq N\}$. To alleviate some notational complications, we will use the convention that if the index set of the partial trace is empty, we do not take a partial trace.
\end{remark}

\begin{proof}
We first show uniqueness. Fix $N\in\N$ and $k\in\{0,\ldots,N-1\}$. Fix $A^{(N)}\in\L(\Sc(\R^{N}),\Sc'(\R^{N}))$. Suppose that there are two maps $\Tr_{k+1,\ldots,N}$ and $\widehat{\Tr}_{k+1,\ldots,N}$ satisfying \cref{eq:gpt_up}. Since every element $\gamma^{(N)}\in\L(\Sc'(\R^{N}),\Sc(\R^{N}))$ is of the form
\begin{equation}
\gamma^{(N)} = \sum_{j=1}^{\infty}\lambda_{j} f_{j}^{(k)}\otimes f_{j}^{(N-k)}\otimes g_{j}^{(k)}\otimes g_{j}^{(N-k)},
\end{equation}
where $\{\lambda_{j}\}_{j\in\N}\in\ell^{1}$ and $f_{j}^{(k)},g_{j}^{(k)}$ and $f_{j}^{(N-k)},g_{j}^{(N-k)}$ are sequences converging to zero in $\Sc(\R^{k})$ and $\Sc(\R^{N-k})$, respectively. Since the partial sums converge in $\L(\Sc'(\R^{N}), \Sc(\R^{N}))$, we have by separate continuity that
\begin{align}
\Tr_{k+1,\ldots,N}\paren*{A^{(N)}\gamma^{(N)}} &= \sum_{j=1}^{\infty}\lambda_{j} \Tr_{k+1,\ldots,N}\paren*{A^{(N)}\paren*{f_{j}^{(k)}\otimes f_{j}^{(N-k)}\otimes g_{j}^{(k)} \otimes g_{j}^{(N-k)}}} \nonumber\\
&=\sum_{j=1}^{\infty}\lambda_{j} \widehat{\Tr}_{k+1,\ldots,N}\paren*{A^{(N)}\paren*{f_{j}^{(k)}\otimes f_{j}^{(N-k)}\otimes g_{j}^{(k)} \otimes g_{j}^{(N-k)}}} \nonumber\\
&=\widehat{\Tr}_{k+1,\ldots,N}\paren*{A^{(N)}\gamma^{(N)}},
\end{align}
which completes the proof of uniqueness.

We now prove existence. Let $N,k$ and $A^{(N)}$ be fixed as above. For $f^{(k)}, g^{(k)}\in\Sc(\R^k)$ and $\gamma^{(N)}\in \L(\Sc'(\R^N),\Sc(\R^N))$, we define the integral kernel
\begin{equation}
K_{f^{(k)}, g^{(k)}, \gamma^{(N)}}(\ux_N;\ux_N') \coloneqq g^{(k)}(\ux_k')\int_{\R^k}d\ul{y}_k \gamma^{(N)}(\ux_N; \ul{y}_k, \ux_{k+1;N}') f^{(k)}(\ul{y}_k), \qquad (\ux_N,\ux_N')\in\R^{2N}.
\end{equation}
It is evident that $K_{f^{(k)}, g^{(k)}, \gamma^{(N)}} \in \Sc(\R^{2N})$. Moreover, it is straightforward to check that the trilinear map
\begin{equation}
\Sc(\R^k)\times \Sc(\R^k)\times\Sc(\R^{2N}) \rightarrow \Sc(\R^{2N}), \qquad (f^{(k)}, g^{(k)},\gamma^{(N)}) \mapsto K_{f^{(k)}, g^{(k)}, \gamma^{(N)}}
\end{equation}
is continuous, where we abuse notation by using $\gamma^{(N)}$ to denote the Schwartz kernel as well as the operator. Therefore by the Schwartz kernel theorem and the fact that $A^{(N)}\in \L(\Sc(\R^N),\Sc'(\R^N))$ by assumption, for fixed $f^{(k)}\in\Sc(\R^k)$, the map
\begin{equation}
\Sc(\R^k) \rightarrow \C, \qquad g^{(k)} \mapsto \ipp*{K_{A^{(N)}}, K_{f^{(k)}, g^{(k)}, \gamma^{(N)}}^t}_{\Sc'(\R^{2N})-\Sc(\R^{2N})}
\end{equation}
defines an element of $\Sc'(\R^k)$ and the map
\begin{equation}
\Sc(\R^k) \rightarrow \Sc'(\R^k), \qquad f^{(k)} \mapsto \ipp*{K_{A^{(N)}}, K_{f^{(k)}, \cdot, \gamma^{(N)}}^t}_{\Sc'(\R^{2N})-\Sc(\R^{2N})}
\end{equation}
is continuous. We therefore define $\Tr_{k+1,\ldots,N}(A^{(N)}\gamma^{(N)})$ to be the element of $\L(\Sc(\R^k),\Sc'(\R^k))$ given by
\begin{equation}
\ipp*{\Tr_{k+1,\ldots,N}(A^{(N)}\gamma^{(N)})f^{(k)}, g^{(k)}}_{\Sc'(\R^k)-\Sc(\R^k)} \coloneqq \ipp*{K_{A^{(N)}}, K_{f^{(k)}, g^{(k)}, \gamma^{(N)}}^t}_{\Sc'(\R^{2N})-\Sc(\R^{2N})},
\end{equation}
which is evidently bilinear in $(A^{(N)},\gamma^{(N)})$.

It remains for us to prove separate continuity. Implicit in our work in the preceding paragraph is continuity in the second entry for fixed $A^{(N)}$. Continuity in the first entry for fixed $\gamma^{(N)}\in\L(\Sc'(\R^{N}),\Sc(\R^{N}))$ then follows by duality.
\end{proof}

\subsection{Contractions and the ``good mapping property''}\label{ssec:GMP}
Given $A^{(i)}\in \L(\Sc(\R^{i}),\Sc'(\R^{i}))$, an integer $k\geq i$, and a cardinality-$i$ subset $\{\ell_{1},\ldots,\ell_{i}\} \subset \N_{\leq k}$, we want to define to an operator acting only on the variables associated to $\{\ell_{1},\ldots,\ell_{i}\}$. We have the following result.

\begin{prop}[$k$-particle extensions]
\label{prop:ext_k}
There exists a unique $A_{(\ell_{1},\ldots,\ell_{i})}^{(i)}\in \L(\Sc(\R^{k}),\Sc'(\R^{k}))$, which satisfies
\begin{equation}\label{eq:op_index_notation}
A_{(\ell_{1},\ldots,\ell_{i})}^{(i)}(f_{1}\otimes\cdots\otimes f_{k})(\ux_{k}) = A^{(i)}(f_{\ell_{1}}\otimes\cdots\otimes f_{\ell_{i}})(x_{\ell_{1}},\ldots,x_{\ell_{i}}) \cdot \biggl(\prod_{\ell\in\N_{\leq k}\setminus\{\ell_{1},\ldots,\ell_{i}\}} f_{\ell}(x_{\ell})\biggr)
\end{equation}
in the sense of tempered distributions.
\end{prop}
\begin{proof}
We first consider the case $(\ell_{1},\ldots,\ell_{i})=(1,\ldots,i)$. By the universal property of the tensor product, there exists a unique continuous linear map
\begin{equation}
A_{(1,\ldots,i)}^{(i)} \coloneqq A^{(i)}\otimes Id_{k-i}: \Sc(\R^{i}) \hat{\otimes} \Sc(\R^{k-i}) \rightarrow \Sc'(\R^{i}) \hat{\otimes} \Sc'(\R^{k-i}),
\end{equation}
satisfying
\begin{equation}
A_{(1,\ldots,i)}^{(i)}(f^{(i)}\otimes g^{(k-i)})(\ux_{k}) = A^{(i)}(f^{(i)})(\ux_{i}) g^{(k-i)}(\ux_{k-i}), \qquad \forall f\in\Sc(\R^{i}), g\in \Sc(\R^{k-i}).
\end{equation}
For the general cases where $(\ell_{1},\ldots,\ell_{i})\neq (1,\ldots,i)$, we set
\begin{equation}
A_{(\ell_{1},\ldots,\ell_{i})}^{(i)} \coloneqq \pi^{-1}\circ A_{(1,\ldots,i)}^{(i)} \circ \pi,
\end{equation}
where $\pi\in\Ss_{k}$ is any permutation such that $\pi(\ell_{j})=j$ for $j\in\N_{\leq i}$ and we let $\pi$ act on measurable functions by \eqref{eq:pi_func_def} and on distributions by duality. Let $(\ell_{1}^{*},\ldots,\ell_{k-i}^{*})$ denote the increasing ordering of the elements of the set $\N_{\leq k}\setminus\{\ell_{1},\ldots,\ell_i\}$. Then for test functions $f_1,\ldots, f_k, g_{1},\ldots,g_{k} \in\Sc(\R)$, we have
\begin{align}
&\ipp*{(\pi^{-1}\circ A_{(1,\ldots,i)}^{(i)} \circ \pi)(\bigotimes_{\ell=1}^{k}f_{\ell}), \bigotimes_{\ell=1}^{k}g_{\ell}}_{\Sc'(\R^{i})-\Sc(\R^{i})} \nonumber\\
&= \ipp*{A^{(i)}( \bigotimes_{j=1}^{i}f_{\ell_{j}}) \otimes \bigotimes_{j=1}^{k-i}f_{\ell_{j}^{*}}, (\bigotimes_{j=1}^{k}g_{j})\circ\pi}_{\Sc'(\R^{k})-\Sc(\R^{k})} \nonumber\\
&= \ipp*{A^{(i)}(\bigotimes_{j=1}^{i}f_{\ell_{j}}), \bigotimes_{j=1}^{i}g_{\ell_{j}}}_{\Sc'(\R^{i})-\Sc(\R^{i})} \cdot \ipp*{\bigotimes_{j=1}^{k-i}f_{\ell_{j}^{*}}, \bigotimes_{j=1}^{k-i} g_{\ell_{j}^{*}} }_{\Sc'(\R^{k-i})-\Sc(\R^{k-i})} \nonumber\\
&= \ipp*{A^{(i)}(\bigotimes_{j=1}^{i}f_{\ell_{j}}), \bigotimes_{j=1}^{i}g_{\ell_{j}}}_{\Sc'(\R^{i})-\Sc(\R^{i})} \cdot \prod_{j\in\N_{\leq k}\setminus \{\ell_{1},\ldots,\ell_i\}} \ipp{f_{j},g_{j}}_{\Sc'(\R)-\Sc(\R)},
\end{align}
where the penultimate equality follows from the definition of the tensor product of two distributions. By the density of finite linear combinations of pure tensors in $\Sc(\R^{k})$, it follows from the preceding equality that our definition \eqref{eq:op_coord} is independent of the choice of permutation $\pi\in\Ss_{k}$ satisfying $\pi(\ell_{j}) = j$ for every $j\in\N_{\leq i}$.
\end{proof}

An important property of the above $k$-particle extension is that it preserves self- and skew-adjointness.
\begin{lemma}\label{lem:ext_sa}
Let $i\in\N$, let $k\in\N_{\geq i}$, and let $A^{(i)}\in\L(\Sc(\R^{k}),\Sc'(\R^{i}))$ be self-adjoint (resp skew-adjoint). Then for any cardinality-$i$ subset $\{\ell_{1},\ldots,\ell_{i}\}\subset\N_{\leq k}$, we have that $A_{(\ell_{1},\ldots,\ell_{i})}^{(i)}$ is self-adjoint (resp. skew-adjoint).
\end{lemma}
\begin{proof}
Replacing $A^{(i)}$ by $iA^{(i)}$, it suffices to consider the self-adjoint case. By considerations of symmetry, it suffices to consider the case $(\ell_{1},\ldots,\ell_{i}) = (1,\ldots,i)$. The desired conclusion then follows from the fact that
\begin{align}
\ip{A_{(1,\ldots,i)}^{(i)}(f^{(i)}\otimes f^{(k-i)})}{g^{(i)}\otimes g^{(k-i)}} &= \ip{Af^{(i)}}{g^{(i)}} \ip{f^{(k-i)}}{g^{(k-i)}} \nonumber\\
&= \ip{f^{(i)}}{A^{(i)}g^{(i)}}\ip{f^{(k-i)}}{g^{(k-i)}} \nonumber\\
&= \ip{f^{(i)}\otimes f^{(k-i)}}{A_{(1,\ldots,i)}^{(i)}(g^{(i)}\otimes g^{(k-i)})},
\end{align}
for all $(f^{(i)},f^{(k-i)},g^{(i)},g^{(k-i)})\in (\Sc(\R^{i}) \times \Sc(\R^{k-i}))^{2}$, linearity, and density of linear combinations of such pure tensors in $\Sc(\R^{k})$.
\end{proof}

Now let $i,j\in\N$, let $k\coloneqq i+j-1$, and let $(\alpha,\beta)\in\N_{\leq i}\times\N_{\leq j}$. To construct a Lie bracket in \cref{ssec:geo_LA}, we need to give meaning to the composition
\begin{equation}
\label{eq:comp_wd}
A_{(1,\ldots,i)}^{(i)}B_{(i+1,\ldots,i+\beta-1,\alpha,i+\beta,\ldots,k)}^{(j)}
\end{equation}
as an operator in $\L(\Sc(\R^{k}),\Sc'(\R^{k}))$, when $A^{(i)}\in\L(\Sc(\R^{i}),\Sc'(\R^{i}))$ and $B^{(j)}\in \L(\Sc(\R^{j}),\Sc'(\R^{j}))$. 

\begin{remark}\label{not_well_def}
Without further conditions on $A^{(i)}$ or $B^{(j)}$, the composition \cref{eq:comp_wd} may not be well-defined. Indeed, consider the operator $A\in \L(\Sc(\R^{2}),\Sc'(\R^{2}))$ defined by
\begin{equation}
Af \coloneqq \delta_{0} f, \qquad \forall f\in\Sc(\R^{2}),
\end{equation}
where $\delta_{0}$ denotes the Dirac mass about the origin in $\R^{2}$. Then for $f,g\in\Sc(\R)$,
\begin{equation}
\int_{\R}dx_{2}(Af^{\otimes 2})(x_{1},x_{2})g^{\otimes 2}(x_{1}',x_{2}) = f(0)g(0) f(x_{1})g(x_{1}')\delta_{0}(x_{1}) \in\Sc'(\R)\otimes\Sc(\R).
\end{equation}
It is easy to show that $f\delta_{0}\in\Sc'(\R)$ does not coincide with a Schwartz function.
\end{remark}

This issue leads us to a property we call the \emph{good mapping property}. The intuition for the good mapping property is the basic fact from distribution theory that the convolution of a distribution of compact support with a Schwartz function is again a Schwartz function. We recall the definition of the good mapping property here.

\gmp*

\begin{remark}\label{rem:gmp}
By tensoring with identity, we see that if $A^{(i)}$ has the good mapping property, then $A_{(\ell_{1},\ldots,\ell_{i})}^{(i)}$ has the good mapping property, where $i$ is replaced by $k$ and $\alpha\in\N_{\leq k}$.
\end{remark}

\subsection{The subspace $\L_{gmp}(\Sc(\R^{k}),\Sc'(\R^{k}))$}

In this subsection, we expand more on $\L_{gmp}(\Sc(\R^{k}),\Sc'(\R^{k}))$ as a topological vector subspace of $\L(\Sc(\R^{k}),\Sc'(\R^{k}))$ and more on the identification of its topological dual. 

\begin{lemma}\label{gmp_dense}
$\L_{gmp}(\Sc(\R^{k}),\Sc'(\R^{k}))$ is a dense subspace of $\L(\Sc(\R^{k}),\Sc'(\R^{k}))$.
\end{lemma}
\begin{proof}
We first show density, beginning by recalling that $\L_{gmp}(\Sc(\R^k),\Sc'(\R^k))$ is endowed with the subspace topology induced by $\L(\Sc(\R^k),\Sc'(\R^k))$. Let $A^{(k)}\in\L(\Sc(\R^{k}),\Sc'(\R^{k}))$, and let $K_{A^{(k)}}\in\Sc'(\R^{2k})$ denote the Schwartz kernel of $A^{(k)}$. Since $\Sc(\R^{2k})$ is dense in $\Sc'(\R^{2k})$, given any bounded subset $\mathfrak{R}\subset \Sc(\R^{2k})$ and $\varepsilon>0$, there exists $K_{\mathfrak{R},\varepsilon}\in \Sc(\R^{2k})$ such that
\begin{equation}
\sup_{\tilde{K}\in \mathfrak{R}} \left|\ipp{K_{A^{(k)}}-K_{\mathfrak{R},\varepsilon},\tilde{K}}_{\Sc'(\R^{2k})-\Sc(\R^{2k})}\right| <  \varepsilon.
\end{equation}
Since the integral operator defined by the kernel $K_{\mathfrak{R},\varepsilon}$ is a continuous endomorphism of $\Sc(\R^{k})$, it belongs to $\L_{gmp}(\Sc(\R^{k}),\Sc'(\R^{k}))$. Since any bounded subset $\mathfrak{S}\subset \Sc(\R^{k})$ induces a bounded subset $\mathfrak{R} \subset \Sc(\R^{2k})$ by
\begin{equation}
\mathfrak{R} \coloneqq \mathfrak{S}\otimes\ol{\mathfrak{S}} \coloneqq \{f\otimes \bar{g} : f,g\in\mathfrak{S}\},
\end{equation}
we conclude that given any $\varepsilon>0$ and bounded subset $\mathfrak{S}\subset\Sc(\R^{k})$, there exists an element $A_{\mathfrak{S},\varepsilon}^{(k)}\in \L(\Sc'(\R^{k}),\Sc(\R^{k}))$ such that
\begin{equation}
\sup_{f,g\in\mathfrak{S}} \left|\ip{(A^{(k)}-A_{\mathfrak{S},\varepsilon}^{(k)})f}{g}\right| <\varepsilon.
\end{equation}
Since the preceding seminorms generate the topology for $\L(\Sc(\R^{k}),\Sc'(\R^{k}))$, the proof of density is complete.
\end{proof}

Using the preceding lemma, we can show that the strong dual of the subspace $\L_{gmp}(\Sc(\R^{k}),\Sc'(\R^{k}))$ is isomorphic to the space of linear operators with Schwartz-class kernels. 

\begin{lemma}\label{lem:gmp_dual}
The space $\L_{gmp}(\Sc(\R^{k}),\Sc'(\R^{k}))^{*}$ endowed with the strong dual topology is isomorphic to $\L(\Sc'(\R^{k}),\Sc(\R^{k}))$.
\end{lemma}
\begin{proof}
Since the canonical embedding $\iota:\L_{gmp}(\Sc(\R^{k}),\Sc'(\R^{k}))\rightarrow \L(\Sc(\R^{k}),\Sc'(\R^{k}))$ is tautologically continuous, the adjoint map
\begin{equation}
\iota^{*}: \L(\Sc(\R^{k}),\Sc'(\R^{k}))^{*} \rightarrow \L_{gmp}(\Sc(\R^{k}),\Sc'(\R^{k}))^{*}
\end{equation}
is continuous. Now since $\L_{gmp}(\Sc(\R^{k}),\Sc'(\R^{k}))$ is dense in $\L(\Sc(\R^{k}),\Sc'(\R^{k}))$, any linear functional
\begin{equation}
\ell \in \L_{gmp}(\Sc(\R^{k}),\Sc'(\R^{k}))^{*}
\end{equation}
extends to a unique element $\tilde{\ell} \in \L(\Sc(\R^{k}),\Sc'(\R^{k}))^{*}$ by the Hahn-Banach theorem. Hence, $\iota^{*}$ is a continuous bijection. Since the domain of the canonical isomorphism
\begin{equation}
\Phi: \L(\Sc'(\R^{k}),\Sc(\R^{k})) \rightarrow \L(\Sc(\R^{k}),\Sc'(\R^{k}))^{*}
\end{equation}
is a Fr\'{e}chet space, it follows from the open mapping theorem that $\iota^{*}\circ\Phi$ is an isomorphism.
\end{proof}

\bibliographystyle{siam}
\bibliography{GPHam}
\end{document}